\documentclass[10pt,a4paper]{report}

\usepackage{fancyhdr}
\usepackage{ifthen}
\usepackage{enumerate}
\usepackage{amssymb}
\usepackage[mathscr]{eucal}
\usepackage[errorshow]{tracefnt}
\usepackage{amsmath}
\usepackage{amssymb}
\usepackage{array}
\usepackage{amsthm}
\usepackage{eucal}
\usepackage{epsf}
\usepackage{epsfig}
\usepackage{pgf}
\usepackage[apple mac]{inputenc}
\usepackage{picins}

\usepackage{multirow}

\usepackage{bbold}
\usepackage{wrapfig}
\usepackage{slashed}	
\usepackage[small]{caption}

\RequirePackage{hyperref}

\numberwithin{equation}{section}

\def\BWworldsheet{\Sigma_g }
\def\BWmetric{{\cal M}_g}
\def\BWspacetime{{\cal E}}
\def\BWnmod{{\kappa}}

\def\BWIZ{\mathbb{Z}}
\def\BWIR{\mathbb{R}}
\def\BWIC{\mathbb{C}}
\def\BWID{\mathbb{D}}
\def\BWIRP{\mathbb{RP}}
\def\BWIMS{\mathbb{MS}}
\def\BWIKB{\mathbb{KB}}

\def\BWNmod{{N}}
\def\BWmodPara{\alpha}

\newtheorem{BWdef}{Definition}
\newtheorem{BWtheorem}{Theorem}
\newtheorem{BWex}{Example}
\newtheorem{BWexerc}{Exercise}
\newtheorem{BWprop}{Proposition}

\newcommand{\BWD}{{\cal D}}
\newcommand{\BWgauge}{{\rm Diff} \times {\rm Weyl}}
\newcommand{\BWe}{{\rm e}}
\newcommand{\BWi}{{\rm i}}
\newcommand{\BWd}{{\rm d}}

\begin{document}

 \pagestyle{empty}

\begin{center}

\vspace*{2cm}

\noindent
{\LARGE\textsf{\textbf{Lectures on Scattering Amplitudes in String Theory}}}
\vskip 2truecm


\begin{center}
{\large \textsf{\textbf{Wieland Staessens$^\spadesuit$\footnote{aspirant FWO.} and Bert Vercnocke$^\heartsuit$$^{\diamondsuit }$ }}} \\
\vskip 1truecm
        {\it $\spadesuit$  {Theoretische Natuurkunde, \\
     Vrije Universiteit Brussel \& International Solvay Institutes, \\
     VUB-campus Pleinlaan 2, B-1050, Brussel, Belgium} \\
  $\heartsuit$ {Institut de Physique Th\'{e}orique,\\
CEA/Saclay, CNRS-URA 2306,\\
Orme des Merisiers, F-91191 Gif sur Yvette, France}\\
 $\diamondsuit$ {Afdeling Theoretische Fysica,\\
     Katholieke Universiteit Leuven, \\
     Celestijnenlaan 200D bus 2415, B-3001, Heverlee, Belgium} \\
[3mm]e-mail:} {\tt Wieland.Staessens@vub.ac.be, Bert.Vercnocke@cea.fr} \\
\end{center}
\vskip .5 cm

\small{Based on lectures given by the authors at the Fifth International Modave Summer School on Mathematical Physics, held in Modave, Belgium, August 2009.}

\vskip 1 cm
\centerline{\sffamily\bfseries Abstract}
\end{center}

\noindent In these lecture notes, we take a closer look at the
calculation of scattering amplitudes for the bosonic string. It is
believed that string theories form the UV completions of
(super)gravity theories. Support for this claim can be found in
the (on-shell) scattering amplitudes of strings. On the other hand,
studying these string scattering amplitudes opens a window on the UV
behavior of the string theories themselves. In these short set of
lectures, we discuss the two-dimensional Polyakov path integral for the string, 
and its gauge symmetries, the connection to Riemann surfaces and how to 
obtain some of the simplest string scattering amplitudes. We end with some 
comments on more advanced topics. For simplicity we  limit ourselves to bosonic
open string theory in 26 dimensions.


\newpage
\chapter*{Preface}

These lecture notes form the written version of the lectures given in August 2009 at the fifth Modave Summer School in Belgium. The Modave Summer School in Mathematical Physics is a school organised by PhD students for PhD students and young postdocs. The main intention of the school is to bring together passionate young  researchers and provide them with a platform to teach each other different subjects in theoretical and mathematical physics. 

As young, enthusiastic researchers ourselves, we had the audacious plan to prepare a series of lectures about scattering amplitudes in string theory. As most of the work on scattering amplitudes in string theory has been done in the '70s and '80s, this research remains often unappreciated by and mysterious for people new to the field, despite the fact that one can extract a lot of important and interesting information (such as UV properties, analyticity, unitarity, low-energy behavior) of the string theory at hand. Moreover, the study of scattering amplitudes reveals very nice connections with the mathematical theory of Riemann surfaces, since the residual conformal symmetry on the world-sheet can be recast into a complex structure. 

With these lecture notes we hope to give a brief, yet comprehensible introduction to scattering amplitudes in string theory. The main focus of these lecture notes will be on string scattering amplitudes on genus 0 (sphere) and genus 1 (torus) surfaces. It is not our intention to give a full, consistent introduction to string theory. We therefore already presume a basic knowledge of string theory. For a complete introduction to string theory we refer to the many good text books and reviews that appeared in the last decades, for instance the Green-Schwarz-Witten volumes \cite{Green:1987sp,Green:1987mn}, Polchinski's books \cite{Polchinski:1998rq,Polchinski:1998rr}, Becker-Becker-Schwarz' recent book \cite{Becker:2007zj}  containing recent advances, and many others. These lecture notes should be considered as a modest attempt of two young and enthusiastic researchers to give an introduction to  scattering amplitudes in string theory. In a sense one can compare our efforts to a small orchestra performing a symphony. We do not claim to have rewritten Beethoven's ninth, but we give here our own interpretation of it.   \\

\vspace*{40pt}
We hope you enjoy our interpretation,\\

{\raggedleft \noindent Bert Vercnocke\\
Wieland Staessens\\
\vspace{12pt}
\noindent September 2009\\}

\pagestyle{plain}
\tableofcontents

\chapter{Introduction\label{BWc:Introduction}}

\textit{Ordinarily, in celebrated quantum field theories such as the standard model, we study point particles that have interactions dictated by a certain Lagrangian. A (quantum mechanical) theory of strings replaces point particles by tiny vibrating strings. At large length scales (low energies for the string), the string effectively looks like a point particle. The reason we do string theory, is because it has many nice features. It can incoroporate both the standard model and quantum gravity, while lacking the disturbing divergences of many other attempts at quantizing gravity. Crucial for the lack of divergences are the local symmetries of (classical) string theory. It is not a priori clear that these (classical) symmetries remain valid upon quantization. The requirement that the symmetries remain valid at the quantum level can be translated in anomaly cancellation conditions. And these conditions constrain for instance the number of dimensions in which the quantum string lives or the type of groups and group representations entering in string theory. Not only anomaly cancellation but also the chiral spectrum of the standard model impose tight constraints on the possible scenarios for a quantum gravity which incorporates the standard model. String theory offers a framework in which all these conditions can be satisfied. With these considerations in mind, we would like to study scattering amplitudes in string theory and investigate which information can be extracted out of the amplitudes. In this chapter, we give a lightning review of the building blocks of (bosonic) string theory (Lagrangian, spectrum of states, operators, S-matrix) needed to discuss amplitudes for interactions between various string states.}

\section{The String and its interactions}

\subsection{What is a String?}
A point particle traces out a one-dimensional worldline in spacetime, parametrized by a single real function, $\tau$, the particle's proper time. When a string propagates in a $D$-dimensional spacetime, it sweeps out a specific two-dimensional area. As the string is modeled as a continuous one-dimensional extended object of length $\ell$,\footnote{We shall choose $\ell$ to be $2 \pi$ without loss of generality.} we expect to parameterize it by two coordinates: $\sigma$ describing a point on the string and $\tau$ denoting the eigen-time of the string, see figure \ref{BWfig:String_Worldsheet}. 
\begin{figure}[ht!]
  \centering
%
\begin{picture}(0,0)
\put(-3,88){$\tau$}
\put(70,-5){$\sigma$}
\put(110,60){$X^\mu(\tau,\sigma)$}
\end{picture}
\includegraphics[height=0.25\textwidth]{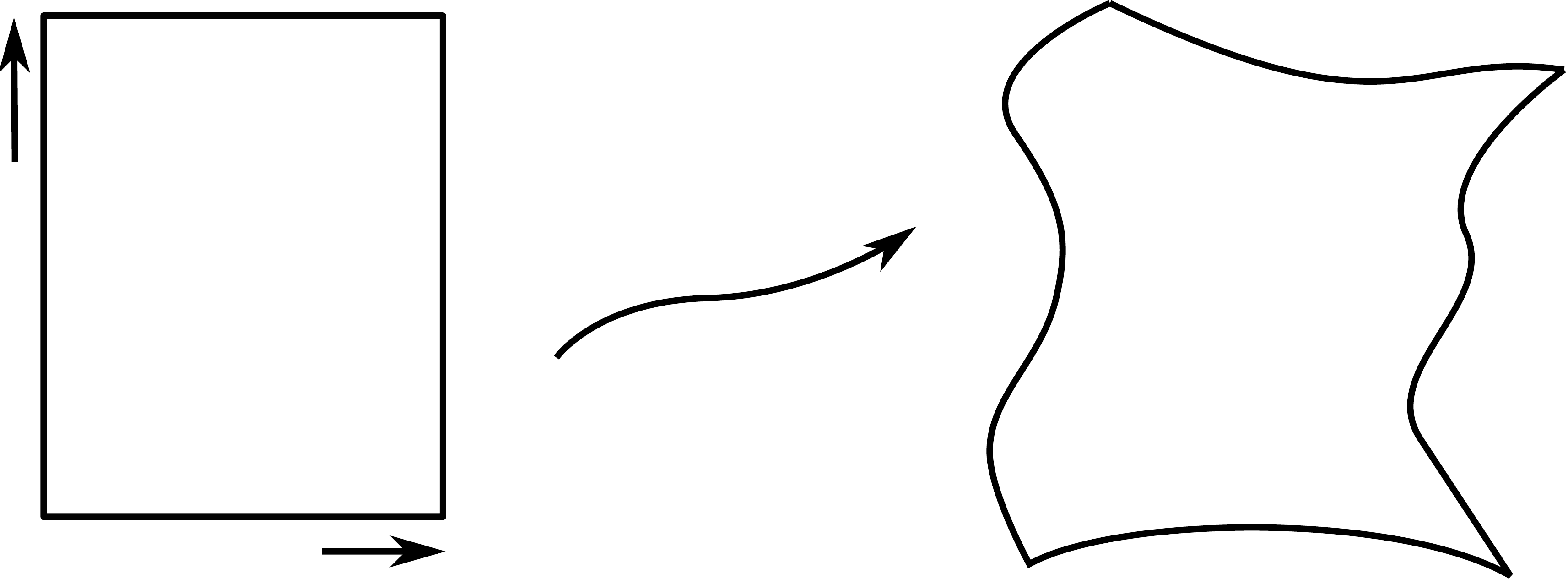}
  \caption{A string worldsheet with coordinates $(\sigma^0,\sigma^1)=(\tau,\sigma)$ as an embedding $X^\mu(\tau,\sigma)$ into $D$-dimensional spacetime.\label{BWfig:String_Worldsheet}}
\label{probebr}
\end{figure}
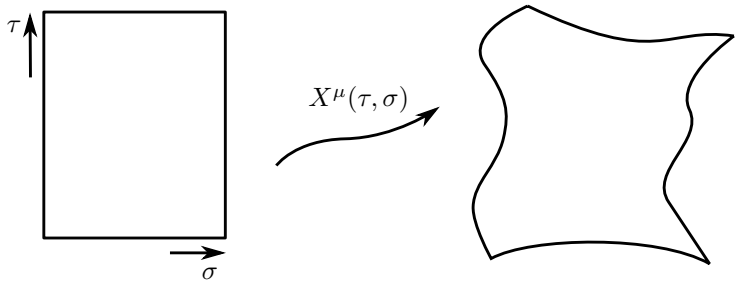
Put differently, we model the area swept out by a
string as a two-dimensional manifold and call it the string
worldsheet $\BWworldsheet$. To find the minimal area swept out by
the string, we have to extremise the following action,
\begin{eqnarray}
{\cal S} \sim \int_{\BWworldsheet} d{\cal A},
\end{eqnarray}
where $d{\cal A}$ describes an infinitesimal area. This type of
action, first proposed by Nambu and Goto, is clearly analogous to
the treatment of a point particle propagating in spacetime.\\
The Nambu-Goto-perspective is a good starting point to describe the
motion of the string, but sooner or later one stumbles upon the
difficulties of formulating its quantum version. These issues are
solved by treating the string worldsheet as a dynamical quantity
itself. To this end, we introduce general worldsheet coordinates
$(\sigma^0, \sigma^1)$ (for instance a choice of $(\sigma^0=\tau, \sigma^1=\sigma)$ with $\tau$ the eigentime and $\sigma$ a spatial coordinate along the string) and a worldsheet metric $g_{ab}$ on
$\BWworldsheet$. The worldsheet is embedded into spacetime with
spacetime metric $G_{\mu \nu} (X)$ by the embedding functions
$X^\mu(\sigma, \tau),\, \mu=1\ldots D$. Note that the components of the worldsheet metric, being symmetric, give three real functions defined on the worldsheet. These are auxiliary functions, which should not appear in the physical description of the minimal string worldsheet. More on this below.

An alternative, (classically) equivalent action  was proposed by Polyakov,
\begin{eqnarray}
{\cal S}_P = \frac{1}{4\pi \alpha'} \int_{\BWworldsheet} d^2 \sigma
 \sqrt{g} g^{ab}
\partial_a X^\mu \partial_b X^\nu G_{\mu \nu}(X),\label{BWeq:Polaction}
\end{eqnarray}
where $g = |det(g_{ab})|$. To avoid issues with boundary
contributions we limit ourselves to closed
strings\footnote{This amounts to an identification: $\sigma^1 \sim
\sigma^1 + 2\pi$}, and for simplicity we consider a flat
spacetime $G_{\mu \nu} = \eta_{\mu \nu}$.  The Polyakov action is invariant under Lorentz transformations \emph{in $D$-dimensional  spacetime}  ($X^\mu \to \Lambda^\mu{}_\nu X^\nu + a^\mu$, with $\Lambda$ satisfying $\Lambda\cdot\eta\cdot \Lambda^T = \eta$ and $a^\mu$ a constant translation), which is what we want in order for string theory to be a candidate theory describing real-world physics. The peculiar properties of this type of action lie in the invariance of the action under the following transformations \emph{on the worldsheet}:
\begin{itemize}
\item[(1)] Two-dimensional diffeomorphisms: 
\begin{equation}
\sigma^i
\rightarrow \tilde \sigma^i(\sigma)\,, \quad g_{ab} \to \frac{\partial \tilde \sigma^c}{\partial \sigma^a}\frac{\partial \tilde \sigma^d}{\partial \sigma^b} g_{cd} .
\end{equation}
\item[(2)] Weyl rescalings of the metric: 
\begin{equation}
g_{ab}(\sigma)
\rightarrow \Lambda(\sigma) g_{ab}(\sigma).
\end{equation}
\end{itemize}

The invariance of the action under these symmetries can be exploited to fix the three auxiliary functions we introduced through the worldsheet metric components $g_{ab}$. In fact, the exact way of doing this (which has several subtleties), makes up a large part of chapter \ref{BWc:Fixing}.  In short, we can choose a fixed metric to get rid of (most of) the gauge freedom (Weyl rescalings and diffeomorphisms), thereby leaving only an invariance under \emph{conformal transformations}, a combined action of the Weyl invariance introduced above with part of the diffeomorphism group, that leaves the metric invariant. 

\subsubsection{Quantization}
We would like to know what an extremum principle for the string worldsheet implies for the movements and interactions of quantized strings. As a first step, we could go to the Hamiltonian formalism and use canonical quantization to study the system \eqref{BWeq:Polaction} of the vibrations of one string quantum mechanically. In most textbooks, this is the starting point, see e.g.\ Polchinski\cite{Polchinski:1998rq}, or Green-Schwarz-Witten \cite{Green:1987sp}. One finds that the excitations of the string are characterized by energy levels of a set of harmonic oscillators (roughly one for each dimension the string can move in)\footnote{To be correct, the number of independent oscillator modes is the number of spacetime dimensions minus two. You can understand this as the string vibrates only in transverse directions (transverse to the string worldvolume), but not longitudinally.}. For the \emph{closed string}, the lowest energy levels are given in table \ref{tab:LowestModes}. 

\begin{table}[ht!]
\centering
\begin{tabular}{|lrl|}
\hline
\textbf{Energy Level}& \textbf{Excitation Mode} & \,\,\,\textbf{Name}\\
\hline
\hline
\multirow{2}{*}{$m^2=-4/\alpha'$}& \multirow{2}{*}{$|0;k\rangle$\qquad}\qquad&\multirow{2}{*}{\text{\hspace{1.5mm}Tachyon}} \\
&&\\
\hline
\multirow{4}{*}{$m^2 = 0$}& \multirow{4}{*}{$\left\{\begin{array}{r}\hat G_{(mn)} |0;k\rangle\\\hat B_{[mn]}|0;k\rangle \\\hat \phi |0;k\rangle\end{array}\right.$} &\multirow{4}{*}{$\left.\begin{array}{l}\text{Graviton}\\\text{Antisymmetric two-tensor}\\\text{Scalar}\end{array}\right.$}\\
&&\\
&&\\
&&\\
\hline
\multirow{2}{*}{$m^2=4 n/\alpha'$}&&\\
&&\\
\hline
\end{tabular}

\caption{The lowest energy modes of the bosonic string. Left, the value for the mass squared $m^2 = -k_\mu k^\mu$ ($k^\mu$ is the $D$-dimensional momentum of the string), then the lowest excitations and their common names. The excitations with $m^2=0$ correspond to a two-index tensor w.r.t.\ $D$-dimensional spacetime. The indices $m,n$ refer to the directions transverse to the string. The excitations are written in an orthogonal decomposition of the two-tensor: $G_{(mn)}$ is the symmetric, traceless part, $B_{[mn]}$ the antisymmetric part and $\phi$ the trace. In the lower line, $n\in \mathbb{N}$ gives the values of the higher mass modes.}\label{tab:LowestModes} 
\end{table}

The lowest energy mode is written down as $|0;k\rangle$, denoting that is carries spacetime momentum $(k^\mu$). It describes a string moving around in spacetime, with no internal excitation (vibrations). Think of a guitar string that is not being plucked, but gets thrown around the room. The higher energy excitations describe strings with momentum and transverse excitations.  From the $D$-dimensional Lorentz group point of view, the modes with zero mass ($m^2=0$) carry the degrees of freedom for a spin 2-particle (called graviton), a spin 1 antisymmetric field (called $B$-field) and a scalar (called dilaton). The presence of the massless spin 2-field hints that maybe there is a particle in the string spectrum mediating the gravitational force. In chapter \ref{BWc:Conclusions}, we mention how indeed bosonic string theory reduces to an ordinary $D$-dimensional gravity theory (general relativity + other fields). Note that the lowest energy mode has negative mass squared $(m^2 = -k^\mu k_\mu <0)$, meaning this is a tachyonic mode. This is taken as a hint that the bosonic string is unstable. It is actually a bad idea to perform perturbative string theory in the bosonic string picture around Minkowski spacetime, as it  cannot be a true vacuum. Going to supersymmetric string theory (``superstring theory'') can solve this problem. We do not go in further detail, but stick to bosonic string theory for simplicity.

\subsection{How does a string interact?}

We can consider other ways of studying the quantization of the string in order to get the form of stringy interactions, The key thing to note is that the interactions are already included in the worldsheet description!\footnote{Cf.\ a treatment of QFT for point particles. Normally we always study the quantization of several fields, but in principle we could also try to quantize the system of point coordinates for a (bosonic) particle (i.e.\ generalizations of the action $\int d \tau \dot \eta_{\mu\nu} x^\mu x^\nu$, describing the length of a particle worldline). See for instance \cite{Strassler:1992zr} and references therein.}  In fact, the Polyakov action \eqref{BWeq:Polaction} is all we need. See figure \ref{fig:Stringy_Interactions}.
\begin{figure}[ht!]
\centering
\includegraphics[height=.15\textheight]{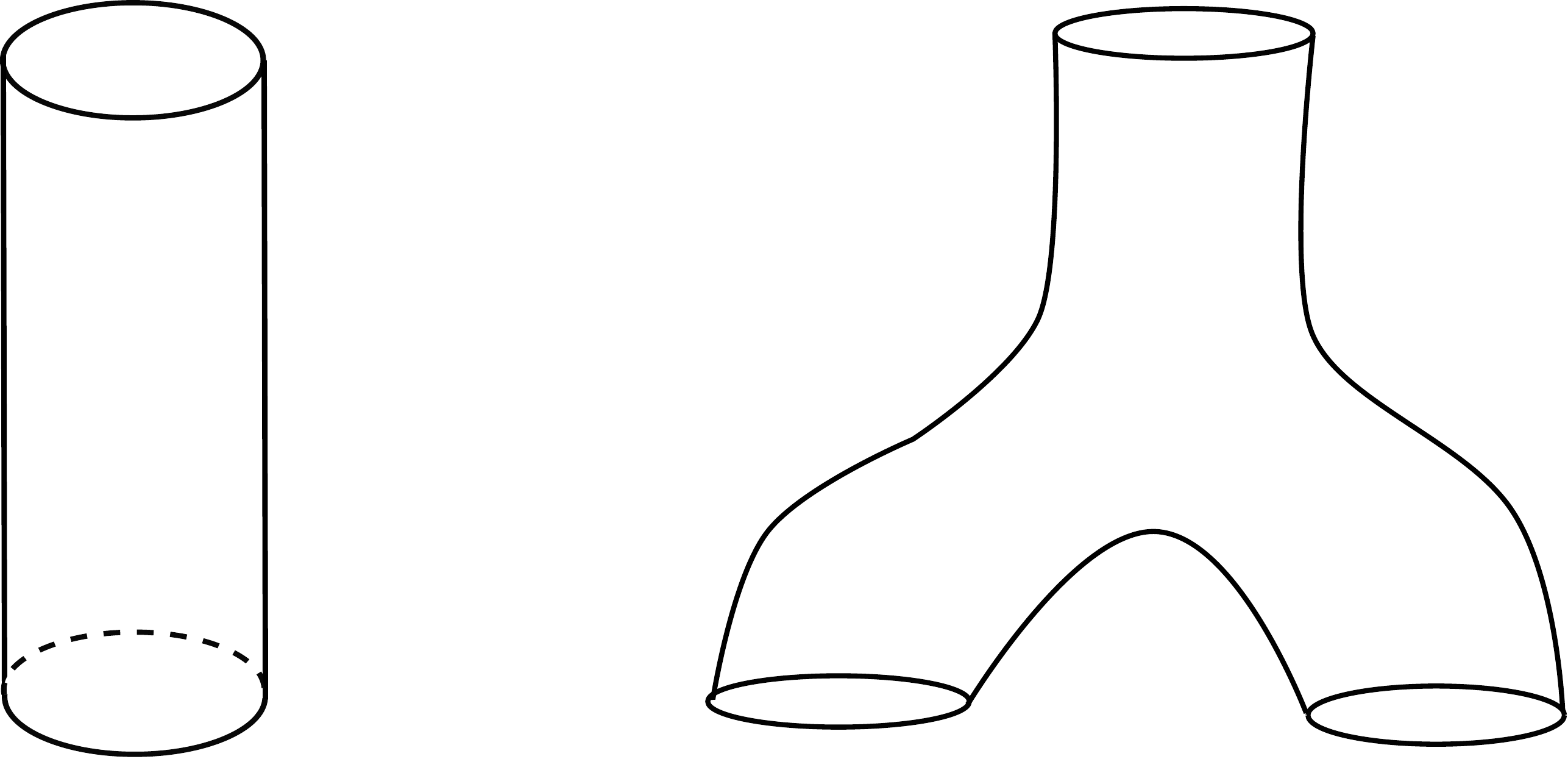}
\caption{Possible worldsheets for the closed string. Interactions are included and can be interpreted from the form of the worldsheet. The cylinder on the left can be interpreted as a closed string moving around in spacetime, while the second picture shows two strings joining to form another string. Note that we cannot localize the interaction at a certain point in a Lorentz invariant manner, see figure \ref{fig:Stringy_Interactions_2}. \label{fig:Stringy_Interactions}}
\end{figure}

\begin{figure}[ht!]
\centering
\begin{picture}(0,0)
\put(130,70){$\tau = \tau_0$}
\put(335,90){$\tau' = \tau'_0$}
\end{picture} 
\includegraphics[height=.15\textheight]{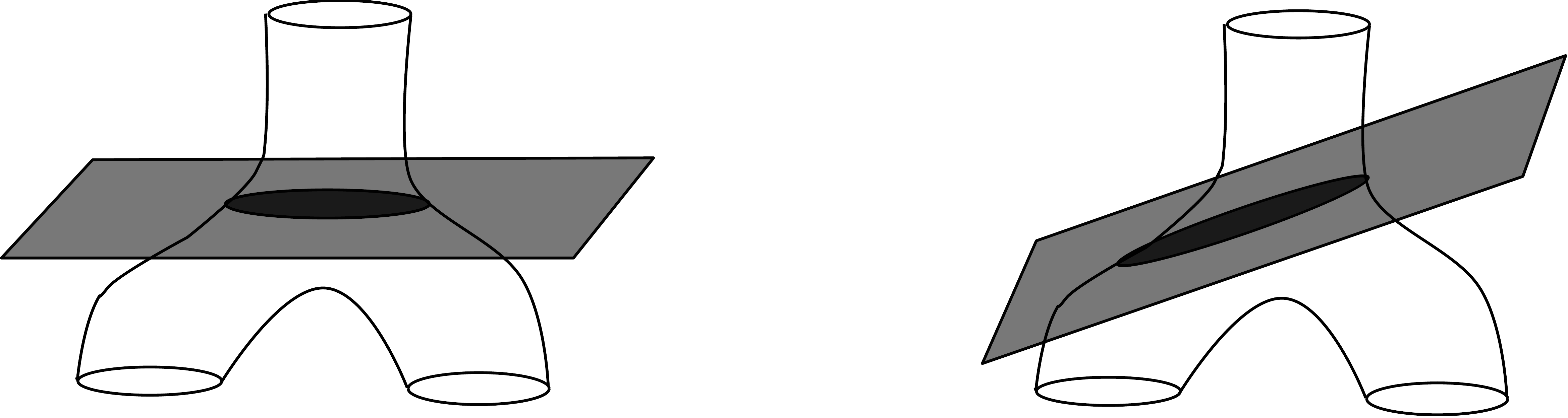}
 \caption{If we would try to localize the interaction a certain time $\tau_0$ (left), a Lorentz boost of the worldsheet coordinates $(\sigma,\tau)$ to $(\sigma',\tau')$ would disagree on the event where the interaction takes place, putting it in a different time slicing at $\tau'_0$.
\label{fig:Stringy_Interactions_2}}
\end{figure}
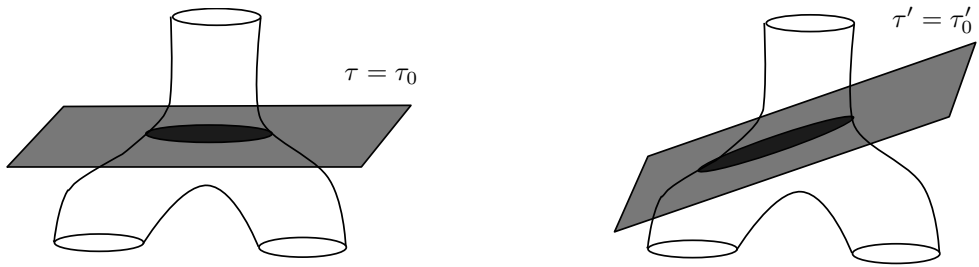

So far so good, ``the interactions are included'', let's go ahead and calculate something! But how? And what? We base the discussion on the intuitive reasoning in \cite{Tong:2009np}, section 6. At first sight, we could try to find the amplitude associated to the process where a certain ``in''-state evolves to a certain ``out''-state.  Therefore, we should specify the in and out states (see figure \ref{fig:Stringy_Interactions_EndStates}). We therefore expect a string amplitude to depend on the out and in states:
\begin{equation}
 \langle out | in\rangle = \langle \hat A_{out}(x_i) \hat A_{in}(x_j) \rangle\,,\label{eq:string_observable}
\end{equation}

where $\hat A_{out}$ and $\hat A_{in}$  are (possibly composite) operators describing the in- and out states, depending on the positions $x_i, \, i= 1\ldots n_{in},\,x_j, \, j= 1\ldots n_{out} $ of the in- and out states. (A Fourier transformation would give the momentum dependence of the amplitude as suggested by figure \ref{fig:Stringy_Interactions_EndStates}).

In ordinary QFT, we can choose the in- and out states at will. We can then calculate the amplitude between any two states we desire. (Mostly, these states are taken at times $t_{in} =-\infty$ and $t_{out}=+\infty$ and are organised in the so-called S-matrix.) However, for string theory, this is not a viable method. Strings describe, as we will later see, a theory of gravity. Such a theory is invariant under general coordinate transformations. Only observables that obey this invariance make sense (the observables should be gauge invariant). But observables as in \eqref{eq:string_observable} are not invariant under spacetime diffeomorphisms. For this reason, one takes the in- and out going states in string amplitudes off to infinity. 
\begin{wrapfigure}{r}{.3\textwidth}
\centering
\begin{picture}(0,0)
\put(5,95){$k_1$}
\put(70,100){$k_2$}
\put(5,-8){$k_3$}
\put(87,-10){$k_4$}
\end{picture}
\includegraphics[height=.15\textheight]{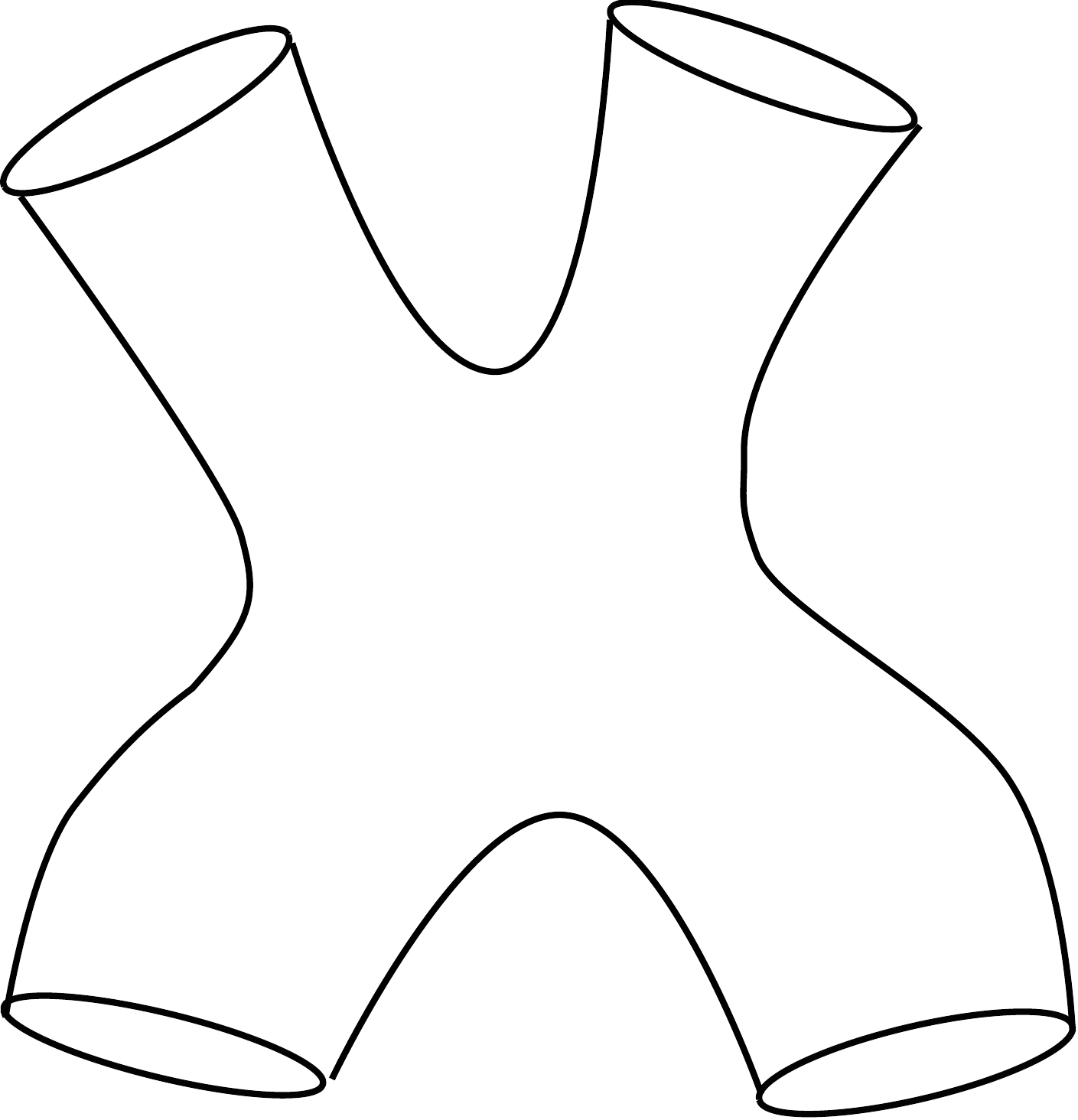}
 \caption{Before we can find the result for a string amplitude, we should specify the details of the in- and out states. E.g.\ we expect the result to depend on the momentum $ k_1, k_2, k_3, k_4$ of the incoming strings.}
\label{fig:Stringy_Interactions_EndStates}
\end{wrapfigure}
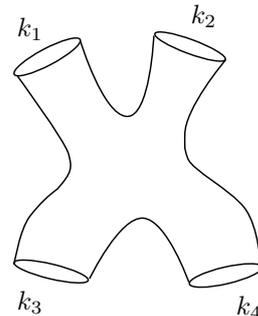

This gives sensible results, that are in accordance with a theory of gravity, since the gauge transformations of gravity (as a gauge theory) are exactly those general coordinate transformations ($D$-dimensional diffeomorphisms) that die off at infinity. Of course this forms a restriction. Such amplitudes with the string sources at infinity, are scattering amplitudes, or S-matrix elements.

When the in- and outgoing string states are taken to infinity, the string diagrams as in figs.\ \ref{fig:Stringy_Interactions}--\ref{fig:Stringy_Interactions_EndStates} have legs that stretch all the way to infinity in spacetime. We can map those string states at infinity to points (``insertion points'') on the worldsheet by a conformal transformation (remember that the string action has conformal invariance -- more on conformal transformations in the next chapter.) I.e.\ the infinitely long legs of the diagram can be rescaled away, see figures \ref{fig:Stringy_Interactions_Punctures} and \ref{fig:Sphere_Punctures}. Moreover, it follows that the considered states must be on the mass-shell (``on-shell'' states, with $m_i^2 = -k_i^\mu k_{i\mu}\,\forall i$). This can be seen in different ways. First, in a gravity theory (such as string theory), off-shell amplitudes make no sense for reasons explained above. Second, the procedure that puts the string states at certain points/insertion points on the worldsheet is a limiting procedure, that exactly puts the states on-shell. 
\begin{figure}[ht!]
\centering
\includegraphics[height=.17\textheight]{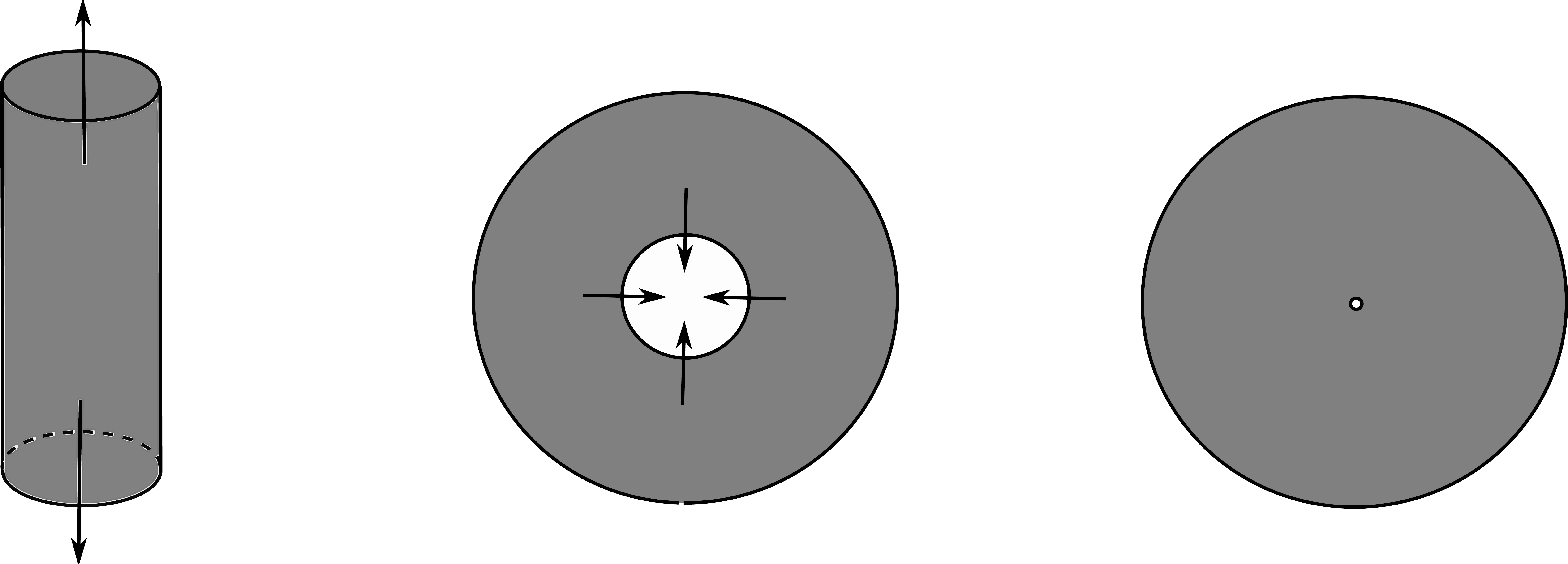}
\caption{By a conformal transformation, a cylinder can be mapped to an annulus, or even, by taking the inner ring of the annulus smaller and smaller, to a disk with a puncture.}
\label{fig:Stringy_Interactions_Punctures}
\end{figure}

\subsubsection{String S-matrix}

We are now ready to express the ``string S-matrix'' as a path integral, describing scattering between states coming in from infinity. Above, we showed that any element of the string S-matrix is associated to some two-dimensional worldsheet with insertion points. There exists a state-operator mapping, that associates to every state a certain operator, localized at such an insertion point. For obvious reasons, these operators are known as vertex operators. An element of the string S-matrix, with vertex operator insertions $V_1\ldots V_n$  can then be written as a path integral (hats for operators, no hats for corresponding functions)
\begin{equation}
 \langle \hat V_1 \ldots \hat V_n\rangle = \int \BWD X \BWD g \,e^{-S_P[X,g]} V_1(p_1)\cdot \ldots \cdot V_n(p_n)\,,\label{eq:Intro_Smatrix1}
\end{equation}
where the integration is over the fields in the Polyakov action $S_P$, namely the embedding coordinates $X$ and the metric $g$. You may find this strange. This 1+1 dimensional path integral promises to give a $D$-dimensional S-matrix, where $D$ is the dimensionality of spacetime! How physics in $D$-dimensions arises, is the subject of the chapters to come.

\subsection{Summing over loops}
So far, you may have the impression that a scattering amplitude in the string S-matrix, corresponds to a worldsheet with the topology of a sphere, with a number of insertion points with associated vertex operator insertions in the path integral, as in figure \ref{fig:Sphere_Punctures}.

\begin{figure}
\centering 
\includegraphics[height=.15\textheight]{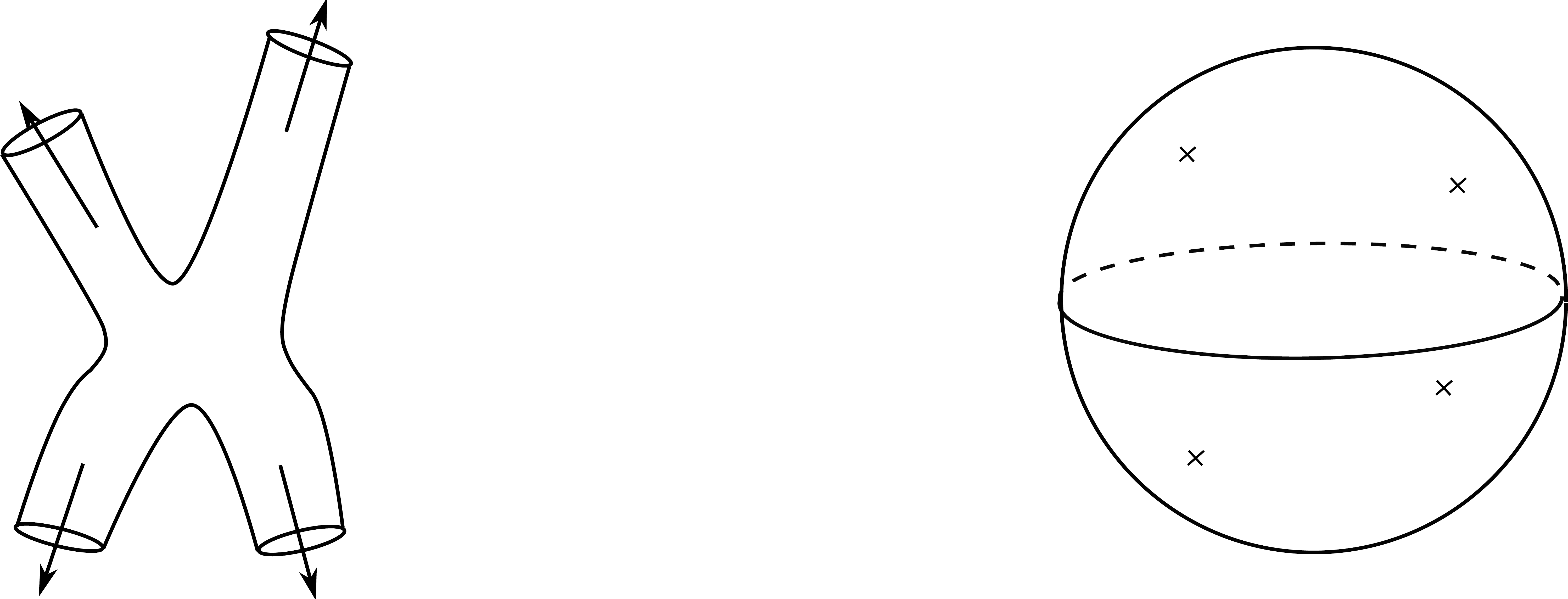}
\caption{By taking the legs to infinity, we can map tree level amplitudes (no loops) with $n$ in and out states, to a sphere with $n$ insertion points (denoted as crosses), presenting the same topology.\label{fig:Sphere_Punctures}}
\end{figure}
However, we should also include worldsheets with other topologies. The classification of possible worldsheet topologies is the subject of the next chapter. The conformal symmetry of the Polyakov action is used to come to the notion of worldsheets as Riemann surfaces and their classification (after a Wick rotation on the worldsheet, that changes the Lorentzian signature into a Euclidean one). For instance, we can take a look at worldsheets with any number of handles. No handles means the worldsheet has the topology of  a sphere, one handle the topology of a torus and so on. The number characterizing the number of handles is called the \textit{genus} of a surface. See figure \ref{fig:DifferentGenus}.
\begin{figure}[ht!]
\centering
\includegraphics[width=.9\textwidth]{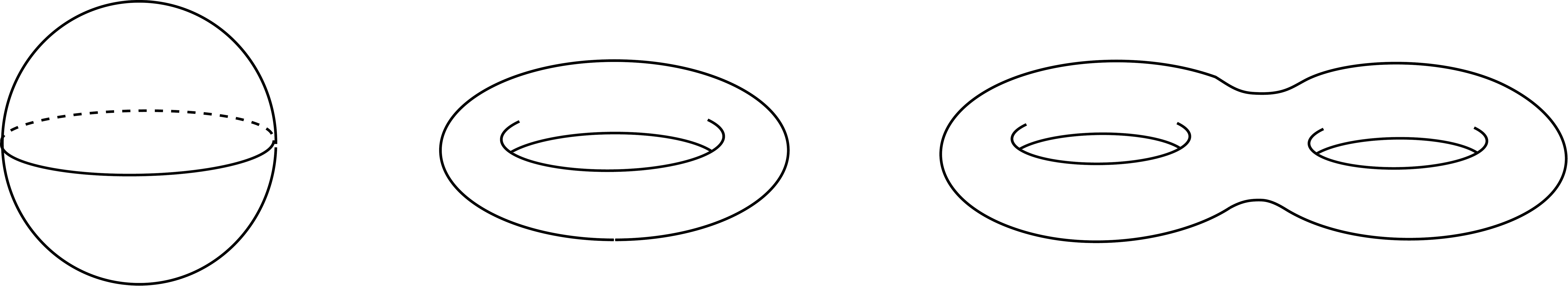}
 \caption{Riemann surfaces of genus 0,1,2.\label{fig:DifferentGenus}}
\end{figure}

It turns out that the path integral \ref{eq:Intro_Smatrix1} can be organised by increasing genus, much like the role the number of loops plays in ordinary QFT. Genus zero corresponds to no loops, genus one can be seen as one (closed) string going around in a loop etc. In order to see that there is also an (increasing) factor of a ``string coupling constant'' associated to the increasing genus, let us go back to the Polyakov action. The Polyakov action is invariant under worldsheet diffeomorphism and Weyl rescalings of the worldsheet metric. It makes sense to consider the most general action with those symmetries:
\begin{eqnarray}
{\cal S}_P = \frac{1}{4\pi \alpha'} \int_{\BWworldsheet} d^2 \sigma
\sqrt{g} g^{ab}
\partial_a X^\mu \partial_b X^\nu G_{\mu \nu}(X) +
\frac{\lambda}{4\pi} \int_{\BWworldsheet} d^2 \sigma \sqrt{g} R,\label{eq:String_Action}
\end{eqnarray}
where $\lambda$ is a constant and $R$ the Ricci scalar of the worldsheet metric $g_{ab}$. By virtue of the world symmetries we can eliminate the three degrees of freedom in the worldsheet metric. This means that $g_{ab}$ does not represent a dynamical gravity theory in two dimensions and the last term in the action is a purely topological term, which can be related to the number of handles g in the worldsheet. Theferfore, we need the Euler number $\chi(M)$ of a two-dimensional surface $M$. For a surface without boundary (appropriate for the description of worldsheets in \textit{closed} bosonic string theory), it is related to the second integral in the adjusted string action \eqref{eq:String_Action} and to the genus of $M$ as: (see section \ref{BWs:WSasRS} for more info)
\begin{equation}
 \chi (M) = \frac{1}{4\pi} \int_M d^2\sigma \sqrt{g}\,,\qquad \chi = 2 - 2 g \,.
\end{equation}
We see the extra term in the action weighs amplitudes in the S-matrix \eqref{eq:Intro_Smatrix1} with a factor of $(e^{-\lambda})^\chi$. We see that the constant $e^{\lambda}$ acts effectively as a string coupling constant, we call $g_s$
\begin{equation}
g_s = e^{\lambda} 
\end{equation}
and we can arrange an amplitude in a sum over worldsheets with different topologies:
\begin{equation}
 \langle \hat V_1 \ldots \hat V_n\rangle = \sum_{topologies} g_s^{-\chi}\int \BWD X \BWD g \,e^{-S_P[X,g]} V_1(p_1) \ldots V_n(p_n)\,,\label{eq:Intro_Smatrix}
\end{equation}
We will later see that the string coupling constant $g_s$ is very natural in string theory, it can be related to the dilaton, the scalar in the massless string spectrum (see table \ref{tab:LowestModes}). Note that the coupling $g_s$ introduced here corresponds to the \emph{open} string coupling constant. The coupling corresponding to adding a closed string to a process is $g_c = g_s^2$.

In conclusion, we see the S-matrix is organised in terms of increasing genus, as an expansion in powers of the string coupling $g_s$, see  figure \ref{BWfig:Loops_Comparison_PointParticle}.
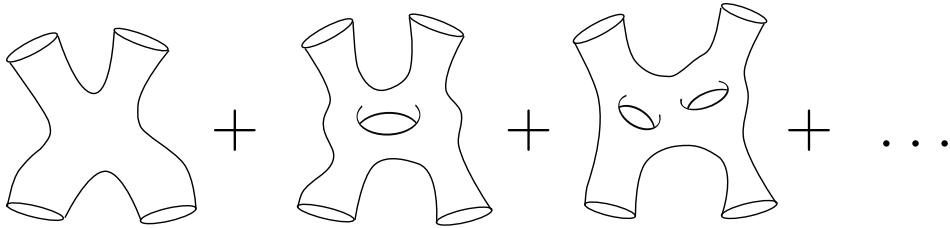
\begin{figure}[ht!]
%
\begin{picture}(20,0)
\put(100,30){\Huge $+$}
\put(210,30){\Huge $+$}
\put(315,30){\Huge $+\,\,\, \ldots$}
\end{picture}
\includegraphics[height=.13\textheight]{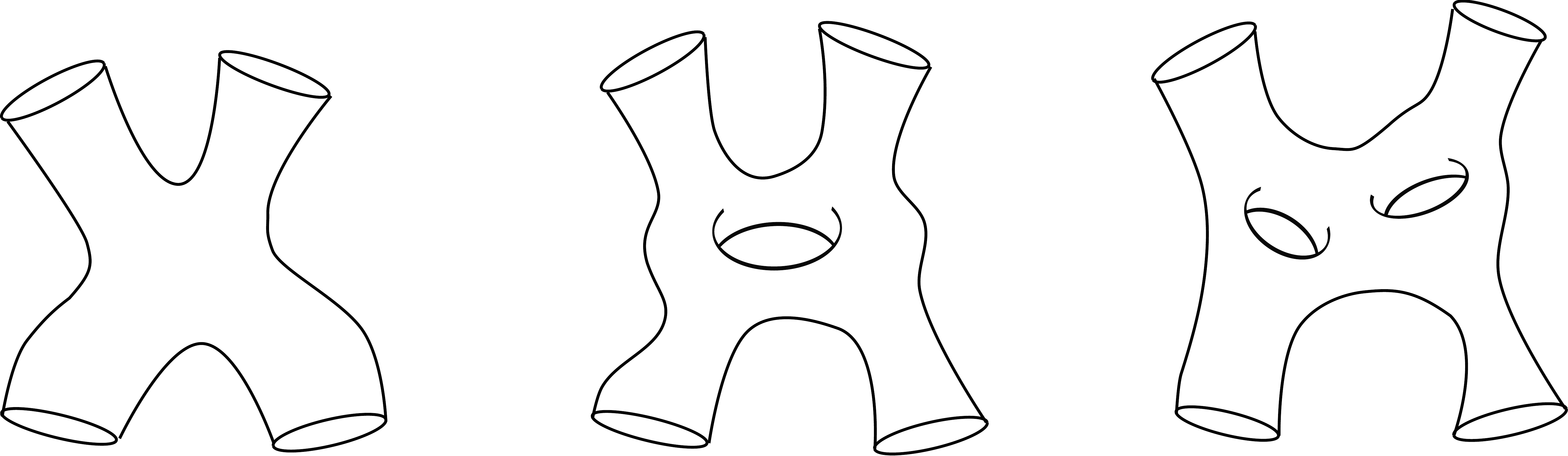}
\caption{One of the many simplifying consequences of using string theory. For a closed (oriented) string theory, there is exactly one diagram (one worldsheet topology) at each loop, while a QFT can have many different types of Feynman diagrams, that increase drastically with the number of loops. (For open string and unoriented string theories, there are more than one possibilities at each loop level, but still far less than in the standard model and the like.) This is in contrast with ordinary QFT, where the number of Feynman diagrams typically grows factorially with the number of loops/number of external legs.\label{BWfig:Loops_Comparison_PointParticle}}
\end{figure}

\subsection{Vertex operators.}\label{BWss:VertexOperators}
We have not mentioned what form the vertex operators, describing the string states in the amplitude, would take. In fact, it is not hard to guess their form. For a pedagogic treatment, see the introduction of \cite{Green:1987sp}. We follow that reference. 

Since the only functions we have available on the worldsheet are the embedding coordinates $X^\mu(\sigma)$ and the worldsheet metric $g_{ab}(\sigma)$, we can easily guess the following form of the operators associated to the lowest mass modes of table \ref{tab:LowestModes}, by demanding the correct correspondence of the quantum numbers of each state/operator. We denote the candidate vertex operator as ${\cal V}(\sigma)$. For example, the tachyon is a scalar field, so we guess the easiest possibility, namely the identity operator. In the same way, we guess for the graviton, a symmetric two-tensor, ${\cal V}_{\rm grav} = s_{\mu\nu}g^{ab} \partial_a X^\mu \partial_b X^\nu$, where $s_{\mu\nu}$ is a symmetric, constant tensor. For the antisymmetric two-tensor we propose a contraction with the totally anti-symmetric symbol $\epsilon_{ab}$ on the worldsheet as ${\cal V}_{\rm antisymm}=a_{\mu\nu}\epsilon^{ab} \partial_a X^\mu \partial_b X^\nu$, and $a_{\mu\nu}$ is an antisymmetric constant tensor. The dilaton is another scalar field as far as spacetime symmetries are concerned. Since the unity operator is already taken by the tachyon, we make the next simple guess ${\cal V}_{dil} = \phi \alpha' R$, where $\phi$ is a constant and $R$ the Ricci scalar of the worldsheet. 
For an overview and some more details, see table \ref{tab:LowestMass_Vertex}. 
\begin{table}[h t!]
\centering
\begin{tabular}{|l|cl|}
\hline
&$\mathbf{m}^2$&{\bf Operator} $\mathbf {\cal V}(\sigma)$\\
\hline\hline
tachyon&$-4/\alpha'$&$\BWe^{\BWi k_\mu X^\mu}$\\
\hline
graviton&0&$g_s s_{\mu\nu}g^{ab} \partial_a X^\mu \partial_b X^\nu \BWe^{\BWi k_\mu X^\mu}$\\
antisymm. tensor&0&$g_s a_{\mu\nu}\epsilon^{ab} \partial_a X^\mu \partial_b X^\nu\BWe^{\BWi k_\mu X^\mu}$\\
dilaton&0&$\phi \alpha' R \BWe^{\BWi k_\mu X^\mu}$\\
\hline
\end{tabular}
\caption{Lowest lying string states and the first guess for their associated vertex operator ${\cal V}(\sigma)$. The factor $\BWe^{\BWi k_\mu X^\mu}$ should be included to give these operators the same behaviour under spacetime translations $X^\mu \to X^\mu + a^\mu$ as the corresponding states, e.g.\ for the tachyon $|0;k\rangle \to \BWe^{\BWi k_\mu a^\mu} |0;k\rangle$ (since momentum space and coordinate space are related by a momentum translation). We have also added a factor of $g_s$, since an extra string is created in the process corresponding to adding the vertex operator. If you do not like this factor, think of it as a way to fix the normalization of the vertex operators.\label{tab:LowestMass_Vertex}}
\end{table}

We have only considered the Lorentz numbers of the vertex operators. In order to ensure invariance under  diffeomorphisms on the worldsheet, a vertex operator should take the following form:
\begin{equation}
 V = \int \BWd^2 \sigma \sqrt{g} {\cal V(\sigma)}\,.\label{BWeq:VertexOp}
\end{equation}

Note that the string coupling constant $g_s = \BWe^\lambda$ introduced in eq. \eqref{eq:String_Action}, is actually determined by the vacuum expectation value of the dilaton. This can be seen from the form of the vertex operators. Say we would put a string in curved spacetime. Its coupling to all $m^2=0$ states (caused by interactions with a background of \textit{other} strings) would then take the form:
\begin{equation}
 S_P' = \frac1{4\pi\alpha'}\int_{\BWworldsheet} \BWd^2 \sigma\sqrt{g} \left(G_{\mu\nu}(X) g_{ab} \partial^a X^\mu\partial^b X^\nu + B_{\mu\nu}(X) \epsilon_{ab} \partial^a X^\mu\partial^b X^\nu + \alpha' \Phi R\right)\,,
\end{equation}
where $G_{\mu\nu}X)$ is the spacetime metric, $B_{\mu\nu}(X)$ an antisymmetric tensor in spacetime and $\Phi(X)$ a spacetime scalar. These exponentials should be seen as coherents states of gravitons, antisymmetric two-tensor fields and dilatons. This is justified by looking at the first terms in the Taylor series when expanding around the Minkowski vacuum $G_{\mu\nu}=\eta_{\mu\nu}, B_{\mu\nu}=0,\Phi=cst$, for instance for the spacetime metric 
\begin{equation}
 \frac1{4\pi\alpha'}\int_{\BWworldsheet} \BWd^2 \sigma\sqrt{g} G_{\mu\nu}(X) g_{ab} \partial^a X^\mu\partial^b X^\nu = S_P + \frac1{4\pi\alpha'}\int_{\BWworldsheet} \BWd^2 \sigma\sqrt{g} G^{(1)}_{\mu\nu}(X) g_{ab} \partial^a X^\mu\partial^b X^\nu
\end{equation}
where the first order term $G^{(1)}_{\mu\nu}$ is exactly interpreted as the vertex operator for the graviton state $G_{\mu\nu}^{(1)}=-4\pi g_s \alpha's_{\mu\nu} e^{\BWi k X}$.

Finally, we see that the extra term with the scalar curvature on the worldsheet, introduced in \eqref{eq:String_Action} and responsible for the organisation of the path integral in increasing loop number, is understood as a constant dilaton background (i.e.\ it arises as the constant term $\phi$ in $\Phi(X) = \phi + \dots$).\footnote{To be precise, the actual dilaton contains both $\Phi$ and the diagonal part of $G_{\mu\nu}$, see \cite{Polchinski:1998rq}.} Or in other words, by the vev of the dilaton. This is one of the many hard aspects of string theory. The coupling inherent to a perturbative treatment, is in principle determined by the dynamics of the strings in their own background. Although very nice from conceptual viewpoint (we do not have a ``free parameter'' as the couplings in the standard model), it is very tricky to calculate it from the coupled string dynamics.

\subsection{Other issues and glimpse forward}
A main problem one encounters in evaluating the path integral for string scattering amplitudes, is that due to the worldsheet symmetries of the Polyakov action (Weyl transformations and diffeomeorphisms), the path integral counts physically equivalents states multiple times. In fact, this overcounting renders the path integral infinite. There exist many equivalent methods to deal with this problem. For instance, one can choose to  fix these gauge symmetries in the Polyakov path integral explicitly, or one can couple the string theory to an appropriate ghost system and use the BRST-formalism to deal with the overcounting problem.
One can also use the conformal symmetry of the worldsheet more explicitly and construct the amplitudes using holomorphic properties. In these set of lectures we want to focus on the nice relation between the residual conformal symmetry on the string worldsheet and the mathematical theory of Riemann surfaces. We also prefer to emphasize the problem of overcounting due to the large amount of symmetries on the worldsheet by gauge fixing the path integral explicitly. The in-and out-states representing the scattering strings in the process will be represented by Vertex-operators inserted at certain points on the string worldsheet.

Let us first focus  on the relationship between the conformal symmetry on the worldsheet and the mathematical theory of Riemann surfaces, the main theme of chapter \ref{BWc:Riemann}. We have already seen that the Polyakov-action contains two types of local symmetry: two-dimensional diffeomorphisms and Weyl-rescalings of the metric. These two types of symmetry are related, as there exist diffeomorphisms which can be seen as Weyl-rescalings as far as the metric is concerned. Those diffeomorphisms that affect the metric ony by a rescaling, are called \emph{conformal transformations}. (The name comes from the fact that these diffeomorphisms, in their infinitesimal form, leave the \emph{angles} between vectors unaffected.) Moreover, one can use the diffeomorphisms and Weyl-rescalings to gauge fix the worldsheet metric, but the conformal transformations will form a residual symmetry which remains present. \footnote{In fact, the transformations leaving the fixed metric invariant, are each combinations of a conformal transformation and a Weyl transformation that rescales the metric to its original form.}

If we do gauge fix the worldsheet metric to the standard two-dimensional Minkowski-metric, we can introduce complex coordinates on the worldsheet,\footnote{This involves a Wick rotation $\sigma^0 \to \BWi \sigma^0$, which is explained in chapter 2. The rationale for such a Euclideanisation of the worldsheet is  twofod. On the one hand, it takes us to the well-known realm of Riemann surfaces instead of Lorentzian worldsheets, and on the other hand it makes the path integrals convergent.}
\begin{eqnarray}
z\equiv \frac{1}{\sqrt{2}} \left( \sigma^1 + \BWi \sigma^0\right), \quad \bar z \equiv \frac{1}{\sqrt{2}} \left( \sigma^1 - \BWi \sigma^0\right).
\end{eqnarray}
We add the additional constraint $z^* = \bar z$. With the choice of unit metric on the worldsheet ($\BWd s^2 = \BWd z\BWd \bar z$ in complex coordinates), this allows us to rewrite the Polyakov-action eq.~(\ref{BWeq:Polaction}) in terms of complex worldsheet coordinates,
\begin{eqnarray}
{\cal S}_P = \frac{1}{2\pi \alpha'} \int_{\BWworldsheet} d^2z\,
\partial X^\mu \bar\partial X^\nu G_{\mu \nu}(X).\label{BWeq:gaugefixedPol}
\end{eqnarray}
In this form of the action\footnote{In the remainder of these notes we assume that the string moves in a Minkowski spacetime, so that we can replace $G_{\mu \nu}(X)$ by $diag(-,+,\cdots,+)$.} the conformal transformations correspond to holomorphic transformations and the worldsheet can be endowed with the structure of a Riemann surface. In chapter \ref{BWc:Riemann} we see that Riemann surfaces can be classified according to their universal covering surface and the fundamental group $\pi_1(\BWworldsheet)$. The second part of chapter \ref{BWc:Riemann} gives an overview of the properties of Riemann surfaces that are are crucial for our further dissertation, such as the automorphism group and the moduli space of Riemann surfaces. The automorphsim group denotes the subset of $\BWgauge$ that does not affect the metric. The moduli parametrize the different choices of inequivalent metrics, not related by the combination of a diffeomorphism and a Weyl-transformation, one can place on most Riemann surfaces.

The diffeomorphism and Weyl-rescaling invariance lie at the heart of the overcounting problem in the Polyakov-action, which is the subject of chapter \ref{BWc:Fixing}. We consider first an easy toy model to see how gauge fixing works in practice. Keeping this toy model in the back of our mind, we discuss the gauge fixing of the Polyakov path integral. Fixing the Polyakov path integral basically comes down to integrating over the diffeomorphisms, the Weyl-rescalings and the moduli instead of integrating over the metric. The Jacobian corresponding to this ``change of coordinates" should be calculated explicitly. By using the moduli explicitly we have already eliminated the subset of diffeomorphisms with Weyl-rescalings (i.e. the diffeomorphisms that can also be seen as Weyl-rescalings). Only the conformal Killing symmetry group should thus be moded out. When there is a sufficient amount of vertex-operators we can use the conformal Killing symmetries to fix the positions of (some of) the vertex-operators. 

Once we know how to perform the integration of the metric in the Polyakov path integral, we are still left with the integration of the string embedding functions $X^\mu$. Also this integration depends on the metric of the worldsheet and the genus of the worldsheet, as we  see in chapter \ref{BWc:Amplitudes}. In this chapter we  work out some explicit examples of string scattering amplitudes for the sphere (genus 0) and the torus (genus 1). For the sphere we  limit ourselves to amplitudes with one, two, three and four tachyon operator insertions. For the torus we  focus on the partition function, which will explicitly exhibit the modulus parameter of the torus. 

In the last chapter we use the toolbox put together in the previous chapters to discuss higher genus amplitudes. We also have a brief look at supersymmetric string theories, where we have similar problems of gauge fixing. Afterwards, we take the low-energy limit of the scattering amplitudes (i.e. $\alpha ' \rightarrow 0$) and argue that the resulting amplitudes can be obtained from gravity theories. We also briefly discuss the Type II and Type Heterotic superstring theory and consider the bosonic part of their low-energy supergravity action. Also here the relationship between superstring theory and supergravity can be made clear looking at the scattering amplitudes, but this would take us to far from the main objectives of these lectures. It is known that many supergravity theories cannot be considered as UV-finite, but it is believed that superstring theories form the UV-finite limit of supergravity theories. This encourages us to look at the UV-behavior of string theories in the last part of chapter \ref{BWc:Conclusions}. \\

\emph{NOTE:} we have already encountered words like conformal and holomorphic many times. These terms are very much related to conformal field theory. For instance, the two-dimensional Polyakov-action is an excellent example of a two-dimensional conformal field theory. In these lecture notes we do not use the machinery of conformal field theory. Instead, we refer to the lectures of Raphael Benichou at this Modave School for an introduction to conformal field theory, or to the bible of conformal field theories \cite{DiFrancesco:1997nk} for a concise treatment. A treatment of scattering amplitudes using conformal field theory can be found in e.g. \cite{Polchinski:1998rq}.


\chapter{A pedestrian's guide to Riemann Surfaces\label{BWc:Riemann}}

\textit{In the introductory chapter \ref{BWc:Introduction} we have introduced the worldsheet of the string as the shape swept out by a string when propagating in spacetime. The two-dimensional field theory living on the worldsheet exhibits two types of local symmetry, i.e.\ diffeomorphism invariance and Weyl transformations. In this chapter we  start by giving a quick review of some basic properties of a two-dimensional manifold (endowed with a metric). The additional conformal symmetry (diffeomorphisms which leave the metric invariant up to a Weyl rescaling) on an oriented worldsheet induces the structure of a Riemann surface. The conformal symmetry thus invites us to have a closer look at Riemann surfaces, including their classification. A proper understanding of the characteristics of Riemann surfaces, such as conformal Killing vectors, moduli, etc is indispensable when we construct the correct string S-matrix. This chapter is based on \cite{Hitchin:2004b3c, Olson:200un, Lee:1997, Farkas:1980, Miranda:1995}.}

\section{The String worldsheet as a Riemann Surface\label{BWs:WSasRS}}
\subsection{The worldsheet as a surface}
When a string propagates in spacetime, it sweeps out a two-dimensional surface in
spacetime, which we called the worldsheet $\BWworldsheet $ of the
string. This shape can be described by the notion of a (smooth)
manifold\footnote{For a proper mathematical definition of the
concepts we use in this chapter, we refer to 
\cite{Nakahara:2003nw}.}. We can stitch patches on this shape that look locally
like $\BWIR^2$ and thus with a point in such a patch we can assign a
pair of (real) coordinates $(\tau, \sigma)$. From a physical point
of view $\tau$ is interpreted as the eigen-time of the string,
$\sigma$ as the eigen-length of the string. However, in most cases
we need several different patches to cover the entire shape. On the
overlap between two patches we can assign two different pairs of
coordinates to one single point. The transition from one patch to
another should make the patches compatible with each other, which is
expressed mathematically by diffeomorphic transition functions. Take
e.g. $(\tau, \sigma)$ in the first patch and $(\sigma^0, \sigma^1)$
in the second patch, then $\sigma^0 (\tau, \sigma)$ and $\sigma^1 (
\tau, \sigma)$ should be diffeomorphic functions. The group of
diffeomorphisms of the worldsheet $\BWworldsheet $ will be called
Diff($\BWworldsheet $). A two-dimensional manifold will be called a {\bf (topological) surface}\footnote{We
prefer to add the word topological to the definition to emphasize
the mathematical use of the word. However in the remainder we shall
use the word surface.}.
\noindent Let us first have a look at some famous examples of surfaces.
\begin{BWex}
The most obvious example is the plane $\BWIR^2$ itself, where the
identity function forms the diffeomorphic transition function.
\end{BWex}
\begin{BWex}
The sphere $S^2$. To see that $S^2$ is indeed a surface, we use the
stereographic projection from the north pole: $(x, y, z) \mapsto  (\frac{x}{1-z}, \frac{y}{1-z})$. A full atlas also contains the stereographic projection from the south pole.
\begin{figure}[h]
\begin{center}
\includegraphics[width=0.6\textwidth]{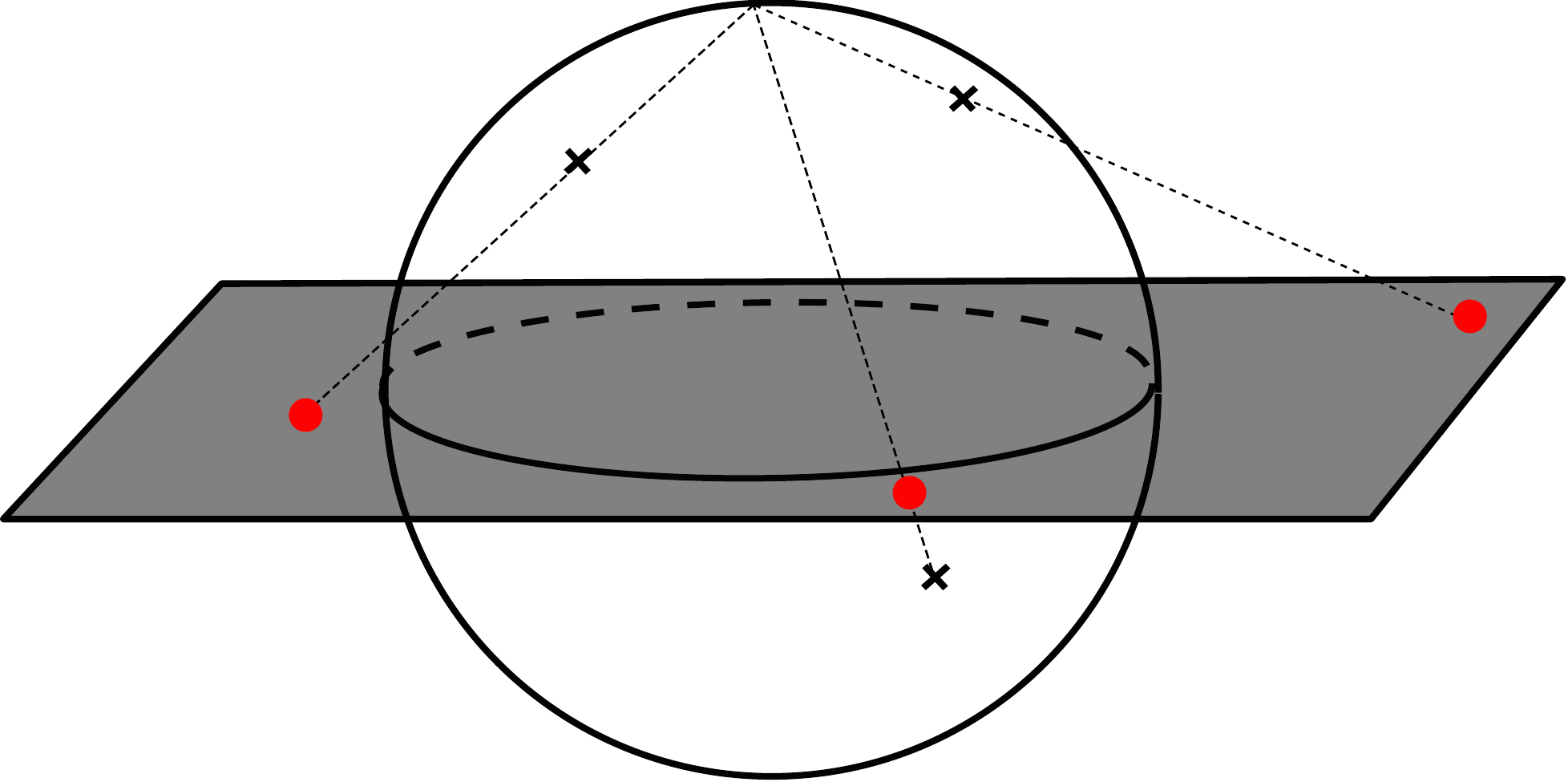}
\caption{Stereographic projection of the 2-sphere. Every point on the 2-sphere can be mapped onto the plain, except the antipodal point of the pole from which the projection starts.}
\label{BWref:StereoProj}
\end{center}
\end{figure}
\end{BWex}
\begin{BWex}
The disk $D_2$ and the closed disk $\bar D_2$.
\end{BWex}
\begin{BWex}
The torus $T^2$.
\end{BWex}
\begin{BWex}
The M\"{o}bius strip $\BWIMS$.
\end{BWex}
\noindent There are many more interesting examples, such as the
projective plane $\BWIRP^2$, the cylinder (or annulus) $C_2$, the
Klein bottle $\BWIKB$, etc. Instead of going through all these
examples explicitly it is more interesting to divide the surfaces
into two subclasses: oriented and unoriented surfaces.
\begin{BWdef}
A surface $M$ is {\bf oriented} (or orientable) if for any two
overlapping charts $U_i$, $U_j$ there exists local coordinates $\{
x^\mu\}$, $\{y^\alpha \}$ respectively, such that the jacobian
$J=\det(\partial x^\mu /
\partial y^\alpha) > 0$.
\end{BWdef}
\noindent When a surface (or more generally a manifold) is orientable, there exists a 2-form (or m-form respectively) on the surface which vanishes nowhere. This 2-form can be seen as a volume element. The topologically invariant property of orientability implies the following division of the surfaces we
discussed earlier,
\begin{table}[h]
\begin{center}
\begin{tabular}{|ccc|}
\hline \multirow{2}{*}{genus g} & \multirow{2}{*}{oriented} & \multirow{2}{*}{unoriented} \\
 & & \\
 \hline
 \hline
\multirow{2}{*}{${\rm g} = 0$} & \multirow{2}{*}{$\BWIR^2$, $S^2$, $D_2$} &\multirow{2}{*}{$ \BWIRP^2$}\\
& &\\
\multirow{2}{*}{${\rm g} = 1$} &\multirow{2}{*}{$T^2$, $C_2$} & \multirow{2}{*}{$\BWIMS$, $\BWIKB$} \\
& & \\
\hline
\end{tabular}
\caption{Overview table of oriented and unoriented surfaces for
genus 0 and genus 1}
\end{center}
\end{table}
\\

Besides orientability there exist other ways to characterize a
surface. One of the simplest ways is to count the number of
holes or handles in a surface. This number is called the {\bf genus}
g of a surface.\footnote{Mathematically, the genus can be defined by a suitable polygonization, see e.g. \cite{Nakahara:2003nw, Hitchin:2004b3c, Farkas:1980}.}
For our purposes the surfaces will be embedded in a higher-dimensional spacetime, by which the handles will be manifested in the shape of the surface. The most famous example is the embedding of the torus $T^2$ in $\BWIR^3$.
\begin{figure}[ht!]
\centering
\includegraphics[width=.8\textwidth]{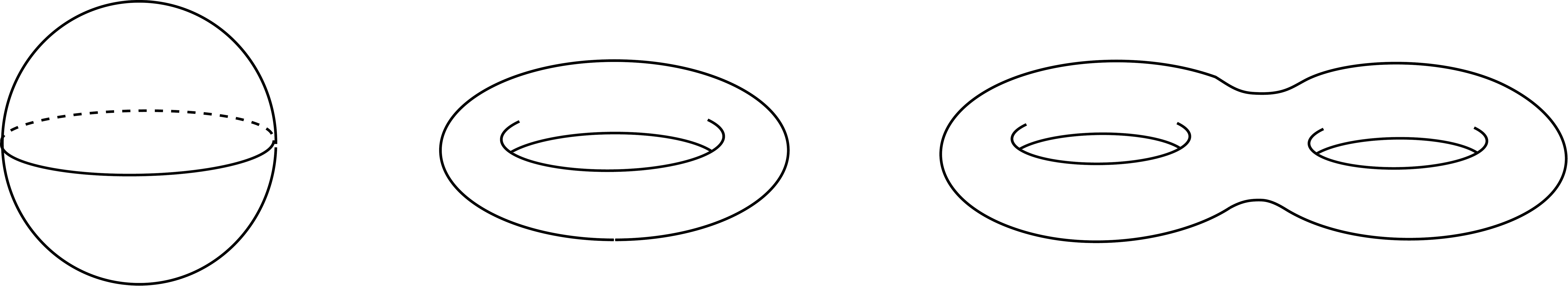}
 \caption{Riemann surfaces of genus 0,1,2.\label{fig:DifferentGenus2}}
\end{figure}


In case of a compact surface $M$ embedded in $\BWIR^3$ (or more
generally in $\BWIR^n$) there is an important topological invariant,
the {\bf Euler characteristic} $\chi(M)$, which allows to
distinguish surfaces from each other. The most straightforward way
to determine the Euler characteristic is to (continuously) transform
the compact surface into a polyhedron with $V$ vertices, $E$ edges
and $F$ faces. The Euler-characteristic is defined as,
\begin{BWdef}
\begin{eqnarray}
\chi(M) \equiv V - E + F.
\end{eqnarray}
\end{BWdef}
An important statement\footnote{As a consequence it is possible to
give a much nicer treatment of deforming a surface into a polyhedron
using simplexes (generalizations of points, lines and triangles).
The use of simplexes allows for a triangulation not only of surfaces
but also of higher-dimensional objects, which leads to the field of
simplicial homology. It is thus possible to assign an
Euler-characteristic to higher-dimensional objects as well.} is that
it does not matter which polyhedron one uses, as long as the
polyhedron can be continuously transformed into the surface. There
is another useful, and more general, formula to calculate the Euler-characteristic of a
surface with g handles, $n_b$ boundaries and $n_c$ cross-caps,
\begin{eqnarray}
\chi (M) = 2 - 2g - n_b - n_c.
\end{eqnarray}
\begin{BWexerc}
Determine the Euler-characteristic of the surfaces given above.
\end{BWexerc}

\vspace{0.3cm}

\noindent Besides the manifold-structure, the string worldsheet
also contains a metric structure $g_{ab}$ in the
Polyakov-formulation. For the remainder of these lectures we will
assume that the worldsheet metric has a Euclidean signature. As is
well-known, going from a Lorentzian signature to a Euclidean
signature amounts to performing a Wick-rotation of the worldsheet
time-coordinate 
\begin{equation}
\sigma^0\to \sigma^2 = \BWi \sigma^0 
\end{equation}
One can show that performing\footnote{This argument
only holds in 1 or 2 dimensions. A brief formal argument is for instance given in chapter 3 of \cite{Polchinski:1998rq}} a Wick-rotation in the Lorentzian
path integral yields the Euclidean path integral, which justifies
the equivalence between the Lorentzian and Euclidean description.\\
Thus, from now on, we shall discuss surfaces which allow a
Riemannian metric structure. This implies the existence of the
metric connection ${\Gamma^{a}}_{bc}$ and of the objects
characterizing the (intrinsic) curvature of the surface, such as the
Riemann curvature tensor ${R^{a}}_{ bcd}$, the Ricci-tensor $R_{ab}$
and the Ricci-scalar $R$. In two dimensions the symmetries of the
Riemann-tensor imply that one independent component is sufficient to
characterize the entire Riemann-tensor. Therefore, we can write the
Riemann-tensor in terms of the Ricci-scalar as follows,
\begin{eqnarray}
R_{abcd} = \frac 1 2 \left(g_{ac} g_{bd} - g_{ad} g_{bc} \right) R.
\end{eqnarray}

We end this swift review on surfaces with one of the major
theorems for two-dimensional, compact, Riemannian manifolds $M$,
i.e.\ the Gauss-Bonnet Theorem. The theorem  states a connection
between a local property, the integral of the curvature, and a
global topological invariant, the Euler-characteristic. Taking also
boundaries into account, the theorem reads,
\begin{eqnarray}
\chi (M) = \frac{1}{4\pi} \int_M d^2\sigma \sqrt{g} R +
\frac{1}{2\pi} \int_{\partial M} ds\, k.
\end{eqnarray}
The second integral describes the integration of the geodesic
curvature $k$ of the boundary along the boundary of the manifold.
The geodesic curvature of the space-like boundary is defined by,
\begin{eqnarray}
k =  t^a n_b \nabla_a t^b,
\end{eqnarray}
where $t^a$ is a unit vector tangent to the boundary and $n^a$ an
outward pointing vector orthogonal to $t^a$. A proof of this
beautiful theorem (without the boundary contribution) can be found
in e.g. \cite{Lee:1997}. To demystify the theorem we propose the
following exercise.
\begin{BWexerc}
Calculate the Euler-characteristic for $S^2$, $D_2$ and $T^2$ using
the Gauss-Bonnet Theorem.
\end{BWexerc}

\subsection{Conformal Symmetry}
Besides diffeomorphism invariance we also noticed that the
Polyakov-formulation of the string exhibits another symmetry:
Weyl-symmetry. This Weyl-symmetry only arises in a Polyakov-type of
action for two dimensions. It therefore forms one of the special
ingredients that makes this formulation of string theory so interesting and attractive to
study.
\begin{BWdef}
A {\bf Weyl-transformation} is a
transformation on the metric of the form,
\begin{eqnarray}
g(x) \rightarrow  e^{2 \omega(x)} g(x)
\end{eqnarray}
where the {\bf conformal factor} $\omega \in C^{\infty}(M)$. Two
metrics that are related to each other through a Weyl-transformation
are called {\bf Weyl-equivalent}.
\end{BWdef}
\begin{BWdef}
A {\bf conformal transformation} is a diffeomorphism which preserves the metric up to a Weyl transformation.
\end{BWdef}
Weyl-equivalence clearly defines an equivalence relation on the
space of metrics $\BWmetric$.
 and an equivalence class $[g]$ is
called a {\bf conformal structure} on $M$. The Weyl-transformations form the group Weyl($M)$ and the conformal transformations the group Conf$(M)$. It is clear from the notations that Weyl-transformations and conformal transformations are closely related to each other. These notions form the basis of conformal geometry, which studies the properties that are invariant under conformal transformations.\\
So far, the definitions we gave are valid for Riemannian manifolds with arbitrary dimensions. Now we shall see that two dimensions are indeed special. When performing a Weyl transformation on a two-dimensional metric, $\sqrt{g} R$ transforms as follows,
\begin{eqnarray}
\sqrt{\tilde g} \tilde R = \sqrt{g} \left(R - 2 \nabla^2 \omega
\right)
\end{eqnarray}
By choosing $\omega$ appropiately we can locally set the
Ricci-scalar to zero. This statement is related to the statement
that every two-dimensional metric can locally be written as a flat
metric up to a conformal factor. Suppose we start from the general
form of the metric:
\begin{eqnarray}
ds^2 = g_{xx} dx^2 + 2 g_{xy} dx dy + g_{yy} dy^2,
\end{eqnarray}
and perform the following coordinate transformation,
\begin{eqnarray}
\frac{1}{\lambda} \left(  du + i dv \right) = \sqrt{g_{xx}} dx +
\frac{g_{xy} + i
\sqrt{det(g_{ab})}}{\sqrt{g_{xx}}} dy , \\
\frac{1}{\bar\lambda } \left( du - i dv \right) =  \sqrt{g_{xx}} dx
+ \frac{g_{xy} - i \sqrt{det(g_{ab})}}{\sqrt{g_{xx}}} dy .
\end{eqnarray}
Then the metric reduces to the form $ds^2 = e^{2\omega} (du^2 +
dv^2)$, where $|\lambda|^2 =e^{-2\omega}$ and $\lambda (x,y) =
\lambda_1 (x,y) + i \lambda_2 (x,y)$. A metric which can be written
locally in this form is called {\bf (locally) conformally flat} and
the coordinates $(u, v)$ are called {\bf isothermal coordinates}.
Basically, we have shown the following theorem,
\begin{BWtheorem}
Every 2-dimensional Riemannian manifold is locally conformally flat.
\end{BWtheorem}
\begin{BWexerc}
Find isothermal coordinates $(u,v)$ and a conformal factor $\omega$
such that the standard metric on $S^2$, $ds^2 = d\theta^2 + \sin^2
\theta d\phi^2$, reduces to a (locally) conformally flat metric.
\end{BWexerc}
We are now ready to give a first (intuitive) correspondence between
conformal geometry and Riemann surfaces. Let us start from a surface
$M$ for which the metric locally can be written as a flat metric,
\begin{eqnarray}
ds^2 = dx^2 + dy^2 = dz d\bar z.
\end{eqnarray}
Here, we introduce the complex coordinate $z\equiv x + iy$ and its
complex conjugate $\bar z $. Suppose we perform a coordinate
transformation that only depends\footnote{In the next section we
shall encounter these type of coordinate transformation again and
call them holomorphic.} on the complex coordinate z,
\begin{eqnarray}
z &\rightarrow& z' = f(z),\\
\bar z &\rightarrow& \bar z' = \bar f (\bar z).
\end{eqnarray}
Due to this coordinate transformation, the form of the metric will
change, but only up to a scale\footnote{We can gauge fixe the worldsheet metric of the Polyakov action, using the two-dimensional reparametrization and Weyl-rescaling invariance. However, there still remains a residual symmetry left, namely the conformal transformations, in the gauge fixed Polyakov action \ref{BWeq:gaugefixedPol}. One can perform a conformal transformation on the gauge fixed worldsheet metric, which renders the same metric up to a conformal factor. The Polyakov action itself remains invariant under conformal transformations.}. This gives us the impression that holomorphic transformations are related to conformal transformations.  Indeed, suppose we transform the metric into a
Weyl-equivalent metric by a Weyl-transformation, we obtain a metric of the form,
\begin{eqnarray}
ds^2 = e^{2 \omega} |\partial_z f |^{-2} dz' d\bar z'.
\end{eqnarray}
By choosing,
\begin{eqnarray}
\omega = \ln | \partial_z f|
\end{eqnarray}
we obtain a flat metric again. We can thus conclude that a conformal transformation can also be seen as a holomorphic transformation, which are the defining transition functions for a complex manifold.

\section{Riemann Surfaces\label{BWs:RiemannSurfaces}}
\subsection{What is a Riemann Surface?}
There exist several equivalent ways of looking at Riemann surfaces
and thus there exist several equivalent ways of defining them,
depending on the point of interest. In this section we  try
to define Riemann surface as straightforward as possible. We
therefore need to introduce the concept of a holomorphic function.
\begin{BWdef}
A function $ f: \BWIC \rightarrow \BWIC$ is called {\bf holomorphic}
if and only if $f = f_1 + i\, f_2$ satisfies the Cauchy-Riemann
relations $\forall \, z = x + i y \in \BWIC$,
\begin{eqnarray}
\frac{\partial f_1}{\partial x} = \frac{\partial f_2}{\partial y},
\quad \frac{\partial f_2}{\partial x} = - \frac{\partial
f_1}{\partial y}.
\end{eqnarray}
\end{BWdef}
These holomorphic functions will serve as transition functions
between two different patches on a Riemann surface.
\begin{BWdef}
$M$ is a {\bf Riemann surface} if and only if it is a (topological)
surface with a set of charts $\{ (U_\alpha, z_\alpha) \}$ such
that,\\
\indent $\bullet$ $\cup_\alpha U_\alpha = M$,\\
\indent $\bullet$ $z_\alpha : U_\alpha \rightarrow \BWIC$ is a
homeomorphism onto an open subset of $\BWIC$,\\
\indent $\bullet$ for every $U_\alpha, U_\beta$ for which $U_\alpha
\cap U_\beta \neq \emptyset$ then the transition function $z_\alpha
\circ z_\beta^{-1}$ is holomorphic.
\end{BWdef}
The dimension of the Riemann surface depends on which algebraic number field
one refers to: dim$_\BWIR M $ = 2 dim$_\BWIC M$ = 2. Again we should look
at some important and well-known examples of Riemann surfaces.
\begin{BWex}
The most obvious example is the complex plane $\BWIC$ itself, with
the identity function as the holomorphic transition function.
\end{BWex}
\begin{BWex}
\label{BWex:RiemannSphere}
Another well-known example is the Riemann sphere $\BWIC \cup
\{\infty \}$. We use the complexified version of the stereographic projection:  $(x, y, z) \mapsto  \frac{x+i y}{1-z}$. A similar relation can be derived for the stereographic projection from the south pole. The transition function which relates the upper part of the sphere to the lower part of the sphere is given by $z= 1/u$.
\end{BWex}
\begin{BWex}
The complex disc $D = \{ w \in \BWIC | \, |w| < 1 \}$ and the
upper half plane ${\cal U} = \{ z \in \BWIC | \, Im (z) > 0\}$ are related surfaces as they can be conformally mapped to each other by the mapping $w= e^{i \phi_0}\frac{ z -z_0}{z-\bar z_0}$ 
\end{BWex}
\begin{BWex}
Another important example is the complex torus $T^2$, which can be defined as follows. Take two $\BWIR$-linearly independent complex numbers $\omega_1$ and $\omega_2$ $\in$ $\BWIC$ and define the following lattice $L(\omega_1, \omega_2) \equiv \{ m \omega_1 + n\omega_2 | m,n, \in \BWIZ   \}$. We find the complex torus by identifying complex coordinates on the plane under this lattice, i.e.\ for $z_1$, $z_2$ $\in \BWIC$ we identify $z_1 - z_2 = m \omega_1 + n \omega_2$ for some $m, n \in \BWIZ$.  
\end{BWex}
\begin{BWex}
The cylinder $C_2$
\end{BWex}

This definition of Riemann surfaces closely follows the basic
concepts of complex geometry \cite{Nakahara:2003nw}. In complex geometry
one can show that instead of characterizing a complex manifold by a
complex atlas, one can use the notion of a complex structure. It is known that a necessary and sufficient condition for a
manifold to be complex is the existence of an integrable, complex
structure ${J^a}_b$. Let us first consider an almost complex structure, which is a (1,1)-tensor which squares to $-\mathbb{1}$ locally. A necessary and sufficient condition for this tensor to square to $-\mathbb{1}$ globally is given by the vanishing of the Nijenhuis-tensor ${N^a}_{b c}$, which is defined as working on two vectors $N(X,Y)$, returning a third as\footnote{$[X,Y]$ is the Lie-bracket on vector fields $X$ and $Y$.},
\begin{eqnarray}
N(X,Y) \equiv [X,Y] + J[JX, Y] + J[X,JY] - [JX, JY]
\end{eqnarray}
If the manifold also admits a
metric\footnote{An additional assumption one must make here is that
the metric is hermitian with respect to the complex structure.
However, since we can construct a hermitian metric out of any
metric, we shall not dwell on this minor technicality.}$g_{ab}$,
one can define a two-form $\Omega\equiv - g J $. For a (real)
two-dimensional manifold the two-form vanishes nowhere
and thus serves as a volume-element. The complex structures induces
thereby a natural orientation on the surface, which implies the
following theorem.
\begin{BWtheorem}
A Riemann surface is orientable.
\end{BWtheorem}

\noindent We should add some more notions about mappings between Riemann
surfaces.
\begin{BWdef}
A continuous mapping $f: M \rightarrow N$ between Riemann surfaces
is called {\bf holomorphic} if and only if $\forall \, (U,z)$ on $M$
and $\forall \, (V,\zeta)$ on $N$ the function $\zeta \circ f \circ
z^{-1}:\BWIC \rightarrow \BWIC$ is holomorphic (on the overlapping
patches). A holomorphic mapping which is bijective is called a {\bf
conformal} mapping. A holomorphic mapping into $\BWIC$ is called a
{\bf holomorphic function}. A holomorphic mapping into
$\BWIC\cup\{\infty\}$ is called a {\bf meromorphic function}.
\end{BWdef}

We are now ready to proof the equivalence between an oriented
surface with a conformal metric $[g]$ and a Riemann surface with a
complex structure $J$. We shall show this equivalence in two steps:
\begin{BWprop}
Every metric $g$ on an oriented surface $M$ determines a unique
complex structure on $M$. The complex structure only depends on the
conformal structure of the metric, thus on $[g]$.
\end{BWprop}
\begin{proof}
In a first step we construct an almost complex structure, using the metric and the orientability of the surface:
\begin{eqnarray}
{J^{a}}_b = \sqrt{g} g^{a c} \epsilon_{c b}.
\end{eqnarray}
One can immediately see that ${J^{a}}_b$ squares to $ - \mathbb{1}$ and thus represents an almost complex structure. As the Nijenhuis-tensor vanishes in two dimensions, ${J^{a}}_b$ is also complex structure. If we perform a Weyl-transformation $g \rightarrow e^{2 \omega} g$, the complex structure is unaffected. We have thus shown that the complex structure only depends on the conformal structure of the metric.\qedhere
\end{proof}

\begin{BWprop}
{Every Riemann surface $M$ admits a Riemann metric compatible
with the complex structure of $M$. The metric is unique up to
conformal equivalence.}
\end{BWprop}
\begin{proof}
A basic theorem of Riemann surfaces (see e.g. \cite{Farkas:1980}, section II.5) tells us that we can always find a non-constant meromorphic function $f$ on $M$, which allows us to write the metric as follows (outside poles and points at which $df = 0$),
\begin{eqnarray}
ds^2 = |df |^2.
\end{eqnarray}
As the number of points at which $df=0$ and of poles is finite, we can cover the surface with small coordinate patches $\{ z_i \}$ at those points and find associated functions $\omega_i$, such that $0 \leq \omega_i \leq 1$ and for $ 0< r_1 < r_2 <1$,
\begin{eqnarray}
\text{for}\, \{|z_i|\leq r_1 \}:\,\omega_i = 1\quad \text{and for}\, \{|z_i|\geq r_2 \}:\, \omega_i = 0.
\end{eqnarray}
If we combine everything now, we can write down the following metric,
\begin{eqnarray}
ds^2 = |df|^2 \left( 1- \sum_i \omega_i \right) + \sum_i \omega_i |dz_i|^2.
\end{eqnarray}
A different choice of meromorphic function and/or coordinate patches will only change the metric up to an overall conformal factor. 
\qedhere
\end{proof}

\noindent This nice equivalence between oriented surfaces with conformal
structure and Riemann surfaces allows us to think in terms of a
conformal metric or a complex structure, depending on which point of
view is most helpful. It also explains why bijective holomorphic
mappings are called conformal mappings. An advantage of this
equivalence is that we can use the powerful machinery of complex
surfaces to tackle problems related to the worldsheet of the
string. We shall encounter some of this machinery in the following
sections. Let us conclude this section by introducing complex metrics on
Riemann surfaces.
\begin{BWex}
\label{BWex:RiemannSphereMetric}
Metric on the Riemann Sphere $ds^2 = \left(\frac{2}{1+|z|^2}\right)^2 dz
d\bar z$
\end{BWex}
\begin{BWex}
Metric on the Plane $ds^2 = dz d\bar z$
\end{BWex}
\begin{BWex}
Metric on the Disc $ds^2 = \left(\frac{2}{1-|z|^2}\right)^2 dz d\bar z$ and metric on the Upper Half Plane $ds^2= \left( \frac{2}{Im(z)}\right)^2 dz d\bar z$
\end{BWex}
Since a Riemann surface can be seen as an oriented surface with a
conformal metric $[g]$, it is possible to distinguish three
different types of surfaces, based upon the value of the (constant)
Ricci-scalar $R$. Although it might be unclear a priori that every
metric can be Weyl-transformed into a metric with constant
curvature, this fact actually follows from the uniqueness of the
solution to the Liouville-equation,
\begin{eqnarray}
2 \nabla^2 \omega = R_g - R_{\tilde g} e^{-2\omega},
\end{eqnarray}
where we performed a Weyl-transformation on the metric $g_{ab}$ of
the form $\tilde g_{ab} = e^{-2\omega} g_{ab}$. $R_{\tilde g}$ is
the Ricci-scalar with respect to $\tilde g_{ab}$ and $R_g$ the
Ricci-scalar with respect to $g_{ab}$. Hence, we obtain the
following theorem
\begin{BWdef}
Every (oriented) Riemann surface is conformally equivalent to one
and only one of the following three surface types,
\begin{enumerate}
\item[(1)] {\bf elliptic} surface: a compact, oriented surface for
which $R = +1$
\item[(2)] {\bf parabolic} surface: an oriented surface for which $R =
0$
\item[(3)] {\bf hyperbolic} surface: an oriented surface for which $R =
-1$
\end{enumerate}
\end{BWdef}
In case of a compact Riemann surface (without boundaries and
cross-caps) we can relate the constant curvature to the
Euler-characteristic by virtue of the Gauss-Bonnet theorem and
thereby with the number handles. When there are no handles (g$=0$),
the Ricci-scalar $R_{\tilde g} =1$. A Riemann surface with one
handle (g$=1$), is conformally equivalent to a parabolic surface
($R_{\tilde g} =0$). A Riemann surface for which g$\geq 2$, is
conformally equivalent to a hyperbolic surface ($R_{\tilde g}=-1$).

\subsection{Classifying Riemann Surfaces}
One of the wonderful facts about Riemann surfaces is that they can
be completely classified based upon their covering space
and the fundamental group $\pi_1(M)$. But before we deal with the
general Uniformization Theorem, we phrase the theorem for
simply connected Riemann surfaces, i.e.\ Riemann surfaces for which
$\pi_1(M) = {0}$,
\begin{BWtheorem}
Every simply connected Riemann surface $M$ is conformally equivalent
to one and only one of the following three:\\
\begin{itemize}
\item[(1)] Riemann sphere $\BWIC \cup\{\infty\}$
\item[(2)] Complex plane $\BWIC$
\item[(3)] Complex Disc $D$ or Upper Half plane ${\cal U}$
\end{itemize}
\end{BWtheorem}
It is not difficult to see that the Riemann sphere represents an
elliptic surface, the complex plane a parabolic surface and the disc
a hyperbolic surface. There exist thus a one-to-one correspondence
between an elliptic/parabolic/hyperbolic surface and the Riemann
sphere/Complex plane/Upper Half Plane. It is also clear that these
three Riemann surfaces are indeed distinct and cannot be
conformally transformed into each other. The Riemann sphere is a
compact surface and is therefore topologically different from the
other two Riemann surfaces. Furthermore, we cannot find a conformal
map going from the complex plane to the disc $D$ due to Liouville's
theorem\footnote{Liouville's theorem states that every bounded,
holomorphic function over the complex plane must be constant. For
our purposes we should take a holomorphic function $f:\BWIC
\rightarrow D$, which is clearly bounded and thus constant.}.

Before we can set-up the general Uniformization Theorem we should
again introduce some concepts as the Kleinian group, the Fuchsian
group and the covering surface. We shall gradually introduce the
concept of a Kleinian (Fuchsian) group $G$, but first we review some
useful notions of group theory. Suppose $G$ is a group, then we say
that $G$ acts left ({\bf left action}) on $M$\footnote{We specifiy
our definition for a Riemann surface $M$, but for most definitions
it suffices that $M$ is a set.} if the map,
\begin{eqnarray}
G\times M \rightarrow M : (g,z) \mapsto g \cdot z,
\end{eqnarray}
satisfies the following two axioms,
\begin{enumerate}
\item[(a)] $\forall \, g, h \in G, \forall\, z \in M: (gh)\cdot z =
g \cdot (h\cdot z)$,
\item[(b)] $\forall\, z \in M: e\cdot z = z$ ($e$ is called the identity).
\end{enumerate}
Let us clarify this concept with an example.
\begin{BWex}
Take a Riemann surface $M$ and take $G = Aut(M)$, i.e.\ the group of automorphisms of $M$, 
bijective holomorphisms from $M$ to $M$. Then
$Aut(M)$ performs a left action on $M$. Later on, we shall determine
the automorphism group for some Riemann surfaces. We mention here
already the automorphism groups for the simply connected Riemann
surfaces: $Aut(\BWIC\cup \{\infty \})\cong PSL(2,\BWIC)$,
$Aut(\BWIC)\cong Aff(1,\BWIC)$ and $Aut(D)\cong PSL(2,\BWIR)$. For
the necessary clarifications we refer to \ref{BWs:PropRS}.
\end{BWex}
A group $G$ is said to act {\bf freely} on $M$ if $\forall\, g, h
\in G$ with $g\neq h$, $\forall\, z \in M$ we have $g\cdot z \neq h
\cdot z$. This is equivalent to saying that the only
element of $G$ which maps some point $z \in M$ to itself is the
identity. In symbols, this reads: $\forall\, g\in G$, if $ \exists\,
z\in M$ such that $g\cdot z = z$ then $g=e$. There exist another way
to check if a group acts freely on $M$, via the concept of {\bf
stabilizer group} or {\bf isotropy group $G_z$} at a point $z \in
M$,
\begin{eqnarray}
G_z \equiv \{ g\in G | \, g\cdot z = z \}.
\end{eqnarray}
It is obvious that $G_z$ is a subgroup of $G$. We can now say that a
group $G$ acts freely on $M$ if and only if $G_z = \{e\}$ for every
$z \in M$.

The {\bf orbit} $O_z$ of a point $z\in M$ expresses how $G$ acts on
a point $z$,
\begin{eqnarray}
O_z \equiv \{ g \cdot z | \, g \in G \}
\end{eqnarray}
The orbit $O_z$ of a point $z$ actually links different points of
$M$ to each other by virtue of a group element acting on $z$. It
implies an equivalence relation between two different points $z_1,
z_2$ of $M$,
\begin{eqnarray}
z_1 \sim z_2 \Leftrightarrow \exists\, g \in G: z_2 = g \cdot z_1.
\end{eqnarray}
\begin{BWexerc}
Show that this relation is an equivalence relation.
\end{BWexerc}
Hence, the orbit $O_z$ is an equivalence class for this equivalence
relation. The set of all orbits of (a surface) $M$ under the action
of $G$ is called the {\bf orbit surface $M/G$},
\begin{eqnarray}
M/G \equiv \{ O_z |\, z \in M \}.
\end{eqnarray}
By introducing the concept of (left) action of a group $G$ on a
surface $M$ we managed to divide $M$ in disjoint sets, which union
forms the entire surface $M$.

We go over to Kleinian and
Fuchsian groups. Suppose we consider a subgroup $G$ of the
automorphisms $Aut(M)$ of a Riemann surface $M$. We say that $G$ is
acting {\bf (properly) discontinuously} at $z_0 \in M$ provided
that,
\begin{enumerate}
\item[(1)] $G_{z_0} $ is finite,
\item[(2)] $\exists$ a neighborhood $U_{z_0}$ of $z_0$ such that
$g(U_{z_0}) = U_{z_0},\, \forall \, g \in G_{z_0}$,
\item[(3)] $g(U_{z_0})\cap U_{z_0} = \emptyset,\, \forall\, g \in
G\backslash G_{z_0}$.
\end{enumerate}
We can take all the points $z_0 \in M$ for which $G$ acts properly
discontinuously at those points and call this set $\Omega(G)$ the
{\bf region of discontinuity} of G. If $\Omega(G) \neq \emptyset$ we
call $G$ a {\bf Kleinian group}. From this definition it is
immediately clear that $\Omega(G)$ is an open $G$-invariant subset
of $M$, i.e.\ $G \Omega(G)=\Omega(G)$. If the disc $D \subset \Omega(G)$, $G$
is called a {\bf Fuchsian group}. One can show that every Kleinian
group is a finite (or at least countable) discrete group. Of course
the simplest example of a Kleinian group (and Fuchsian group) is the
trivial group $G=\{e\}$. Let us take a moment to give some
non-trivial examples of Kleinian and Fuchsian groups.
\begin{BWex}
Consider the group generated by the element $z\mapsto z + b$ in
$Aut(\BWIC)$, for a fixed $b \in \BWIC$. The elements of this group,
can be represented by matrices acting on $z$ as
\begin{eqnarray}
\left( \begin{array}{cc} 1 & n b \\ 0 & 1 \end{array}\right) \cdot z \equiv z + n b  , \quad
n\in \BWIZ.
\end{eqnarray}
This group is a freely acting Kleinian group. One can construct a
group homomorphism between this group and $\BWIZ$.
\end{BWex}
\begin{BWex}
Consider the group generated by the element $z\mapsto \lambda z$ in
$Aut({\cal U})$, with a fixed (real) $\lambda \geq 1$. The elements
of this group are represented by matrices of the form,
\begin{eqnarray}
\left(\begin{array}{cc} \lambda^n & 0\\ 0 & 1 \end{array} \right)  \cdot z \equiv \lambda^n z,
\quad n \in \BWIZ.
\end{eqnarray}
This group is a freely acting Kleinian group. This group is also
(group) homomorphic to $\BWIZ$.
\end{BWex}
\begin{BWex}
An example of a Fuchsian group is given by the {\bf modular group}
$PSL(2,\BWIZ)$ on the upper half plane ${\cal U}$. The modular group
exists of matrices of the form,
\begin{eqnarray}
\left(\begin{array}{cc} a & b \\ c & d \end{array} \right)\cdot z \equiv \frac{a z + b}{c z + d},\quad
a,b,c,d \in \BWIZ,\, \text{such that }\, ad-bc = 1.
\end{eqnarray}
Moreover this group acts non-freely on ${\cal U}$, as we will see
later in the discussion of the moduli space for the torus.
\end{BWex}

The Uniformization Theorem gives a classification of Riemann
surfaces based upon their universal covering surfaces and Kleinian
or Fuchsian groups. Therefore, we should first explain what the
universal covering surface of a Riemann surface is.
\begin{BWdef}
A simply connected surface $\hat M$ is called the {\bf universal
covering surface} of a Riemann surface $M$ if and only if $\exists$
a continuous surjective map $ \eta: \hat M \rightarrow M$ such that
$\forall \, z \in M, \exists$ open neighbourhood $N_z$ of $z$ for
which $\eta^{-1}(N_z) = \cup_\alpha  L_\alpha$ and the $\{
L_\alpha\}$ are disjoint, open sets in $\hat M$ that are mapped
homeomorphically onto $N_z$ by $\eta$. $\eta$ is called the {\bf
covering map}.
\end{BWdef}

\begin{BWex}
Recall that the torus can be defined as the complex coordinates on the plane identified under a lattice $L(\omega_1, \omega_2)$, thus $\BWIC / L(\omega_1, \omega_2)$. Now we consider the following projection,
\begin{eqnarray}
\eta : \BWIC \rightarrow \BWIC / L(\omega_1, \omega_2)
\end{eqnarray}
which projects every complex coordinate to its equivalence class. It is thus not difficult to see that the complex plane is the universal covering surface of the torus. 
\end{BWex}

We now have all the basic ingredients to formulate (one of) the most
important theorems for Riemann surfaces,
\begin{BWtheorem}{{\bf Uniformization Theorem}}
Every Riemann surface $M$ is conformally equivalent to $\hat M / G$
where
\begin{eqnarray}
\hat M = \left\{%
\begin{array}{c} \BWIC \cup \{\infty \} \\ \BWIC \\
{\cal U}
         \end{array}%
         \right.
\end{eqnarray}
$G \subset Aut(\hat M)$ acts freely and properly discontinuously
on $\hat M$. Furthermore $G \cong \pi_1(M)$.
\end{BWtheorem}
If the covering surface is $\BWIC \cup \{\infty \}$, or $\BWIC$,
then the group $G$ is Kleinian with maximally two elements (besides
the identity element). If the covering surface is the upper half
plane ${\cal U}$, then the group $G$ is Fuchsian. Hence, the study
of Riemann surfaces is reduced to the study of Kleinian and Fuchsian
groups by virtue of the Uniformization Theorem. For most Riemann
surfaces the covering surface is the upper half plane ${\cal U}$ and
the group $G$ is a non-commutative Fuchsian group.\\
Nonetheless, we can distinguish seven Riemann surfaces for which the
covering group $G$ is abelian. These Riemann surfaces are called
{\bf exceptional Riemann Surfaces} and are characterized by the fact
that $\pi_1(M) \cong \{ e\}, \BWIZ$ or $\BWIZ \oplus \BWIZ$. We
classify them here, based upon their universal covering
surface and Kleinian group $G$, see table \ref{BWref:7ExcRiemann}.
\begin{table}[h]
\begin{center}
\begin{tabular}{|c|c|c|}
\hline \multirow{2}{*}{ $\hat M$} & \multirow{2}{*}{  $G$} &  \multirow{2}{*}{Riemann Surface}\\
& & \\
\hline
\hline
 \multirow{2}{*} {$\BWIC \cup \{\infty\}$} & \multirow{2}{*}{$\{e\}$} &\multirow{2}{*}{ $\BWIC \cup \{\infty\}$}\\
&&\\
\hline
\hline \multirow{4}{*}{$\BWIC$}&\multirow{4}{*}{$\begin{array}{c} \{ e \} \\ \BWIZ \\ \BWIZ \oplus \BWIZ \end{array}$}&\multirow{4}{*}{$\begin{array}{l}\BWIC \\ \BWIC^* \\ T^2 \end{array}$} \\
& & \\
&&\\
&&\\
\hline
\hline 
\multirow{4}{*}{$D$ or ${\cal U} $ } & \multirow{4}{*}{$\begin{array}{c} \{e \} \\ \BWIZ \\ \BWIZ \end{array}$} &\multirow{4}{*}{$\begin{array}{c} {\cal U} \\ D^* \\ D_r \end{array}$}\\
&&\\
&&\\
&&\\
\hline
\end{tabular}
\caption{The Seven Exceptional Riemann Surfaces}\label{BWref:7ExcRiemann}
\end{center}
\end{table}\\
We have not yet encountered some of the Riemann surfaces in this
table, but we will define them now,
\begin{itemize}
\item The punctured plane $\BWIC^* \equiv \BWIC\backslash \{ 0\}$,
\item The punctured disk $ D^* \equiv \{ z\in \BWIC |\, 0<|z|<1 \}$,
\item The cylinder $C_2$ or annulus $D_r \equiv \{ z \in \BWIC |\, r<|z|<1
\}$, for $ 0<r<1$.
\end{itemize}

\section{Properties of Riemann Surfaces \label{BWs:PropRS}}
\subsection{Automorphism Group and CKV's}\label{BWs:CKVs}
We have already noticed that the automorphism group $Aut(M)$ of a
Riemann surface $M$ plays an important role in the classification of
Riemann surfaces. Especially the automorphism groups of the simply
connected Riemann surfaces were crucial. Therefore it is essential
to give an extended discussion of automorphism groups of Riemann
surfaces. Moreover, we will see later on that the automorphism group will play an important role when gauge fixing the Polyakov path integral. The automorphism group of a Riemann surface consists of combinations Diff$\times$Weyl that leave the metric invariant. Gauge fixing the path integral does not eliminate this symmetry group. But we can use these symmetries to fix the insertion points of the vertex operators.

For most Riemann surfaces $Aut(M)$ is a discrete group and the symmetry group does not cause any additional problems.
Only the exceptional Riemann surfaces have a continuous automorphism
group. And for a continuous $Aut(M)$ it is possible to find conformal
(bijective holomorphic) transformations of $M$ which leave the
metric (complex structure) on $M$ Weyl-invariant. 
The infinitesimal form of the transformations connected to the identity are known as {\bf Conformal Killing
Vectors (CKV)} and they generate the {\bf Conformal Killing Group
(CKG)}. One can define a CKV in a mathematical way as a Killing vector (symmetry of the metric) which leaves the metric in the same conformal class,
\begin{eqnarray}
{\cal L}_\xi g_{ab}(x) = \phi(x)\, g_{ab} (x),\label{BWeqconfkill}
\end{eqnarray}
for some (in general) non-constant function $\phi(x)$. When we use this definition for a Riemann surface with its Hermitian metric, then it follows immediately that the CKV are holomorphic functions, i.e.\ $\partial_{\bar z} \xi^{z}=0=\partial_{z} \xi^{\bar z}$.
It is the connected component $Aut_0(M)$ that raises additional concerns when gauge fixing the path integral.   

We start by discussing $Aut(M)$ of simply connected Riemann surfaces $M$, after which we comment on the
remaining exceptional Riemann surfaces. We also discuss the CKV's of the simply connected Riemann surfaces and the torus using their infinitesimal form. It is left as an exercise to check that the given expressions for the CKV's satisfy eq.~(\ref{BWeqconfkill}). \\

\subsubsection{The Riemann Sphere}
To determine the automorphisms of the Riemann Sphere $\BWIC\cup
\{\infty \}$ we should study bijective meromorphic maps. A
meromorphic function with domain $\BWIC\cup \{\infty \}$ can be
written as the ratio of two polynomials\footnote{For a proof of this
statement we refer to R. Miranda, p.30. The proof is build upon the
statement that a non-zero meromorphic function on a compact Riemann
surface has a finite number of zeroes and poles.},
\begin{eqnarray}
f(z) = \frac{P(z)}{Q(z)}.
\end{eqnarray}
The constraint that $f(z)$ is bijective, implies that $f(z)$ has
only simple poles and zeroes\footnote{If there would be a pole or
zero of order higher than one, then $f(z)$ would be mutli-valued,
loosing its injective character.}. Therefore the most general
automorphism of $\BWIC\cup \{\infty \}$ can be written as,
\begin{eqnarray}
f(z) = \frac{az+b}{cz+d},\quad a,b,c,d \in \BWIC.
\end{eqnarray}
We can define this type of transformation using the left action of
$Aut(\BWIC\cup \{\infty \})$ on $\BWIC\cup \{\infty \}$,
\begin{eqnarray}
\left(\begin{array}{cc} a & b \\ c & d \end{array} \right) \cdot z
\equiv \frac{az+b}{cz+d}. \label{BWeq:Actingz}
\end{eqnarray}
The question for which values of $a, b, c$ and $d$ the function
$f(z)$ would be invertible, reduces to asking the question for which
values of $a, b, c$ and $d$ the matrix would be invertible (i.e.\ the
matrix would be an element of $GL(2,\BWIC)$). The answer to that
last question is that $ad-bc \neq 0$. Transformations of this type
are called {\bf M\"{o}bius transformations}. The most general
M\"{o}bius transformation is a composition of a translation, a
dilation, a rotation and a complex inversion (not in that order).
\begin{BWexerc}
Find a general expression for the inverse of a M\"{o}bius
transformation.
\end{BWexerc}

Since an overall scaling of $a, b, c$ and $d$ does not change the
transformation, we can rewrite\footnote{We can always multiply a
matrix in $GL(2,\BWIC)$ by a complex constant such that the metric
has a positive determinant.} the invertibility constraint as $ad-bc
= 1$. If we multiply $a, b, c$ and $d$ all by $-1$, then the
determinant of the transformation remains $1$. This means that we
can identify a transformation $(a,b,c,d)$ with $(-a,-b,-c,-d)$.
Hence, we can conclude that,
\begin{eqnarray}
Aut(\BWIC \cup \{\infty \} ) \cong PSL(2,\BWIC) \equiv
SL(2,\BWIC)/\{\pm e\}.
\end{eqnarray}
We notice that this group is a 3-dimensional complex (or
6-dimensional real) Lie-group.\\
Now, we can also have a look at the CKV's for the sphere. As the CKV's are those elements of the automorphism group that are connected to the identity, we should start with elements in $SL(2,\BWIC)$. In case we want to look at the infinitesimal form of the transformation, we write an element using the exponential map between the Lie-group and the Lie-algebra, i.e.\ $e^A$. The condition that the determinant of a matrix $\BWe^{A}$ in $SL(2,\BWIC)$ is equal to 1 leads to the condition that $A$ is traceless, $Tr(A)=0$. We can write the infinitesimal form of the M\"{o}bius transformation as,
\begin{eqnarray}
e^A = \mathbb{1} + A  + \cdots =   \left(\begin{array}{cc} 1+ \alpha & \beta \\ \gamma &  1- \alpha \end{array}\right) + \cdots , \label{BWeq:infmobius}
\end{eqnarray}
with $\alpha$, $\beta$, $\gamma \in \BWIC$. We can now see that an infinitesimal M\"{o}bius-transformation is given by $z \rightarrow z + \delta z$ with,
\begin{eqnarray}
z &\rightarrow& z + \delta z , \quad \delta z = \beta + 2 \alpha z - \gamma z^2, \label{BWeqconfkillsphere1}\\
\bar{z} &\rightarrow& \bar{z} + \delta \bar{z} , \quad \delta \bar{z} = \bar{\beta} + 2 \bar{\alpha} \bar{z} - \bar{\gamma} {\bar{z}}^2. \label{BWeqconfkillsphere2}
\end{eqnarray}
In order to obtain this expression we made a power series expansion of the transformation (\ref{BWeq:infmobius}) in $z$ and keep only those terms that are linear in the parameters $\alpha$, $\beta$, $\gamma$. This implies that there are no terms of order $z^3$ or higher in the transformation. Another way to see this is by requiring that the CKV's are defined globally. This implies that they should also be defined at the patch where $\infty$ lies. If we make the transformation $u=1/z$, we see that $\delta u = - 1/z^2 \delta z$ and that the transformation in the $u$-patch is holomorphic if $\delta z$ does not grow faster than $z^2$ for $z\rightarrow \infty$. 

\subsubsection{The Complex Plane}
The automorphism group $Aut(\BWIC)$ of the complex plane are those
M\"{o}bius transformations that leave $\infty$ fixed, i.e.\ those transformations (\ref{BWeq:Actingz}) with
$c=0$. We can rescale $a$ and $b$ such that $d=1$. Thus the
transformations are of the form,
\begin{eqnarray}
f(z) = a z + b = \left(\begin{array}{cc} a & b \\ 0 & 1 \end{array}
\right) \cdot z, \quad a \in \BWIC\backslash\{0\}.
\end{eqnarray}
This group is nothing else than the group of affine transformations of the
plane,
\begin{eqnarray}
Aut(\BWIC) \cong Aff(1,\BWIC).
\end{eqnarray}
This group is a 2-dimensional complex (or 4-dimensional real)
Lie-group.
To find the CKV's we can follow the same pattern as for the Riemann Sphere. In the end we will find the same type of transformations as eq.~(\ref{BWeqconfkillsphere1}) and (\ref{BWeqconfkillsphere2}) with $\gamma = 0$.

\subsubsection{The Upper Half Plane}
Now we are looking for M\"{o}bius transformations that leave the
upper half plane ${\cal U}$ invariant. These transformations should
also map the boundary $\BWIR \cup \{\infty \}$ to itself. This
implies that $a, b, c$ and $d$ $\in \BWIR$, and an overall rescaling
yields the constraint $ad -bc = \pm 1$. In this case we must make a
difference between matrices (in $GL(2,\BWIR)$) with a positive
determinant and those with a negative determinant. We can not
multiply a matrix with negative determinant by a real number to
obtain a matrix with positive determinant. Since $Aut(\cal U)$
should contain the identity, we must choose those matrices with a
positive determinant ($ad-bc = 1$), thus $SL(2,\BWIR)$. We can again
identify matrices with entries $(a,b,c,d)$ and $(-a,-b,-c,-d)$. The
group of automorphisms of the upper half plane is thus,
\begin{eqnarray}
Aut({\cal U}) \cong PSL(2,\BWIR) \equiv SL(2,\BWIR)/\{\pm e\}.
\end{eqnarray}
This group is a 3-dimensional real Lie-group.
The analysis for the CKV's is fully analogous to the one of the Riemann Sphere. We find transformations of the form eq.~(\ref{BWeqconfkillsphere1}) and (\ref{BWeqconfkillsphere2}) with $\alpha$, $\beta$ and $\gamma$ $\in$ $\BWIR$.

\subsubsection{The Remaining Exceptional Riemann Surfaces}
The Uniformization Theorem tells us that the universal covering
surface of a Riemann surface is one of the simply connected Riemann
surfaces. We  use the knowledge of the automorphism group
$Aut(\hat M)$ of the covering surface $\hat M$ and the covering
group G to determine the automorphism group $Aut(M)$ of the
remaining four exceptional Riemann surface at hand. Before we can
pose the theorem that will help us to determine $Aut(M)$, we need to
refresh some more concepts. The normalizer of the Kleinian group $G$
in $Aut(\hat M)$ are those elements of $Aut(\hat M)$ which commute
with all elements of $G$,
\begin{eqnarray}
N(G) \equiv \{ h\in Aut(\hat M) |\, h G h^{-1} = G  \}.
\end{eqnarray}
We have the following inclusions $G \subset N(G) \subset Aut(\hat
M)$. Since $G$ is a group, also $N(G)$ is a group. Moreover,  $G$ is the largest normal subgroup of $N(G)$, by
virtue of the definition of a normal subgroup and of a normalizer.
Since $G$ is a normal subgroup of $N(G)$ we can construct the
quotient group $N(G)/G$, which is precisely isomorphic to $Aut(M)$,
\begin{BWtheorem}
If $M$ is a Riemann surface, $\hat M$ its universal covering surface
and $G$ the covering group, then
\begin{eqnarray}
Aut(M) \cong N(G)/G,
\end{eqnarray}
where $N(G)$ is the normalizer of $G$ in $Aut(\hat M)$.
\end{BWtheorem}

\noindent This theorem is very useful to determine the automorphism group of the remaining four exceptional Riemann
surfaces.\\
\\
\noindent {\bf (1) The Torus $T^2$}\\
We have seen above that the Torus can be seen as the Complex Plane $\BWIC$ modded out by a lattice $L(\omega_1, \omega_2)$. Using the above theorem we can determine those elements of $Aut(\BWIC)$ which  leave the lattice invariant. The lattice itself can be generated by the two elements,
\begin{eqnarray}
\left(\begin{array}{cc}
1 & \omega_1 \\
0 & 1
\end{array} \right) \quad \text{and} \quad
\left(\begin{array}{cc}
1 & \omega_2 \\
0 & 1
\end{array} \right).
\end{eqnarray}
One can then see that under conjugation of an element of $Aut(\BWIC)$ a basis vector of the lattice transforms as,
\begin{eqnarray}
\left(\begin{array}{cc}
a & b \\
0 & 1
\end{array} \right) \left(\begin{array}{cc}
1 & \omega_i \\
0 & 1
\end{array} \right) \left(\begin{array}{cc}
a^{-1} & a^{-1} b \\
0 & 1
\end{array} \right) = \left(\begin{array}{cc}
1 & a \omega_i \\
0 & 1
\end{array} \right) .
\end{eqnarray}
The elements of the normalizer of the lattice are those elements which leave the lattice invariant. This means that the transformed lattice vectors should be linear combinations of the old lattice,
\begin{eqnarray}
\left(\begin{array}{c} a \omega_1\\  a \omega_2  \end{array}\right)
 = \left(\begin{array}{cc}
\alpha & \beta\\
\gamma & \delta
\end{array} \right)  \left( \begin{array}{c}  \omega_1 \\  \omega_2  \end{array}\right),\label{BWeqcondaut}
\end{eqnarray}
with $ \alpha, \beta, \gamma, \delta \in
\BWIZ$, and $\alpha \delta - \beta \gamma= 1$. This condition is for instance satisfied for translations ($a=1$, $z\rightarrow z +b$), in which case $\alpha=\delta=1$ and $\beta=\gamma =0$. Indeed, the two-dimensional translation group ${\cal T}^2$ of $\BWIC$ form the identity component $Aut_0(T^2)$. It is the identity component $Aut_0(T^2)$ which should be used to find the CKV's. We thus find that the CKV's for the torus are infinitesimal translations,
\begin{eqnarray}
z & \rightarrow & z + \epsilon,\\
\bar z & \rightarrow& \bar z+ \bar \epsilon.
\end{eqnarray}
The condition eq.~(\ref{BWeqcondaut}) is also satisfied when $a=-1$, in which case $\alpha=\delta=-1$ and $\beta=\gamma =0$. This means that we should consider $Aut(T^2)$ as a $\BWIZ_2$-extension of $Aut_0(T^2)$. For certain specific values of $\omega_1$ and $\omega_2$ there are even more possibilities for $a$, which lead to $\BWIZ_4$- or $\BWIZ_6$- extensions\footnote{When $(\omega_1, \omega_2) = (1, i)$ also $a=\pm i$ is an allowed transformation, leading to a $\BWIZ_4$-extension. When $(\omega_1, \omega_2)= (1, e^{ i 2 \pi /3})$ we notice that $a = e^{\pm i 2  \pi /3}$ and $a=e^{\pm i \pi/3}$ are four other allowed transformations, yielding a $\BWIZ_6$-extension.}.
\\
\\
\noindent {\bf (2) The Punctured Plane $\BWIC^\times$}\\
The universal covering surface is the complex plane $\BWIC$ with $Aut(\BWIC) = Aff(1,\BWIC)$. We are thus looking for affine transformation of the complex plane which leave $0$ fixed, i.e.\ $z \rightarrow a z$ with $a\in \BWIC^\times$. We also have to include the discrete transformation $z\rightarrow 1/z$ interchanging $0$ and $\infty$. We can conclude: $Aut_0(\BWIC^\times) = \BWIC^\times$ and $Aut(\BWIC^\times)/ Aut_0(\BWIC^\times) \cong \BWIZ_2$.\\
\\
\noindent {\bf (3) The Punctured Disk $D^\times$}\\
The elements of $ PSL(2,\BWIR)$ that leave the origin or the Punctured Disk fixed are given by $z \rightarrow e^{i \phi} z$, with $\phi \in \BWIR$ and there are no discrete transformations. So we find: $Aut(D^\times)=Aut_0(D^\times)\cong U(1)$.
\\
\\
\noindent {\bf (4) The Annulus $D_r$}\\
Also for the Annulus we find that the continuous automorphisms are of the form $z \rightarrow e^{i \phi} z$ with $\phi \in \BWIR$ and thus that $Aut_0 (C_2) \cong U(1)$. However, there exists a discrete automorphism of the form $z\rightarrow r/z$ which inverts the Annulus and swaps the inner and outer boundary. This means that the full automorphism group is given by a $\BWIZ_2$ extension of $U(1)$.

\subsection{Moduli space}\label{BWss:moduli_space}
The Uniformization Theorem tells us that every Riemann surface can
be seen as an orbit surface with respect to its covering Riemann
surface and a Kleinian group. Suppose now that we have two
topologically equivalent Riemann surfaces, we could ask ourselves
the question: ''which are the necessary and sufficient conditions
for which the two Riemann surfaces are conformally (in)equivalent?"
Let us elaborate this question a little bit further. Given two
topologically equivalent Riemann surfaces, we can place a metric on
both of the surfaces. However the metrics are not necessarily
conformally equivalent. The difference between two conformally
inequivalent metrics can be expressed in parameters, which can not
be eliminated by a coordinate- or Weyl-transformation. To clarify
these statements, we can take a look at the torus $T^2$ defined by
the equivalence relation $z \sim z + \lambda$, where $\lambda \in
\BWIZ + i \BWIZ$, and with a flat metric. Next, we consider a
different torus, where the periodicity and the metric are not yet
fixed. We can perform a Weyl-transformation on the metric to obtain
a flat metric, after which we perform a coordinate transformation to
bring the metric in the unit form. But, in general, we will not
obtain the same periodicity relation of our first torus. The
periodicity now reads,
\begin{eqnarray}
\tilde z\sim \tilde z + m + n \tau,\, \quad \tau \in \BWIC,\, m, n
\in \BWIZ.
\end{eqnarray}
We obtain the same periodicity for the choice $\tau = i$, but all other
''choices" for $\tau$ correspond to different tori, see also figure \ref{BWfig:Two_Tori}\footnote{The
discussion given here corresponds to fixing the metric as the unit
metric and change the periodicity conditions via the parameter
$\tau$. One could also choose to keep the periodicity conditions
fixed. In that case the parameter $\tau$ arises in the metric and
thus the metric is different for two conformally inequivalent tori.}.
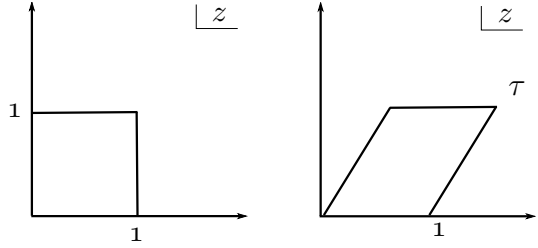
\begin{figure}[ht!]
\centering
\begin{picture}(0,0)
 \put(190,55){\large $\tau$}
\end{picture}
\includegraphics[width=.5\textwidth]{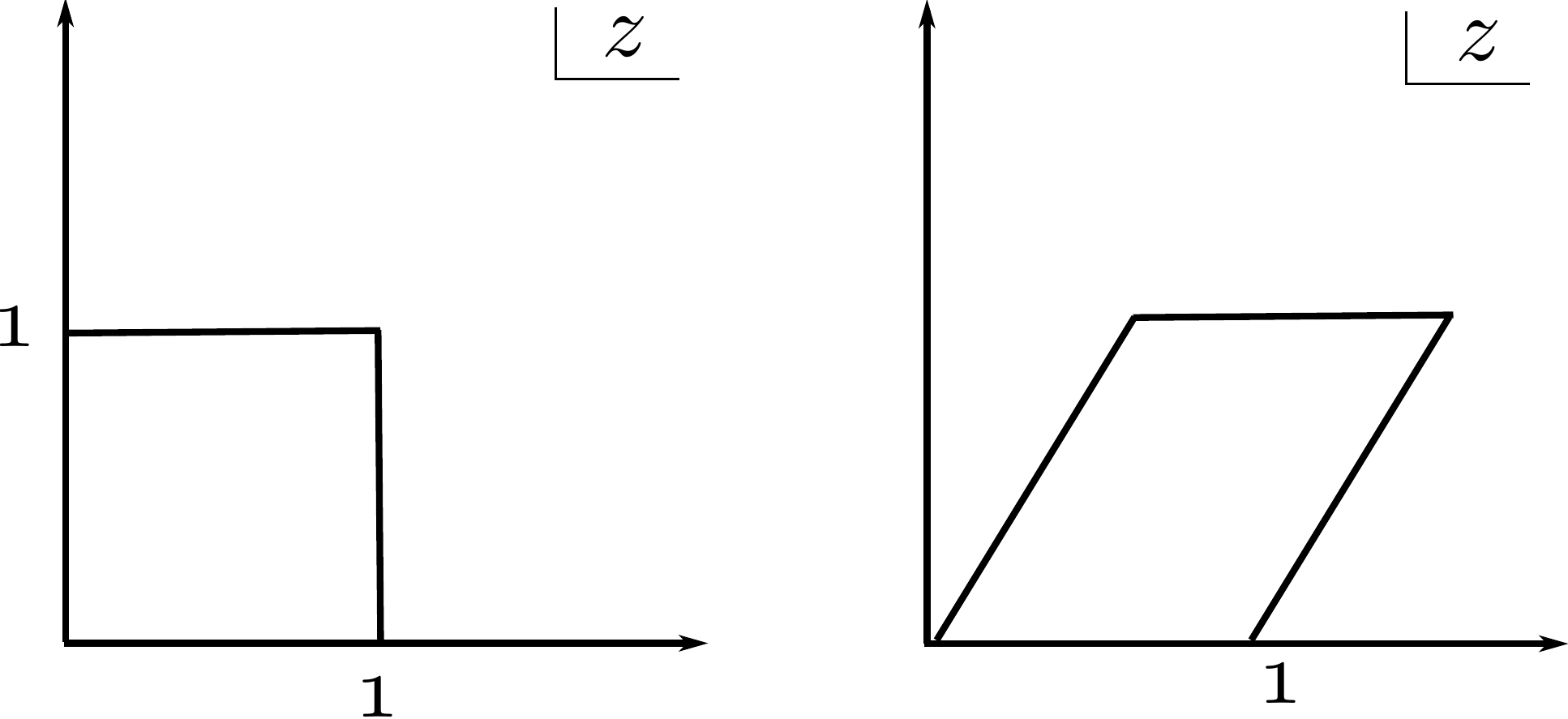} 
\caption{Two different tori. Left: the identification on the complex plane for the unit torus, $z\sim z + m + \BWi n,\,m, n
\in \BWIZ.$ Right: identification with complex parameter $\tau$ as $\tilde z \sim \tilde z + m + n\tau,\, m, n
\in \BWIZ$. \label{BWfig:Two_Tori}}
\end{figure}

The parameters which represent inequivalent Riemann surfaces, as $\tau$ in case of the torus, are
called {\bf moduli}. The space of all (conformally) inequivalent
metrics on a Riemann surface $M$ is called the {\bf moduli space
${\mathfrak M}_g$}. The moduli space is basically the space of metrics
on $M$ modded out by the Diff($M$)$\times$Weyl-group. At first sight
this space appears to be nothing but a set of equivalence classes
$[g]$. But we will argue later that the moduli space is a quotient
space of a complex manifold, which can be parameterized by the
moduli. The dimension\footnote{For a proof we refer to chapter 5 of \cite{Dijkgraaf:1997ip}.}
of the moduli space depends on the genus g of the Riemann surface,
\begin{eqnarray}
dim_\BWIC\, {\mathfrak M}_{\rm g} = \left\{%
\begin{array}{ccc}
    0&\quad & {\rm g} = 0 \\
    1&\quad & {\rm g} = 1 \\
    3{\rm g}-3&\quad & {\rm g} \geq 2. \\
\end{array}%
\right.
\end{eqnarray}

\noindent We shall now discuss the moduli spaces for Riemann surfaces in
detail.

\subsubsection{The Riemann Sphere or genus 0}
We have seen that $Aut(\BWIC \cup \{\infty \}) = PSL(2,\BWIC)$.
Since every element of this group fixes at least one point of the
Riemann sphere, it is impossible to find a non-trivial discrete
subgroup acting freely on the Riemann sphere. Hence, if the
universal covering surface is $\BWIC \cup \{\infty \}$, the Kleinian
group is trivial and the only Riemann surface with $\BWIC \cup
\{\infty \}$ as its universal covering surface is the Riemann sphere
itself. Any surface of genus g $=0$ is conformally equivalent to the
Riemann sphere and thus ${\mathfrak M}_{\BWIC \cup \{\infty \}}={0}$. 

\subsubsection{The Torus or genus 1}
Since we can not obtain other Riemann surfaces using the Riemann
sphere as the universal covering surface, we should focus on the two
other possible covering surfaces, i.e.\ $\BWIC$ or ${\cal U}$. In
this section we shall assume that $\BWIC$ is the universal covering
surface. We expand on the earier example of the torus. The torus can indeed be seen as the Riemann surface with
universal covering surface $\BWIC$ and a covering group $G$
generated by two elements,
\begin{eqnarray}
z&\mapsto& z + \omega_1, \quad \omega_1 \in \BWIC,\\
z&\mapsto& z + \omega_2, \quad \omega_2 \in \BWIC.
\end{eqnarray}
Without loss of
generality we can choose the basis vectors $\omega_1$ and $\omega_2$ such that\footnote{If $Im(\omega_2/\omega_1)=0$, the group $G$ will no be a discrete group.}
$Im(\omega_2/\omega_1)>0$. Next, we rescale the coordinate $z\mapsto
z/\omega_1$ and thus obtain the transformations,
\begin{eqnarray}
z &\mapsto& z + 1, \\
z &\mapsto& z + \tau, \quad \tau \in {\cal U},
\end{eqnarray}
where we have introduced $\tau \equiv \omega_2/\omega_1$. Since this
group $G$ acts discontinuously on $\BWIC$, we obtain the following
equivalence relation,
\begin{eqnarray}
z\cong z + m + n \tau, \quad m, n \in \BWIZ.
\end{eqnarray}
But this equivalence relation reveals that $\tau$ is still not
unambiguously defined. We can still perform the transformation $\tau
\mapsto \tau + 1$, provided that $(m, n) \mapsto (m-n, n)$. Moreover
we can also allow $\tau \mapsto -1/\tau$, if we send $z\mapsto \tau
z$ and $(m, n) \mapsto (n, -m) $. These transformation define a
group acting on the upper half plane ${\cal U}$, generated by the
elements,
\begin{eqnarray}
T: \tau &\mapsto& \tau + 1,\\
S: \tau &\mapsto& -1/\tau,
\end{eqnarray}
satisfying the relations $S^2 = (ST)^3 = e$. By iterating these
transformation we find the group that leaves $\tau$ invariant,
\begin{eqnarray}
\tau \mapsto \frac{a\tau + b }{c\tau + d},\quad a, b, c, d \in
\BWIZ\quad {\rm such\, that}\,\, ad-bc = 1.
\end{eqnarray}
Reversing the signs of $a, b, c, d$ does not change the
transformation, therefore we can identify transformations with
entries $(a,b,c,d)$ and $(-a,-b,-c,-d)$. Thus the group of
transformations which leave $\tau$ invariant is $PSL(2,\BWIZ) \equiv
SL(2,\BWIZ)/\{\pm e\}$ and is called the {\bf modular group}.
Elements of the modular group are called {\bf modular
transformations}. Hence, it is clear that two parameters $\tau$,
$\tau '$ define the same torus if they are related through a
modular transformation.  The moduli space ${\mathfrak M}_{T^2}$ for the
torus is given by,
\begin{eqnarray}
{\mathfrak M}_{T^2} = {\cal U}/PSL(2,\BWIZ).
\end{eqnarray}
$\tau$ represents the modulus of the torus. $PSL(2,\BWIZ)$ is a group acting on the $ {\cal U}$ and therefore we can give a representation of the orbits under the group action. This representation is called the {\bf Fundamental domain} for the action of $PSL(2,\BWIZ)$. The purpose of the Fundamental Domain is to depict only one representative for every orbit. In figure \ref{BWref:Funddom} we give the Fundamental domain for the moduli space of the torus $T^2$.
\begin{figure}[h]
\begin{center}
\includegraphics[totalheight=0.2\textheight,width=0.6\textwidth]{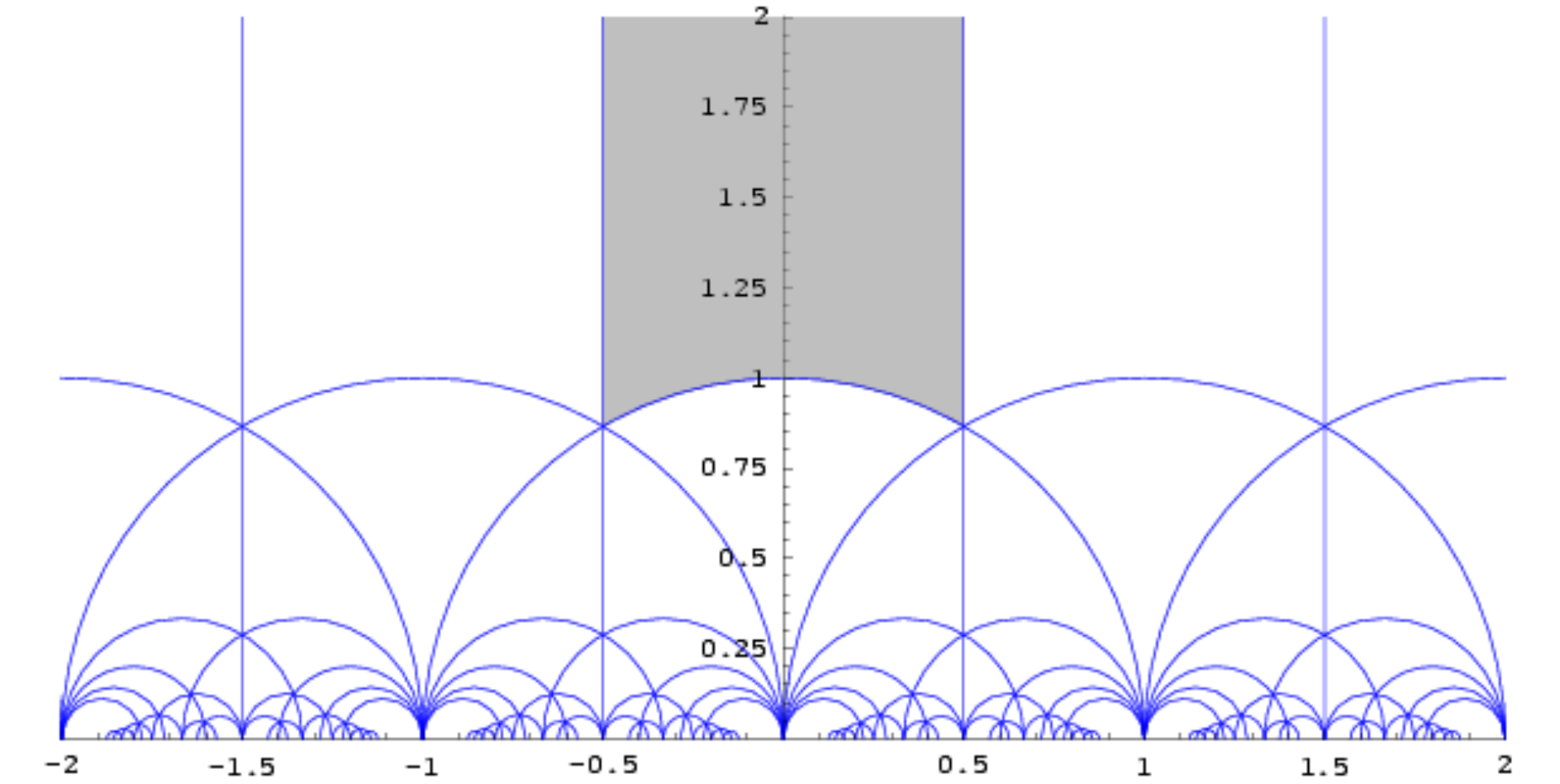}
\caption{Fundamental Domain for the moduli space of the torus $T^2$.\\ Source: \url{http://en.wikipedia.org/wiki/Fundamental_domain}}
\label{BWref:Funddom}
\end{center}
\end{figure}

\subsubsection{Higher Genus and Teichm\"{u}ller space}
The example of the torus teaches us that the moduli space can be
seen as the quotient space of a covering space (${\cal U}$) modded
out by the action of a discrete group ($PSL(2,\BWIZ)$). This
quotient space is not a manifold, since $PSL(2,\BWIZ)$ does not act
freely on ${\cal U}$. There are points $\tau$ (i.e.\ $i$ and $e^{2\pi
i /3}$) for which the stabilizer 
$G_\tau \neq \{ e\}$, which are called fixed points or {\bf orbifold
singularities}. The moduli space of
the torus is therefore an example of an {\bf orbifold space}.\\
We can repeat this construction for Riemann surfaces $M$ with genus
$> 1$. To this end, we look at those diffeomorphism of $M$ that are
homotopic to the identity map,
\begin{eqnarray}
{\rm Diff}_0(M)\equiv \{ f\in {\rm Diff}(M)|\, f\sim id_M \},
\end{eqnarray}
and quotient this normal subgroup out of Diff$^+(M)$, the
diffeomorphisms preserving the orientation of the surface. For a
compact Riemann surface of genus g this discrete quotient group is
called the {\bf mapping class group (MCG)} or {\bf modular group
$\Gamma_g$} of genus g,
\begin{eqnarray}
\Gamma_{\rm g} \equiv {\rm Diff}^+(M) / {\rm Diff}_0(M).
\end{eqnarray}
The elements of $\Gamma_{\rm g}$ can be thought of as
diffeomorphisms not continuously connected to the identity, the
so-called global diffeomorphisms. As the name already suggests, this
will be the group performing the action on the covering space of the
moduli space. Next step is to obtain the covering space of the
moduli space. As for the case of the torus, the covering space of
the moduli space will consist of classes of conformally inequivalent
metrics. This correspond to taking the following quotient space,
\begin{eqnarray}
{\mathfrak T}_{\rm g} \equiv \frac{{\cal M}_g}{ {\rm Weyl}(M)\times {\rm
Diff}_0(M)},
\end{eqnarray}
where ${\cal M}_g$ is the space of complex metrics and Weyl($M$) the
space of Weyl-transformations. The {\bf Teichm\"{u}ller space}
${\mathfrak T}_{\rm g}$ represents a finite dimensional, simply
connected manifold with the same dimension as the moduli space. The
Teichm\"{u}ller space is therefore a complex manifold, which can be
parameterized by complex Teichm\"{u}ller parameters. The group
$\Gamma_g$ acts on these Teichm\"{u}ller parameters and for those
parameters where $\Gamma_g$ does not act freely, we encounter
orbifold singularities. The moduli space consists of those
Teichm\"{u}ller parameters which can not be identified under the
action of $\Gamma_{\rm g}$,
\begin{eqnarray}
\mathfrak{M}_{\rm g}={\mathfrak T}_{\rm g} / \Gamma_{\rm g}.
\end{eqnarray}

\chapter{Gauge fixing of the Polyakov path integral\label{BWc:Fixing}}

\emph{At this moment, the string S-matrix \eqref{eq:Intro_Smatrix1} is not well-defined. In the path integral, we would like to count only over physically inequivalent configurations. However, the Polyakov action is invariant under an infinite-dimensional symmetry-group we treated in the previous chapter, namely diffeomorphisms and Weyl transformations on the world sheet. Therefore, we correct the expression for the scattering amplitudes to:}
\begin{equation}
 S_{j_1\ldots j_n} (k_1\ldots k_n) = \sum_{topologies} g_s^{-\chi}\int \frac{\BWD X \BWD g}{\cal N} \BWe^ {-S_P(X,g)} \prod_{i=1}^n \int \BWd^2 \sigma_i \sqrt{g(\sigma_i)}{\cal V}_{j_i}(k_i,\sigma_i)\,,\label{eq:Smatrix}
\end{equation}
\emph{where we use the form of the vertex operators as integrals over the worldsheet, see eq.\ \eqref{BWeq:VertexOp} and $\cal N$ is a normalization constant to be determined later. We already guess that the appropriate choice for this constant will be the volume of the gauge group, and as such infinite. Indeed, it is exactly this infinite overcounting due to the gauge invariance of the integrand in \eqref{eq:Smatrix} that we would like to factor out.
}

\section{Main idea -- example}

Gauge equivalent configurations give the same contribution to the path integral \eqref{eq:Smatrix}. Since gauge equivalent cofigurations describe the same physical situation, we should count each configuration only once. As anticipated in the introduction above, we do this by factoring (and dividing) out the volume of the gauge group. We try to write the path integral as
\begin{equation}
 \int \BWD X \BWD g (\ldots) = \int (\BWD \rm gauge)( \BWD \rm physical) (\ldots)\,,
\end{equation}
with a contribution over physical states, factoring out the gauge dependence. 

Let us make the way in which we tackle the problem clear by means of an example. Consider the Gaussian integral:
\begin{equation}
 I = \int \BWd x \BWd y\, \BWe^{-\frac12 (x-y)^2}\,.
\end{equation}
on the two-dimensional real plane $\mathbb{R}^2$.
As is the case with the string path integral, this integral is infinite due to an overcounting associated with a gauge invariance of the integrand. In this case the gauge invariance is given by (local) translations in the $x\-y$-plane in the ``North-East/South-West''-direction:
\begin{equation}
\left\{
\begin{array}{lcl}
  x&\to&x +a(x,y)\\
 y&\to& y + a(x,y)
 \end{array}
\right.
\end{equation}
For concreteness, we will call this translational gauge group $G$. We can regularize the above integral by factoring out the volume of the gauge group:
\begin{equation}
I = {\text{Vol}(G)} I'
\end{equation}
with
 $\text{Vol}(G) = \int da$
the (infinite) volume of the gauge group. $I'$ then represents the regularized value of the integral $I$. In the regularized integral $I'$,  every gauge invariant configuration is counted exactly once.

\begin{figure}[b!]
\begin{center}
\unitlength .9mm
\begin{picture}(100,70)(10,0)
\linethickness{.3mm}
%
\linethickness{.4mm}
\put(60,0){\vector(0,1){60}}
\put(63,57){$y$}
\put(0,20){\vector(1,0){120}}
\put(118,23){$x$}
%
\linethickness{.3mm}
%
\put(0,0){\vector(1,1){20}}
\put(0,30){\vector(1,1){20}}
\put(30,0){\vector(1,1){20}}
\put(60,0){\vector(1,1){20}}
\put(90,0){\vector(1,1){20}}
%
\put(18,18){\vector(1,1){20}}
\put(18,48){\line(1,1){12}}
\put(48,18){\vector(1,1){20}}
\put(78,18){\vector(1,1){20}}
\put(108,18){\line(1,1){12}}
%
\put(36,36){\line(1,1){24}}
\put(66,36){\line(1,1){24}}
\put(96,36){\line(1,1){24}}
%
\linethickness{.6mm}
\put(80,0){\textcolor{blue}{\line(-1,1){60}}}
\put(80,0){\textcolor{blue}{\line(-1,1){60}}}
\put(80,0){\textcolor{blue}{\line(-1,1){60}}}
\put(26,55){\textcolor{blue}{$(\hat x,\hat y) = (m,-m)$}}
%
\linethickness{.3mm}
\put(80,0){\textcolor{red}{\line(0,1){60}}}
\put(83,55){\textcolor{red}{$(\hat x,\hat y) = (x_0,m)$}}
%
\linethickness{.3mm}
\put(80,0){\textcolor{green}{\line(1,1){40}}}
\put(95,40){\textcolor{green}{$(\hat x,\hat y) = (m,m)$}}
\end{picture}
\caption{Several gauge choices for the Gaussian integral $I'$, see eq. \eqref{BWeq:xhatyhat}. The black arrows represent the translational flow of the gauge group. The colored lines represent several choices of gauge slices. Note that a gauge slice does not have to be orthogonal to the action of the gauge group (red vs. blue slice). However, the representative slicing has to span the entire space of configurations ($\mathbb{R}^2$) when acting on it with the gauge group: the green line parametrizes a bad choice of gauge slice in this respect. \label{BWfig:GaugeSlicesEx}}
\end{center}
\end{figure}
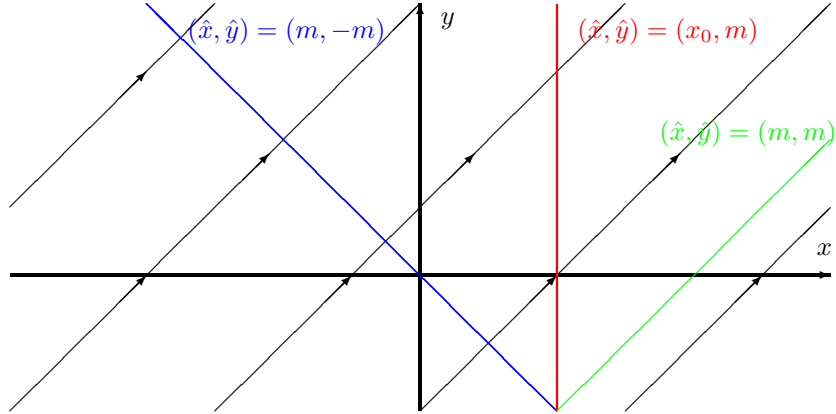

To obtain this factorization, we perform a change of variables. Therefore, we choose a coset representative $(\hat x, \hat y)$ as an element of $\mathbb{R}^2/G$, i.e.\ the real plane modulo the gauge group $G$. We call the choice of coset representative a \textit{gauge slice}. The quotient group $\mathbb{R}^2/G$ is one-dimensional and we can parametrize its representative as $(\hat x(m), \hat y(m))$, with $m \in \mathbb{R}$. We call the parameter $m$ a modulus. It labels gauge inequivalent coordinates on the real plane.

In order for the representative element $(\hat x(m), \hat y(m))$ to be a good choice, we must demand that any $(x,y)\in \mathbb{R}^2$ can be related to some $(\hat x(m), \hat y(m))$ by an element of $G$:
\begin{equation}
\begin{array}{lccc}
 &{\rm Physical~ change}&  &{\rm Gauge~ change}\\
  x =& \overbrace{\qquad\hat x(m)\qquad }&+& \overbrace{\,\quad \, \,a(x,y)\,\quad\,\,} \\
 y =& \hat y(t) &+& ~\qquad a(x,y)\qquad \,, 
\end{array}
\label{BWeq:xhatyhat}
\end{equation}
for a function $a$. It is important that the transformation $(x,y) \to (m,a)$ is one to one. In that case, we can write the integration measure as
\begin{equation}
 \BWd x \BWd y = J  \BWd t \BWd a\,,
\end{equation}
where $J$ is the Jacobian of the transformation. It can be evaluated using  \eqref{BWeq:xhatyhat}:
\begin{equation}
 J = \left|\frac{\partial (x,y)}{\partial (m,a)}\right| = \left|\det \begin{pmatrix} \frac{\partial \hat x }{\partial t}& \frac{\partial \hat y }{\partial t}\\ 1 &1 \end{pmatrix}\right|=\left| \frac{\partial \hat x }{\partial m}-\frac{\partial \hat y }{\partial m}\right|
\end{equation}
Notice that the Jacobian is independent of the gauge group $G$. We can thus factor out the integral over $G$ in the integral \eqref{BWeq:xhatyhat}.
The regularized integral, counting gauge inequivalent choices only once, is then given as:
\begin{equation}
 I' = \int \BWd t J(m) \BWe ^{-\frac12 [\hat x (m) - \hat y (m)]^2}
\end{equation}

It is illustrative to give a few  possible gauge choices as an example, in order to show that the exact choice of representative element is not important (as long as it defines a good gauge slice, of course). We refer to figure \ref{BWfig:GaugeSlicesEx}. Choosing the red gauge slice, we have $(\hat x (t), \hat y (t)) = (x_0,t)$ and  $J = 1$.
The integral $I'$ becomes:
\begin{equation}
  I' = \int \BWd t \BWe ^{-\frac12 m^2} = \sqrt{\pi/a}\,.
\end{equation}
For the blue gauge slice in figure \ref{BWfig:GaugeSlicesEx}, we find $(\hat x (m), \hat y (m)) = (m,-m)$ and $J = 2$. The Gaussian integral $I'$ returns the same result:
\begin{equation}
   I' = 2\int \BWd t \BWe ^{-\frac 12 (2m)^2} = \sqrt{\pi/a}\,,
\end{equation}
after an easy change of coordinates.

In the following section, we bring the ideas of the above example into play when considering the more difficult case of finding a gauge slice for the Polyakov path integral.

\section{Regularizing the Polyakov path integral\label{BWs:Fixing_The_Gauge}}


\subsection{Gauge slice for Polyakov string}
Now we want to apply the method of the example to the string S-matrix. The question is now to find a good gauge slice for this case. We are performing an integral over the space
\begin{equation}
 \BWmetric\times\BWspacetime\times\left(\BWworldsheet\right)^n
\end{equation}
This is the space of metrics $g_{ab}\in \BWmetric$, embeddings $X^\mu\in \BWspacetime$ and vertex operator positions $\{\sigma_i\}\in \left(\BWworldsheet\right)^n$(remember that the latter are just coordinates on the worldsheet). We would like to write the path integral \eqref{eq:Smatrix} as an integral over the physically inequivalent configurations, given by the quotient space:
\begin{equation}
 \frac{ \BWmetric\times\BWspacetime\times\left(\BWworldsheet\right)^n }{\BWgauge}\label{eq:String_Coset}
\end{equation}
The action of the gauge group is given as follows. In terms of worldsheet coordinates $\sigma^a$, we are performing the path integral over configurations $ (g(\sigma),X,\sigma_i)$. The integrand  (consisting of Polyakov action and vertex operator insertions) is invariant under Weyl transformations and diffeomorphisms. We can write the action of a diffeomorphism $f$ and a Weyl transformation defined by a function $\phi$ as
\begin{equation}
 (g,X,\{\sigma_i\})\quad \to \quad(f_*(e^{\phi}g),f_*X,f(\sigma_i))\,.
\end{equation}
(The notation $f_*$ denotes the pullback of $f$.)

We count that the gauge group $\BWgauge$ has three real parameters (two for diffeomorphisms, one for Weyl transformations), so a first suggestion for a good choice of gauge slice (a coset representative of the coset space \ref{eq:String_Coset}) would be to keep a fixed metric $\hat g$ as a gauge slice, thus allowing the embeddings and vertex coordinates to vary freely along the slice:
\begin{equation}
 (\hat g, X, \sigma_i)
\end{equation}
In the next chapte, we will choose a flat metric $\hat g_{ab} = \delta_{ab}$, or $g_{z\bar z} = g_{\bar z z} =1/2, g_{zz}=g_{\bar z \bar z} = 0$ in compelx coordinates. Note that this metric will have a remaining conformal invariance (see below), related to the CKVs.
\begin{figure}[ht!]
\centering
\includegraphics[height=.13\textheight]{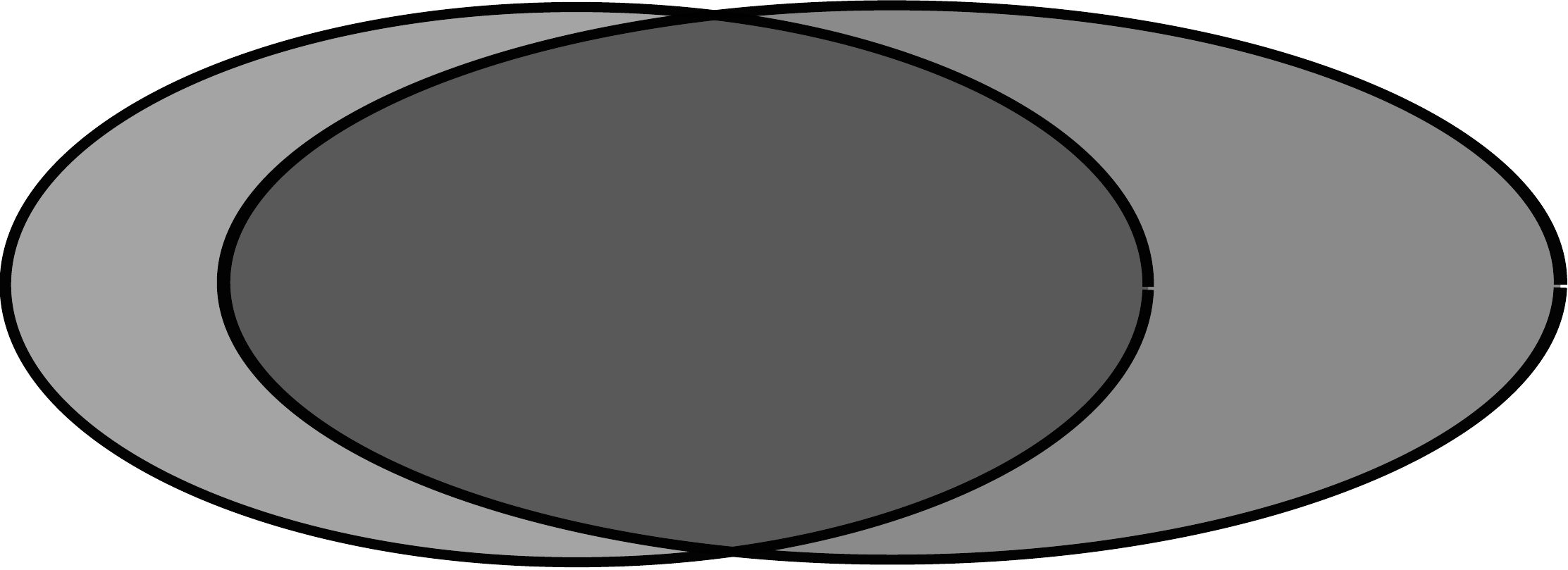}
\caption{Cartoon comparing possible metric changes and the action of $\BWgauge$. A large subset of $\BWgauge$ affects the metric (darkest grey), while a small subgroup does not affect the metric at all, the conformal Killing group (CKG, right).  Finally, moduli changes denote the freedom in metric choice independent of $\BWgauge$ transformations (left).\label{fig:metricdiff}}
\centering
\begin{picture}(-100,-150)
\put(-160,155){\bf \large Metric changes}
\put(-164,107){\textcolor{white}{Moduli}}
\put(-152,95){\textcolor{white}{$\delta t$}}
\put(20,101){\textcolor{white}{CKG}}
%
\put(-40,155){\large \bf Diff $\times$ Weyl}
\end{picture}
\end{figure}
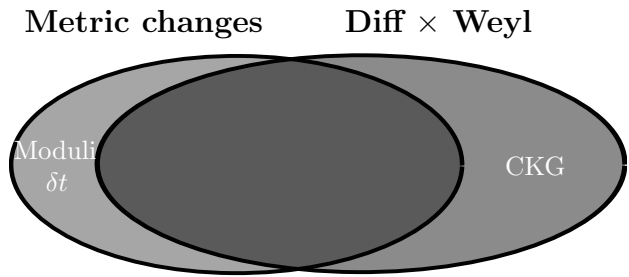

However, this is unsatisfactory in two ways. From the discussion in the previous chapter,  we know that the gauge group cannot fix the metric completely and that there is a subgroup of $\BWgauge$ leaving the metric invariant. This situation is sketched in figure \ref{fig:metricdiff}. We elaborate on those two issues:

\begin{itemize}
 \item \underline{Moduli.} Not all metrics are equivalent under $\BWgauge$  transformations, see section \ref{BWss:moduli_space}. Therefore a gauge slice (coset representative) parametrized by the metric depends on the moduli space describing the gauge inequivalent metrics:
\begin{equation}
 \hat g (m^\BWmodPara), \qquad \{m^\BWmodPara\} \in \frac{ \BWmetric}{\BWgauge}.
\end{equation}
In more abstract terminology, we choose a representative $\hat g (m^j)$ of the quotient space $ \frac{ \BWmetric}{\BWgauge}$. This (in general more than one-dimensional) space is called the moduli space, with elements parametrized by a moduli \textit{vector} $m^j$.
 \item \underline{Conformal Killing Vectors.} A subset of $\BWgauge$ does not even change the metric at all, see section \ref{BWs:CKVs}. These transformations are built up out of a simultaneous action of a diffeomorphism with the only effect of rescaling the metric and a  Weyl transformation canceling the rescaling diffeomorphism. Such diffeomorphisms are exactly generated by the conformal Killing vectors, so the group of elements of $\BWgauge$ leaving the metric unchanged is isomorphic to the conformal Killing group. (Remember that, in complex coordinates, the CKG consists of the conformal transformations $z \to f(z), \bar z \to \bar f (\bar z)$, that are furthermore bijective in order to be a diffeomorphism, see chapter \ref{BWc:Riemann}. These exactly affect the metric by a rescaling $g = \BWe^{\phi(z,\bar z)}\BWd z \BWd \bar z \to \Lambda g$, with $\Lambda = \left|\frac{\partial f}{\partial z}\right|^2$, which can then be counteracted by a Weyl transformation.)
\end{itemize}
The above discussion hints at how we should choose the gauge slice over which we have to integrate to regularize the string path integral. First, the slice is parametrized by the moduli $m^\BWmodPara$, through the parametrization of gauge inequivalent metrics $\hat g (m^\BWmodPara)$. And second, we need to compensate for the remaining overcounting, since any representative $\hat g (m^\BWmodPara)$ still contains elements which transform into each other under the gauge group.

When the path integral includes enough vertex operators, we can fix this remaining overcounting by restricting the positions of the appropriate number of vertex operator positions. Another way of defining a gauge slice is by restricting the embedding coordinates $X$ modulo $\BWgauge$. We choose the first possibility to avoid unnecessary clutter. (For the few examples we give in subsequent sections where the number of vertex operator insertions is too low for this procedure, we proceed on a case to case basis.)

Say the conformal Killing group had $2k$ real dimensions. Then we can define a gauge slice by fixing $k$ of the two-dimensional vertex coordinates on the worldsheet $\BWworldsheet$:
\begin{equation}
 (\hat g(m), \hat \sigma_1\ldots \hat \sigma_k; X, \{\sigma_i\})\,,.
\end{equation}
where $i$ runs over the $n-k$ unfixed vertex coordinates. This is a good choice of slice, since every element of  $\BWmetric\times\BWspacetime\times\left(\BWworldsheet\right)^n $ can be obtained from it through a combination of a Weyl transformation and a diffeomorphism.

%
%

\subsection{Restricting the integral to the gauge slice}

We want to write the path integral measure in terms of gauge parameters and moduli describing the gauge slice. This involves performing a change of variables from the original variables in the path integral i.e.\ the metric $g$ and $k$ fixed vertex positions $\sigma_i$ to new variables. Letting $\cal F$ denote the set of fixed vertex positions, we have:
\begin{equation}
 \BWD g \prod_{i\in {\cal F}} \BWd \sigma_i  \to \BWd^{\BWnmod} m^\BWmodPara ~\BWD {\rm (gauge)}\,,\label{eq:ChangeOfVars}
\end{equation}
where the moduli are denoted as an $\BWNmod$-dimensional vector in moduli space ($\BWnmod= 1 \ldots \BWNmod$) and the gauge parameters are assumed to have a good measure. This change of variables involves a Jacobian determinant. Finding its exact form is what we are concerned with now.


\subsubsection{Infinitesimal variation of the gauge slice}
When computing the determinant related to the change of variables \eqref{eq:ChangeOfVars}, it is easier to keep track only of infinitesimal gauge transformations. This means we are actually performing the  change of variables \eqref{eq:ChangeOfVars} on the  tangent space of the space of configurations $ \BWmetric\times\BWspacetime\times\left(\BWworldsheet\right)^n$. This is a sensible thing to do, since the Jacobian for coordinate transformations on a space $M$ is equal to the Jacobian on the tangent space $TM$.\footnote{Consider a space $M$ with coordinates $x^\mu$ in some patch and $y^\mu$ in some other patch. In the overlap of these patches, around a point $p$, say, we have $x(y)$ and we can calculate the Jacobian as $J = \left|\frac{\partial x^\mu }{\partial y^\nu}\right|_p$. Now consider an element $V\in T_pM$ on the overlap op those patches. Then we have $V = V^\mu_{(x)} \partial_x = V^\mu_{(y)} \partial_y$, where $\{V_{(x)}\}$ and $\{V_{(y)}\}$ are fibre coordinates on $T_pM$. In particular, $\det \left|\frac{\partial V^\mu_{(x)}}{\partial V^\nu_{(y)}}\right| = \det \left|\frac{\partial x^\mu }{\partial y^\nu}\right| = J$.}

We concentrate on the part of the measure including the vertex operator positions we want to fix and consider infinitesimal transformations only. It is convenient to split the measure on the gauge group in terms of (infinitesimal) Weyl transformations, denoted $\delta \phi$, and diffeomorphisms, denoted as $\delta \sigma^a$. The Jacobian determinant $J$ related to the change of variables \eqref{eq:ChangeOfVars} is then obtained from the transformation:
\begin{equation}
  \BWD \delta g \prod_{i\in f} \BWd \sigma_i = J \BWd^{\BWnmod} m \BWD \delta \phi \BWD \delta \sigma\label{eq:ChangeOfVarsJacob}
\end{equation}
We assume there is a well-defined measure on the space of Weyl transformations $\delta \phi$ and diffeomorphisms $\delta \sigma$.
The Jacobian $J$ is most easily computed after writing the metric deformations as an orthogonal decomposition\footnote{Note that we have not specified how this orthogonal decomposition should be understood. I.e., we have not defined a metric on the space of metric variations nor on the space of gauge transformations. A rigourous treatment of this problem can be found in Appendix \ref{BWapp:Gaussians}}:
\begin{equation}
  \delta g =\{Weyl\}\oplus\{Diff\}\oplus\{moduli\}\label{eq:OrthogDecompMetric}
\end{equation}
since then the Jacobian will nicely factorize.

Consider the action of an infinitesimal diffeomorphism generated by a vector field $\delta \sigma$ on the metric and vertex coordinates:
\begin{equation}
 \delta_D g_{ab} = (\nabla_a \delta \sigma_b)_{\hat g} + (\nabla_b \delta \sigma_a)_{\hat g}\,, \qquad \qquad \delta_D \sigma_i^a = \delta \sigma^a(\hat\sigma_i)
\label{BWeq:infindiffeo}
\end{equation}
Keep in mind that these variations depend on the choice of metric $\hat g$ and fixed vertex coordinates $\hat \sigma_i$. In the following, we will omit this explicit dependence, but it is always implied. An infinitesimal Weyl transformation $\delta \phi$ acts on the metric as:
\begin{equation}
 \delta_W g = \delta \phi \hat g\,.
\end{equation}
The most general metric variation can then be written as a part generated by these gauge transformations, and physical variations due to a change in the metric moduli:
\begin{equation}
\begin{array}{lccc}
   &{\rm Physical~ change}&&{\rm Gauge~ change}\\
 \delta g_{ab}=&\overbrace{~\delta m^\BWmodPara \partial_\BWmodPara \hat g_{ab}~}& +&\overbrace{\nabla_a \delta \sigma_b+\nabla_b \delta \sigma_a +\delta\phi \hat g_{ab}~}\,.
\end{array}
\label{BWeq:metricchanges}
\end{equation}
Now we could in principle change variables from the metric $g$ to gauge transformations $\delta \phi, \delta v$. But the decomposition \eqref{BWeq:metricchanges} is not orthogonal, it is not of the form \eqref{eq:OrthogDecompMetric}. Before we attempt to (sketchily) compute the Jacobian, we concentrate on finding such an orthogonal decomposition.

\subsubsection{Writing down the Jacobian\label{BWss:WritinDownJacob}}

Assuming for now that we can apply the rules for ordinary functions to functionals, we apply our usual knowledge of Jacobians to the path integral measure. (We comment on possible issues with path integral measures in section \ref{ss:RemarksMeasure} and Appendix \ref{BWapp:Gaussians}.) The Jacobian \eqref{eq:ChangeOfVarsJacob} can then be obtained from:
\begin{equation}
 J = \left|\frac{\partial (\delta g_{ab},\delta \hat \sigma_i)}{\partial (\delta \sigma, \delta \phi, \delta m)}\right|\,.\label{BWeq:Jacobian}
\end{equation}
We show how to compute this Jacobian, by highlighting two issues.

\begin{itemize}
\item \underline{Orthogonal metric variation} We perform an orthogonal decomposition as in \eqref{eq:OrthogDecompMetric} of the metric deformation in three steps. First, we isolate the pure trace part of the deformations. This is the infinitesimal change which is proportional to the metric itself, i.e.\ the part of the transformation acting as a scaling. It is important to know that not only Weyl transformations act in this way, but also a subclass of diffeomorphisms (the CKG) and certain moduli transformations rescale the metric. Secondly, we write the other diffeomorphisms that do not act as a rescaling (i.e.\ the diffeomorphisms which are not conformal transformations). And finally, we have the changes of the metric due to the moduli which are orthogonal to those two types of variations.

This gives us the following expression for the metric deformations:
\begin{equation}
 \delta g_{ab}=\delta\bar \phi g_{ab}+ (P_1\delta \sigma)_{ab} +\delta m^\BWmodPara \partial_\BWmodPara \hat g_{ab}|_{pr}.\label{BWeq:deltag_ortho}
\end{equation}
The parameter $\delta \bar \phi$ denotes Weyl transformations \'and the trace parts of the other transformations (i.e.\ those generated by the conformal Killing group and the trace part of the variations due to moduli). The operator $P_1$ is the unique operator mapping vectors to symmetric traceless tensors\footnote{One could generalize this operator to operators $P_n$ that take symmetric $n$-tensors to symmetric, traceles $(n+1)$-tensors. Since we only look to an operator on vectors, we restrict to the subscript ``1'' in $P_1$.}, such that $P_1\delta v$ is the traceless part of the effect of an infinitesimal diffeomorphism on the metric. This operator is defined as
\begin{equation}
 P_1\delta \sigma=\nabla_a \delta \sigma_b+\nabla_b \delta \sigma_a - \hat g_{ab}\nabla_c \delta \sigma^c\,,
\end{equation}
the trace part is clearly subtracted. Finally, the subscript ``$pr$'' on the moduli variations means we project those variations to the traceless subspace of metric variations, orthogonal to the diffeomorphisms $P_1\delta \sigma$.\footnote{Note: we assume we have a clear notion of orthogonality to perform this procedure. More on this in appendix \ref{BWapp:Gaussians}.}

We can classify these three spaces of variations in terms of the operator $P_1$. Note that the physical variations in the above decomposition are restricted to the orthogonal complement of the range of  $P_1$. Using the isomomorphism $({\rm Range}\, P_1)^\perp\cong {\rm Ker}\, P_1^\dagger$, we can then decompose the metric variations \eqref{BWeq:deltag_ortho} symbolically as:
\begin{align}
  \delta g&=\{Weyl\}\oplus\{Diff\}\oplus\{moduli\} \nonumber \\
&=\{Weyl\}\oplus\{{\rm Range}\, P_1\}\oplus\{{\rm Ker}\, P_1^{\dagger}\}.\label{eq:OrthogDecompMetric2}
\end{align}
The operator $P_1^{\dagger}$ is the adjoint of $P_1$, mapping traceless symmetric tensors to vectors. The relation between the differential operator $P_1$ and the metric deformations is summarized in table \ref{BWtab:P_1}.

\begin{table}[h!]
\centering{
 \begin{tabular}{|ccc|}
\hline
Math. space&$\cong$& Physical space\\
\hline
\hline
${\rm Range } \, P_1$&$\cong$& Diffeo effects\\
&& on metric\\
\hline
${\rm Ker } \, P_1$&$\cong$& CKG\\
\hline
${\rm Ker} \, P_1^{\dagger}$&$\cong$&Moduli variations\\
&&of metric\\
\hline
 \end{tabular}
}
\caption{Spaces associated to the differential operator $P_1$ and their interpretation (to be understood as isomorphisms, not necessarily identifications). Note that the conformal Killing group (CKG) is isomorphic to the kernel of the operator $P_1$, because exactly the kernel describes diffeomorphisms that act on the metric as a rescaling.\label{BWtab:P_1}}
\end{table}

\item \underline{Variation of vertex operator positions}
Remember from \eqref{BWeq:infindiffeo} that only diffeomorphisms affect the positions of vertex operators. 
In order to find the Jacobian in the path integral, it is sufficient to restrict attention to the effect of the CKG on the positions $\hat \sigma_i$. Denoting an infinitesimal element of the conformal Killing group as a real vector $\delta a^r, r = 1\ldots 2k$, we write its effect on the vertex operator positions as:
\begin{equation}
 \delta \hat\sigma_i^a = \delta a^ r C_r^ a (\hat \sigma_j)\,.\label{BWeq:deltaCKV}
\end{equation}
\end{itemize}

Now we are ready to sketch the calculation of the Jacobian. Using the orthogonal metric decomposition \eqref{BWeq:deltag_ortho} and the action of diffeomorphisms (and especially conformal Killing vectors) on the vertex operator positions, we can write for the Jacobian of the transformation \eqref{BWeq:deltag_ortho}:
\begin{align}
 J &= \left|\frac{\partial (\delta g_{ab},\delta \hat \sigma_i)}{\partial (\delta \phi, \delta \sigma, \delta t)}\right|=\det\left.\begin{array}{cccc|c}
  \multirow{3}{*}{$\left(\rule[1cm]{0pt}{0pt}\right.$}&1 &0&0&0\\ &* &P_1&0&0\\&0&0&\partial_i \hat g|_{pr}&0\\&0&1&0&C_s^a(\hat\sigma_i)
 \end{array}\right)
\end{align}
Assuming we can use ordinary matrix calculus, the Jacobian reduces to:
\begin{equation}
 J = \det(P_1)\cdot \det (\partial_\BWmodPara \hat g|_{pr}) \cdot \det(C_s^a(\hat\sigma_i))\,.\label{BWeq:Jacob_ill}
\end{equation}

The first factor contains the determinant of the differential operator $P_1$. Since this operator takes vectors to symmetric 2-tensors, the Jacobian factor $\det P_1$ is not well defined. In order to do so, we write it as:
\begin{equation}
 \det(P_1) \to \sqrt{\det{}' (P_1^{\dagger} P_1)}\,.
\end{equation}
In appendix \ref{BWapp:Gaussians} we give a better derivation of this result. Furthermore, we  put a prime on this determinant, to indicate we do not include the zero modes of $P_1^{\dagger} P_1$, which correspond to the zero modes of $P_1$ itself. These zero modes would make the determinant trivial. The reason we should not include them, is that they exactly correspond to the elements of the CKG and do not change the metric, see table \ref{BWtab:P_1}. The effects of the CKG sit in the last term in the Jacobian. This term should be understood as follows. The matrix $C_s^a(\hat\sigma_i)$ has its rows labeled by $s$, the columns by $a,i$. This is indeed a square $(2k\times 2k)$-matrix, with $2k$ the  (real) dimension of the  conformal Killing group. Finally, the second term in \ref{BWeq:Jacob_ill}, describing the part of the Jacobian due to the moduli, can be written down more precisely. To obtain the projection of $\partial_i \hat g$ to ${\rm Range }\, P_1\simeq{\rm Ker}\, P_1^{\dagger}$, choose a basis $\psi_A$ of ${\rm Range }\, P_1$. Moreover, assume we have a notion of scalar product on ${\rm Range }\, P_1$, denote it $(\cdot,\cdot)$. The projection to this space is then obtained as:
\begin{equation}
 \partial_\BWmodPara \hat g_{ab}|_{pr} = \sum_\alpha\frac{(\psi_A,\partial_\BWmodPara \hat g_{ab})}{(\psi_A,\psi_A)}\,.
\end{equation}

This leads us to the final expression for the Jacobian:
\begin{equation}
J(m^\BWmodPara)= \det {}' (P_1^{\dagger} P_1)\cdot\sum_\alpha\frac{(\psi_A,\partial_i \hat g_{ab})}{(\psi_A,\psi_A)}\cdot \det(C_s^a(\hat\sigma_i))\,. \label{BWeq:Jacob_well}
\end{equation}
To obtain the result \eqref{BWeq:Jacob_well}, we have been very quick and left some things ill-defined. The derivation above relies heavily on the concept of orthogonal decompositions, but we have not yet defined a notion of scalar product on the spaces in question! Moreover, for the computation of the several Jacobians, we were all too eager to apply rules from ordinary differentiation to path integrals. For a more careful analysis of the above result, we refer to appendix \ref{BWapp:Gaussians}. In section \ref{BWc:Amplitudes}, we continue with writing down the gauge fixed Polyakov path integral for the (closed) bosonic string, based on the above calculation of the Jacobian determinant.

\subsection{Final result for the Polyakov path integral}
We can write down the S-matrix path integral for the change of variables \eqref{eq:ChangeOfVarsJacob}.
\begin{align}
 S_{j_1\ldots j_n} (k_1\ldots k_n) =\label{BWeq:GaugeFixed_PI_EndResult1}
\sum_{\rm g} g_s^{2g-2}&\int \frac{\BWD \delta \phi \BWD \delta \sigma}{\cal N}\int \BWd^{\BWnmod} m J(m^\BWmodPara)
\\
&\int\BWD X \, \exp({-S_P(X)}) \nonumber
\left(\prod_{i\slashed \in {\cal F}} \int \BWd^2 \sigma_i\right) \sqrt{g(\sigma_i)} \prod_i {\cal V}_{j_i}(k_i,\sigma_i)\,.
 \end{align}
The integration over $\delta\phi$ and $\delta\sigma$ returns the volume of the gauge group of Weyl transformations and diffeomorphisms, since the integrand is explicitly independent of these $\delta\phi$ and $\delta\sigma$ (we have `gauge fixed' the path integral.) Therefore, this factor cancels with the normalization constant $\cal N$. With the expression for the Jacobian \eqref{BWeq:Jacob_well}, the regularized expression for the string S-matrix is then:
\begin{align}
 S_{j_1\ldots j_n} (k_1\ldots k_n) = \sum_{\rm g} g_s^{2g-2} &  \int\BWd^{\BWnmod} m  \det {}' (P_1^{\dagger} P_1)\,\sum_\alpha\frac{(\psi_\alpha,\partial_i \hat g_{ab})}{(\psi_\alpha,\psi_\alpha)}\, \det(C_s^a(\hat\sigma_j))\label{BWeq:GaugeFixed_PI_EndResult}\\\
&\int \BWD X \exp( {-S_P(X)}) \left(\prod_{i\slashed \in {\cal F}} \int \BWd^2 \sigma_i\right) \sqrt{g(\sigma_i)} \prod_i {\cal V}_{j_i}(k_i,\sigma_i)\,.
\nonumber
 \end{align}
This is the final result we use in the remainder of these lectures. Remember that the calculation above was only for a gauge slice where we assumed there wer at least as many vertex operator insertions as there were elements in the CKG. Below we comment on the path to follow when there are not enough vertex operators present.


%
%

\section{Problems and other methods}\label{ss:RemarksMeasure}

\subsection{The Weyl anomaly}

Defining a good measure in the path integral can be seen as the basis of a good quantum theory. A priori, a classical theory, defined through its Lagrangian, can correspond to a lot of candidate-quantum theories. Heuristically this can be understood as that the quantum Lagrangian should reduce to the classical one in the limit $\hbar \to 0$, thereby allowing a plethora of possible quantum Lagrangians with the same classical behaviour. The important thing to know is that there is nothing unique or special about a quantum version of a certain system, or, as pretold by Shakespeare \cite{Shakespeare}: ``What's in a name?''. The choice of quantum theory corresponding to a classical one is entirely ours. Most often, people choose a theory with the same symmetries as the original Lagrangian, as this is a natural guess.

In the path integral formalism, one is tempted to think the quantum theory is uniquely defined, because one integrates over fields weighted with the exponential of the action of the classical theory. However, this is a misconception. The path integral measure of the fields hides the quantum nature of the path integral. This measure depends on $\hbar$ and it should be specified what form it takes exactly. For example, it can be shown that the path integral in quantum mechanics obtained either by slicing the path in equal time distances, or by taking a non-uniform parametrization of the time steps, gives a different result. One can write both versions of the path integral as an integral with the same measure, but then the action in both cases is only the same up to corrections of order $\hbar$. This shows that choosing a certain action in the path integral, does not determine the theory completely: one should in principle give a precise prescription for the measure.

In the discussion of the previous sections, we left the issue of finding a good measure untreated. We go into this now. In order to define a quantum theory of excitations on a string, we use  symmetry as a guideline. From a spacetime point of view, we do not wish to distinguish between worldsheets which are diffeomorphic, since this does not give us a difference in the spacetime spectrum. Therefore, we demand our quantum theory to be diffeomorphism invariant. Also for the Weyl invariance, we play this game (even though it may be less obvious what this symmetry is from the target space point of view). This gives us the following requirement: we want a quantum theory invariant under worldsheet diffeomorphisms and Weyl transformations. In the path integral approach, we can translate this requirement to demanding that the path integral is invariant under these symmetries (as the Polyakov action is already invariant under these symmetries). We have started from a classical action which was invariant under the worldsheet symmetries and use this action in the definition of the path integral. In order to have an invariant quantum theory, we  have to find a Diff and Weyl invariant measure in the path integral. It turns out to be possible to define a worldsheet diffeomorphism \'or Weyl invariant path integral measure in general, but this can not be said for both symmetries. One says the quantum theory has an anomaly: it does not respect the same symmetries as the classical theory. Usually, people like to keep the diffeomorphism invariance manifest, thus giving rise to an anomaly of the Weyl symmetry. The Weyl anomaly disappears when one takes the dimensionality of spacetime to be $D=26$ for the bosonic string, as Polyakov showed \cite{Polyakov:1981rd}.  This leaves us with the observation: if we want (bosonic) string theory to describe physics which does not depend on the local description of the worldsheet, we need to consider 26 spacetime dimensions. From now on, we put the dimension of spacetime to $D=26$.

\subsection{The path integral measure}
Before, we were cavalier about the measure in two ways:
\begin{enumerate}
 \item The measures $\BWD g, \BWD X$ were not defined. This holds as well for the measure on the gauge group $\BWD \delta v, \BWD \delta \phi$.\footnote{Note that $\BWd^\kappa t$,m the measure on the moduli space, is well-defined. Since it is a measure on $\mathbb{R}^{\BWnmod}$, we can take it to be the ordinary Lebesgue-measure on this space}
 \item We assumed orthogonality of the metric variations, see \eqref{eq:OrthogDecompMetric2}. But no notion of (metric) distance on the space of integration for the gauge parameters, nor for the fields, was given.
\end{enumerate}

We clearly need a good clean definition of the path integral measure to cope with these difficulties. Sadly enough, a natural generalization of the Lebesgue measure to functional integral is not at hand. A way to define the measure in functional integration does however exist, but it can only be done (semi-)rigourously for Gaussian integrals. Luckily for us, the integrals we look at are exactly of this form. Consider \emph{defining} a Gaussian integration, by saying that, for functional integration over a collection of (bosonic) fields $\phi_\alpha$, we have:
\begin{equation}
 \sqrt{\det A} = \int \prod_{\alpha }\BWD \phi_\alpha \BWe^{(\phi, A \phi )}\label{BWeq:GaussianFieldInt}\,,
\end{equation}
where $(\cdot,\cdot)$ denotes some scalar product available to us on the space of fields $\phi_\alpha$. This is the straightforward generalization of the corresponding Gaussian integral for ordinary integration. In general, there is a general notion of such a scalar product, see the appendix. When taking this definition, and assuming that the linearity rules of ordinary differentiation apply equally to functional integration, we have a well-defined measure and can perform Gaussian integrals without any problems. We expand on this in appendix \ref{BWapp:Gaussians}.

Please note that the way of defining the measure through an equation as \eqref{BWeq:GaussianFieldInt}, as is done in appendix \ref{BWapp:Gaussians}, clearly shows that the invariance of the measure under gauge symmetries is determined by the invariance of the scalar product on the space of fields.

\subsection{Faddeev-Popov procedure \label{ss:FPprocedure}}

Up until now, we have concentrated on a seemingly straightforward way of dealing with the overcounting in the path integral. However, several problems have appeared. First, there are subtleties with defining a measure as discussed above. Second, in order to develop the proposed method carefully, a lot of calculations are involved. Luckily, a formalism to deal with the overcounting in path integrals has been developed, which is technically less involved. This procedure of gauge fixing goes by the name of Faddeev-Popov procedure, after the physicists who first introduced it \cite{Faddeev:1967fc} in constructing the path integral for the Yang-Mills field.

The Faddeev-Popov trick consists of introducing extra fields in the path integral which interact with the original fields in the physical problem.  These extra fields obey the wrong spin-statistics (i.e.\ bosonic symmetries generate fermionic fields  and vice versa). The nice feature of these fields is that by including them in the calculation of amplitudes, the correct result (without the overcounting) rolls out automatically: one obtains a calculationally attractive way of obtaining the gauge-fixing of gauge symmetries, by adding extra Feynman diagrams. Luckily, the Faddeev-Popov fields only appear in Feynman loops, but never in external legs of diagrams. They are in this sense ``unphysical'' and should be seen as giving corrections to the naive amplitudes of physical processes, thereby returning the desired result. Because of their peculiar properties, these fields are often referred to as ``(Faddeev-Popov) ghosts''. The computational simplifications brought to us by the Faddeev-Popov ghosts -- replacing calculations of determinants with extra fields in Feynman loops -- have made the Faddeev-Popov procedure the most popular way of dealing with overcounting in gauge theories and gave the path integral method a real boost. Even though details can remain somewhat obscure, calculations return correct results: for most physicist enough reason to widely apply the procedure.

Although the inclusion of ghosts makes the calculation of the Jacobian determinant we introduced before easy, it obscures some of the physics. Yes, we can easily compute the gauge-fixed amplitudes of our theory by letting the ghosts run in loops and taking them into account. But what is their interpretation in physical terms? In our opinion, they should only be seen as a mathematical trick to deal with overcounting problems in quantum gauge theories. Therefore, we keep to the die-hard calculations discussed above. They may appear cumbersome, but we keep a closer eye on what is going on. Note that the overcounting can be tackled by purely diagrammatical techniques as well, without reference to ghost fields or even path integrals. Especially Veltman was and is a firm supporter of this technique, see for instance \cite{'tHooft:1973pz}. The Faddeev-Popov procedure becomes especially handy when considering the quantization of the superstring. In that setting, the straightforward way of writing down a gauge-fixed path integral following the methods of the previous section breaks down and the Faddeev-Popov procedure gives a handle on the computation. Since we only look at the bosonic string, we do not go into the details of the Faddeev-Popov procedure.  The reader interested in the basic ideas of the procedure, is invited to consult appendix \ref{BWapp:FaddeevPopov}. We also refer to the excellent texts \cite{Green:1987sp,Polchinski:1998rq,Kiritsis:2007zz}, where a good account of the Faddeev-Popov procedure can be found.

\chapter{Amplitudes\label{BWc:Amplitudes}}

\emph{In the previous chapter, we have learned that gauge-fixing the Polyakov path integral removes infinite overcounting and leaves us with the path integral \eqref{BWeq:GaugeFixed_PI_EndResult} for the string S-matrix. The integrals for the amplitudes in the S-matrix have three main features. First, they all contain an integral over the embedding coordinates $X^\mu$.  Second, there is an additional integral over the metric moduli. And finally, the path integrals may contain Jacobian determinant factors. These can arise from residual conformal Killing symmetry of the Riemann surface and from the moduli.}

\emph{We perform the $X$ integral for a general Riemann surface in section \ref{BWs:AmplitudesGenFunct}. The issues related to the CKG and moduli, we treat on a case by case basis. We single out the tree level amplitudes  leading to Riemann surfaces of genus zero, and the one-loop amplitudes, by considering genus one surfaces. Riemann surfaces of higher genus are not considered. They contribute to the perturbative string expansion in higher orders of the string coupling constant, which will be briefly discussed in the next chapter. Section \ref{BWs:OSScatteringAmplitudes} goes into a discussion of the (on-shell) amplitudes on the Riemann sphere, by increasing number of vertex operator insertions, while section \ref{BWs:PartitionFunction} treats the notion of a string partition function, for both the Riemann sphere and the torus.}

\section{Generating functional approach to \texorpdfstring{$X$}-integral}\label{BWs:AmplitudesGenFunct}

We follow the setup of Polchinski's first book \cite{Polchinski:1998rq}, chapter 6. Other good references are Nakahara \cite{Nakahara:2003nw} and D'Hoker and Phong \cite{D'Hoker:1988ta}. Since we focus on vertex operator insertions, it proves useful to study the generating functional
\begin{equation}
 Z[J] = \langle \exp ({\BWi \int \BWd^2 \sigma J_\mu X^\mu}) \rangle\,,\label{BWeq:ZJ}
\end{equation}
in terms of a spacetime vector of sources $J^\mu (\sigma)$. Consider for now only the  expectation value over the embedding coordinates, weighted with the Poyakov action, which we can rewrite using partial integration as
\begin{equation}
 S_P = \frac1{4\pi \alpha'}\int \BWd^2\sigma X_\mu\nabla^2 X^\mu\,.
\end{equation}
up to boundary terms, where the Laplacian is given by\footnote{Remember we use a Euclidean signature for the worldsheet, which is obtained after a Wick rotation. See section \ref{BWs:WSasRS}.}
\begin{equation}
 \nabla^2 = \frac 1{\sqrt{g}}\partial_a \sqrt{g}g^{ab}\partial_b\,.
\end{equation}

Before digging into the calculation, we can guess what the expression for $Z[J]$ should look like. From our knowledge of Gaussian integrals, see for instance appendix \ref{BWapp:Gaussians}, we expect the following result:
\begin{equation}
 Z[J] = Z[0] \exp\left({\int\BWd^\sigma \BWd^2\sigma' J(\sigma) G(\sigma,\sigma') J(\sigma')}\right)\,,
\end{equation}
where $G(\sigma,\sigma')$ is some Green's function corresponding to the inverse of the Laplacian. The vacuum contribution $Z[0]$ corresponds to  the determinant of the Laplacian. Although the final result has (roughly) the above form, care has to be taken in constructing the Green's function due to zero modes of the Laplacian possible on (compact) Riemann surfaces. The reader not interested in the mathematical details, may skip the calculations and go immediately to eq. \eqref{BWeq:ZJfinal}.

The path integrals in the generating functional are Gaussian, and can be performed by completing the squares. We do this by first expanding the embedding coordinates and sources in terms of eigenfunctions of the Laplacian  as follows:
\begin{align}
 X^\mu &= \sum_i x^\mu_i \psi_i \,,\qquad J^\mu = \sum_i J^\mu_i \psi
 _i\,,\nonumber\\
 \nabla^2 \psi_i &= - \lambda_i^2 \psi_i \quad{\rm(no~sum)}\,,
\end{align}
The eigenfunctions $\psi_i$ are chosen to be orthogonal w.r.t. the natural inner product:
\begin{equation}
 \int \BWd^2 \sigma \sqrt{g}\, \psi_i\cdot  \psi_j = \delta_{i,j}\,.
\end{equation}
Note that the zero mode is constant on a compact Riemann surface: since $\nabla^2 \psi_0 = 0$, $\psi_0$ reaches its extremum in the interior of the Riemann surface, and by applying the maximum principle it must be constant.

We see that the generating functional becomes:
\begin{align}
 Z[J] &= \int \BWD X \exp\left({\frac1{4\pi \alpha'}\int \BWd^2\sigma \left(X_\mu \nabla^2 X^\mu + \BWi J_\mu X^\mu\right)}\right)\nonumber\\
&=\prod_{i,\mu}\int \BWD X^\mu \exp\left({-\frac{\lambda_i^2}{4\pi \alpha'} x_i\cdot x_i + \BWi J_i\cdot x_i}\right)
\end{align}
The integral over the constant mode produces a delta function for the corresponding source $J_0 = \int \BWd \sigma \sqrt{g }\psi_0 J$.\footnote{For the vacuum amplitude ($J=0$), the term $J_0\cdot X_0$ is absent and consequently the delta function is replaced by the factor
\begin{equation}
V^{-1} \left(\int \BWd^2 \sigma \sqrt{g}\right)^{-d/2}\,,
\end{equation}
where $V$ is the spacetime volume. The reader is invited to perform this calculation in exercise \ref{BWexerc:volume}.} For convenience, we keep the number ged of spacetime dimensions unfixed. Completing the squares and performing the Gaussian integral as in \eqref{BWeq:app-FP_Gaussian_integral}, we find
\begin{align}
 Z[J] &=\BWi(2\pi)^d \delta^d(J_0)\prod_{i\neq0} \left(\frac{2\pi^2\alpha'}{\lambda^2_i}\right)^{d/2}\exp(-{\frac{\alpha' \pi}{\lambda_i^2} J_i\cdot J_i})
\nonumber\\
&=\BWi (2\pi)^d \delta^d(J_0)\det{}'\left(\frac{\nabla^2}{2\pi^2\alpha'}\right)^{-d/2}\exp(-\frac12{\int\BWd^\sigma \BWd^2\sigma' J(\sigma) G'(\sigma,\sigma') J(\sigma')})\,,\label{BWeq:ZJfinal}
\end{align}
where the prime on the determinant means omitting of the zero modes and the Green's function is given by:
\begin{equation}
 G'(\sigma,\sigma')=\sum_{i\neq 0} \frac{2 \pi \alpha'}{\lambda_i^2} \psi_i(\sigma)\psi_i(\sigma')\,.
\end{equation}
It is the solution of the differential equation:
\begin{align}
 -\frac1{2\pi\alpha'}\nabla^2 G'(\sigma,\sigma') &= \sum_{i\neq0} \psi_i(\sigma) \psi_i(\sigma')\nonumber\\
&= \frac1{\sqrt{g}} \delta^2(\sigma-\sigma') - \psi_0^2\,.\label{BWeq:GreenFunction}
\end{align}
Only in the strict sense of the equation \eqref{BWeq:GreenFunction} is  the Green's function $G'(\sigma,\sigma')$ to be seen as the inverse of the Laplacian. Note also the appearance of the zero mode $\psi_0$ on the right-hand side.
From previous experience, you could expect only a delta function source and the constant term may seem strange.  In more common situations (e.g. electrodynamics), other CFTs) one considers the Laplacian and the corresponding Green's function on a non-compact space. In other words: the zero-mode does not contribute. However, for compact spaces, like  the Riemann surfaces we are describing in these lectures, this no longer holds. Physically, the field lines for the source have no place to go and we can view the constant $\psi_0^2$-term as a neutralizing background contribution.

It is instructive to take a closer look at the different factors making up the generating functional $Z[J]$. First, there is the $\delta$-function, which translates in momentum conservation, see below. Then we have a determinant factor. After a suitable renormalization, it can be dealt with. The result can only depend on the worldsheet coordinates through the metric moduli. We explain this below for the torus.  When we consider tree-level amplitudes (operator insertions on the sphere), there are no moduli and the resulting factor is just a constant, there is no dependence on the worldsheet coordinates $z,\bar z$.  Finally, the exponential $\exp(-\frac12\int J G J)$ is specific to each amplitude (i.e. it depends on the specific form of the sources $J$). This will always be a certain holomorphic function of the complex coordinates $z,\bar z$ on the worldsheet.

\begin{BWexerc}\label{BWexerc:volume}
{\bf See D'Hoker \& Phong (1988).}
 Perform the path integral \eqref{BWeq:ZJ} for the vacuum $Z[0]$ using the Jacobian methods of chapter the previous chapter. You can first perform the minimal split of the $X$-variable into a (constant) zero mode $X_0$ and all modes orthogonal to it: $X = X_0 + X'$ and proceed from there. In particular, use this calculation to find the normalization of the vacuum path integral:
\begin{equation}
 Z[0] = \frac 1V \left(\frac{\det{}'\nabla^2}{\int \BWd^2 \sigma \sqrt{g}}\right)^{-d/2}\,, \label{BWeq:Xvacuum}
\end{equation}
where $V = \int \BWD X_0$ is the spacetime volume.
\end{BWexerc}

We conclude that in order to evaluate the string S-matrix, with certain vertex operator insertions, it suffices to apply the result \eqref{BWeq:ZJfinal} for the generating functional $Z[J]$ to the case at hand. In particular, this requires finding the approriate Green's function on the Riemann surface. We now explain this straightforward calculation of the path integral to the Riemann sphere and other surfaces with tachyon operator insertions.

\subsection{The sphere}
Consider the Riemann sphere with complex coordinates as in example \ref{BWex:RiemannSphere} and equipped with the metric (see example \ref{BWex:RiemannSphereMetric})
\begin{equation}
 \BWd s^2 = \frac{4\BWd z\BWd \bar z}{(1+z\bar z)^2}
\end{equation}
The solution to the differential equation \eqref{BWeq:GreenFunction} is then given by:\footnote{Use the property
\begin{equation} 
 \partial \bar \partial \ln |z| = {2\pi} \delta(z)\,.
\end{equation}
See Modave Lectures of Raphael Benichou on conformal field theory.}
\begin{equation}
 G'(\sigma^1,\sigma^2) = -\frac{\alpha'}2 \ln |z_1-z_2|^2 + f(z_1,\bar z_1) + f(z_2,\bar z_2)\,,
\end{equation}
with the function $f$ given by:
\begin{equation}
 f(z,\bar z)=\frac{\alpha' \psi_0^2}{4}\int\BWd^2 z' \BWe^{2\omega(z,\bar z)}\ln |z-z'|^2 + C\,.
\end{equation}
This function solves the equation $\nabla^2 f = \pi \alpha' \psi_0^2$. It will drop out in the end result for the $n$-tachyon amplitude and can be seen as a regulator. The constant $C$ is determined by demanding orthognonality of $G'$ to the zero mode $\psi_0$. We do not keep track of this constant, since the function $f$ will drop out of the physical result anyway.

\begin{BWexerc}
 Relate the solution for the Green's function to the $X$-propagator. In order to define the latter properly, you will have to make use of a cut-off. This should give meaning to calling the function $f$ a ``regulator'': independence of the propagator on the cut-off, corresponds to the function $f$ dropping out of the physical amplitudes.
\end{BWexerc}

As a first example, consider $n$ tachyon vertex operators on the sphere. We know from section \ref{BWss:VertexOperators} that this corresponds to an insertion of $n$ operator $\BWe^{i k_i\cdot X_i }$. Our first naive guess for the path integral is
\begin{equation}
 M_n^{S^2}(k_1\ldots k_n,\sigma_1\ldots\sigma_n)\equiv\langle\prod_{i=1}^{n}\BWe^{\BWi k_i\cdot X(\sigma_i)}\rangle.
\end{equation}
Comparing to the definition of the generating functional $Z[J]$ \eqref{BWeq:ZJ}, we see this corresponds to $J(\sigma)= X(\sigma)$:
\begin{equation}
 M_n^{S^2}(k_1\ldots k_n,\sigma_1\ldots\sigma_n)=\BWi C_{S^2}(2\pi)^d \delta^d(\sum_i k_i)\exp({-\frac12\sum_{i,j=1}^n}G'(\sigma_i,\sigma_j))\label{BWeq:tachyon1}
\end{equation}
where we absorbed the determinant of the worldsheet Laplacian in the constant $C_{S^2}$:
\begin{equation}
 C_{S^2} = \frac1{\psi_0^d}\det{}'\left(\frac{\nabla^2}{4\pi^2\alpha'}\right)^{-d/2}\,.
\end{equation}
This determinant can in principle be regularized and computed, but we will not need to do this here. We just note that this determinant is a constant from the worldsheet point of view and leave its value unspecified. It is more important to notice that the amplitude we proposed diverges, due to the self-contractions in the result \eqref{BWeq:tachyon1}: $G'(\sigma_i,\sigma_i)$ is infinite. To regularize these divergences, it is customary to introduce some renormalized tachyon vertex operator:
\begin{equation}
 \left[\BWe^{\BWi k\cdot X}\right]_r\,,
\end{equation}
whose explicit form depend on the regularization scheme. The coordinate dependence of the $n$-tachyon amplitudes can be split in terms that do not involve self-contractions and those that do:
\begin{equation}
 M_n^{S^2}(k_1\ldots k_n,\sigma_1\ldots\sigma_n)=\BWi C_{S^2}(2\pi)^d \delta^d(\sum_i k_i)\exp\Big{(}-\sum_{i,j=1, i>j}^nG'(\sigma_i,\sigma_j)-\frac12 \sum_{i}G'_r(\sigma_i,\sigma_i)\Big{)}\label{BWeq:tachyon2}.
\end{equation}
The self-contractions now involve:
\begin{equation}
 G'_r(\sigma,\sigma) = G'(\sigma,\sigma) + {\rm ~extra~ term(s)}\,.
\end{equation}
The extra terms depend on the regularization scheme. For our purposes, it is satisfactory to use normal ordered operators in order to have well-behaved amplitudes:
\begin{equation}
 \left[\BWe^{\BWi k\cdot X}\right]_r = ~:\BWe^{\BWi k\cdot X}:\,.
\end{equation}
It then follows that
\begin{equation}
 G'_r(\sigma,\sigma) = 2f(z,\bar z)
\end{equation}
and the amplitude becomes:
\begin{equation}
  M_n^{S^2}(k_1\ldots k_n,\sigma_1\ldots\sigma_n)=\BWi C_{S^2}(2\pi)^d \delta^d(\sum_i k_i)\prod_{i,j =1,i>j}^n|z_i-z_j|^{\alpha' k_i\cdot k_j}\label{BWeq:tachyon3}\,.
\end{equation}
We see that both the self-contractions and the dependence on the function $f$ drop out in the final expressions. The latter arises because $f$ only appears in the combinations \mbox{$\sum_{i,j} k_i\cdot k_j f(z_i,\bar z_i)$}, which give zero due to momentum conservation. In conclusion, we see that the $n$-point tachyon amplitude on the Riemann sphere contains a constant prefactor, a delta-function enforcing momentum conservation and a specific coordinate dependence, stipulating the amplitude only depends on the distances between the vertex operator positions measured with the locally flat metric.


\subsection{The torus}\label{BWref:torus}
We calculate the $X$-dependent part of the vacuum amplitude on a torus with metric $ds^2 = d z d \bar z$ and identifications $z \sim z +1$ and $z\sim z +\tau$ (we need the vacuum amplitude on the torus below). As indicated above in eq.~(\ref{BWeq:Xvacuum}) this part is simply given by,
\begin{equation}
 Z[0] = \frac 1V \left(\frac{\det{}'\nabla^2}{\int \BWd^2 \sigma \sqrt{g}}\right)^{-d/2}\, .
 \end{equation}
It suffices to evaluate $\det{}'\nabla^2$ on the torus. The appropriate eigenfunctions of the laplacian on the torus are given by,
\begin{eqnarray}
\psi_{n,m} (\sigma^1, \sigma^2) = e^{2 \pi i\, \left( n \sigma^1 - \frac{\tau_1}{\tau_2} n \sigma^2 - \frac{1}{\tau_2} m \sigma^2 \right)},
\end{eqnarray}
and the corresponding eigenvalues read,
\begin{eqnarray}
\lambda_{n,m}= \frac{4 \pi^2}{\tau_2^2}  \left\{ (m + \tau n) (m + \bar \tau n) \right\}.
\end{eqnarray}
As we do not take into account the zero modes (this means $m$ and $n$ cannot be zero simultaneously), the determinant we should calculate reduces to,
\begin{eqnarray}
\det{}'\nabla^2 = \prod_{m, n{\rm \, not\,both\, zero}} \frac{4 \pi^2}{\tau_2^2}  \left\{ (m + \tau n) (m + \bar \tau n) \right\}.
\end{eqnarray}
The details of this calculation can be found in e.g. \cite{Bagger:1987rd}. We shall limit ourselves to giving the result,
\begin{eqnarray}
\det{}'\nabla^2 = \tau_2^2 \eta(\tau)^2 \bar\eta(\bar\tau)^2,
\end{eqnarray}
where $\eta(\tau)$ represents the Dedekind-$\eta$ function. We will introduce this function in the next section. For $D=26$ the final result for the $X$-dependent part reads,
\begin{eqnarray}
Z[0] = \frac 1V  \frac{1}{\tau_2^{13}} \frac{1}{(|\eta (\tau) |^{2})^{26}}\, .
\end{eqnarray}
Recall that there is  an additional factor $(\int \BWd^2 \sigma \sqrt{g})^{13} = \tau_2^{13}$ to take into account.

\section{On-Shell Scattering Amplitudes\label{BWs:OSScatteringAmplitudes}}

We now give a brief treatment of some exemplary scattering amplitudes. We choose to discuss only the $n$-point tachyon-amplitudes in closed bosonic string theory in detail. Nevertheless, a lot of rich ``stringy'' physics can be found from those amplitudes.  The main results apply to other amplitudes as well. We stick to on-shell amplitudes (elements of the S-matrix). A prescription for scattering of particles off the mass shell (i.e.\ $k\neq 4/\alpha$ for tachyon modes of the closed bosonic string) is very hard, for reasons discussed in the introduction.

In this section, we discuss $n$-point tachyon amplitudes on the Riemann sphere. In other words, we consider tree-level amplitudes in the field theory limit (there are no string-loops). Recall from chapter \ref{BWc:Riemann} on Riemann surfaces that the sphere has three complex conformal Killing vectors. Combine this with the discussion in section \ref{BWs:Fixing_The_Gauge} to see that we need at least three operator insertions for the general S-matrix result \eqref{BWeq:GaugeFixed_PI_EndResult} to be applicable here. Therefore, we first single out the case with 1 or 2 vertex operators in section \eqref{BWss:1,2_Vertex_Ops} and see that those amplitudes vanish. We continue with a discussion of three or more vertex operators in section \ref{BWss:>3_Vertex_Ops}. The calculation without operator insertions (the vacuum amplitude) is deferred to section \ref{BWs:PartitionFunction}. All calculations in this chapter use the result for the $X$ integral of eq.\ \eqref{BWeq:tachyon3}.


\subsection{1, 2 vertex operators\label{BWss:1,2_Vertex_Ops}}
In this case, we do not have enough vertex operators to fix the gauge invariance of conformal transformations on the worldsheet completely. There remains a residual invariance under a subgroup of the conformal Killing group. Therefore the path integral is reduced to (recall the discussion in chapter \ref{BWc:Fixing}):
\begin{equation}
\int \frac{\BWD X}{Vol(CKG)}J\prod_{i=1}^n  \int \BWd^2 \sigma_i \sqrt{g(\sigma_i)}{\cal V}_{j_i}\,.
\end{equation}
Now $n=1$ or $n=2$ and we have some (finite) Jacobian factors conveniently written as $J$. We can fix the vertex operator position(s) using a subgroup $H$ of the CKG, factoring out the volume of that subgroup $H$. This gives rise to a Jacobian which we absorb in $J$, writing the integral as:
\begin{equation}
  \int \frac{\BWD X}{Vol(CKG/H)}J \prod_{i=1}^n  \sqrt{g(\sigma_i)}{\cal V}_{j_i} = 0\,.
\end{equation}
The resulting integral is effectively zero, because the $X$ path integral gives only a finite result and the volume of the residual conformal symmetry $CKG/H$ is still infinite (see also \cite{Moore:1985ix,D'Hoker:1985pj,Polchinski:1985zf}). We conclude that $n$-point tachyon amplitudes on the Riemann sphere are zero whenever we have less vertex operator insertions then the number of (complex) conformal Killing vectors.

\subsection{3 or more vertex operators\label{BWss:>3_Vertex_Ops}}

We continue with the discussion of the path integral on the sphere with $n\geq3$ tachyon operator insertions. We start from the path integral result \eqref{BWeq:GaugeFixed_PI_EndResult}, the main result of the previous chapter. Remember that we are working on the 
sphere, there are no metric moduli. Consider $n\geq 3$ tachyon vertex operators at positions $z_i$. Using the M\"obius group of CKVs on the sphere, see section \ref{BWs:CKVs}, we may again fix three of these positions. This leaves us with an integral over the $n-3$ remaining positions and a determinant over the CKG generators, which we write as an expectation value over the $X$ variable weighted with the Polyakov action as
\begin{equation}
  \langle \det (C^a_r(\hat \sigma_j) )\prod_{i=4}^n \int \BWd^2 \sigma_i \sqrt{g} \prod_{j=1}^n \BWe^{\BWi k_j\cdot X} \rangle
\end{equation}
Compared to the general result for the path integral in equation \eqref{BWeq:GaugeFixed_PI_EndResult}, we do not discuss constant factors and absorb them into an overall constant at the end of the calculation. Those constant factors are the determinant $\det' \nabla^2$, as before, but also the determinant $\det ' P_1^T P_1$ and all numerical constants, as for instance factors of $\pi$ and a possible multiplicative factor if the basis of CKVs is not orthonormal. The two determinants are constants since they can only depend on the worldsheet coordinates through the metric moduli. In section \ref{BWs:AmplitudesGenFunct} we learned how to calculate the expectation value of the exponentials. Only the Jacobian determinant remains to be calculated. For the sphere, this determinant is easily evaluated using the infinitesimal generators of the M\"obius group. Remember that a M\"obius transformation (the CKG of the sphere) takes the form:
\begin{equation}
  \begin{pmatrix}
    a&b\\c&d
  \end{pmatrix} \in PSL(2,\mathbb{Z})\,:\qquad z \to \frac{a z + b}{c z + d}\,,\qquad  \bar  z \to \frac{a \bar z + b}{c \bar z + d}\,.
\end{equation}
This group has six real generators that act infinitesimally on the worldsheet coordinates as (see eqs.\ (\ref{BWeqconfkillsphere1},\ref{BWeqconfkillsphere2})):
\begin{align}
  \delta z &= \delta a_1 \,\partial + \delta a_2 \,z \partial + \delta a_3 \,z^2\partial\,,\nonumber\\
  \delta \bar z &= \delta \bar  a_1 \,\bar \partial + \delta \bar a_2\, \bar z \bar \partial + \delta \bar a_3\, \bar z^2\bar \partial\,.
\end{align}
The six real group parameters of the M\"obius group are denoted by three complex parameters $\delta a_i$, and we can immediately link to the discussion of the role of the CKVs of section \ref{BWss:WritinDownJacob} and equation \eqref{BWeq:deltaCKV}. The determinant of the  matrix  $C_r^a(\hat \sigma_j)$ (defined as in \eqref{BWeq:deltaCKV}) appearing in the Jacobian \eqref{BWeq:Jacob_ill} is easily computed as:
\begin{align}
  \det (C_r^a(\sigma_j)) &= \det \begin{pmatrix}
    1 &\hat z_1 &(\hat z_1)^2\\
    1 &\hat z_2 &(\hat z_2)^2\\
    1 &\hat z_3 &(\hat z_3)^2
  \end{pmatrix} \det \begin{pmatrix}
    1 &\hat {\bar z}_1 &(\hat {\bar z}_1)^2\\
    1 &\hat {\bar z}_2 &(\hat {\bar z}_2)^2\\
    1 &\hat {\bar z}_3 &(\hat {\bar z}_3)^2
  \end{pmatrix}\nonumber\\
  &= |\hat z_1-\hat z_2|^2|\hat  z_1-\hat z_3|^2|\hat z_2-\hat z_3|^2\label{BWeq:detCKVsphere}
\end{align}
We have written the fixed vertex positions, denoted before as $\hat \sigma_i$, as $\hat z_i,\hat {\bar z}_i$. 


We now piece together the results of the previous sections. For notational convenience, we drop the hats on the fixed positions $z_1,z_2,z_3$. The integral we want to evaluate is:
\begin{align}
&\langle \det (V^a_i(\hat \sigma_j) \prod_{i=4}^n \int \BWd^2 \sigma_i \sqrt{g} \prod_{j=1}^n \BWe^{\BWi k_j\cdot X} \rangle\label{BWeq:ntachyonampl}\\
& = \BWi g_s^{2+n} \tilde C_{S^2}(2\pi)^d \delta^d(\sum_i k_i)| z_1- z_2|^2| z_1- z_3|^2| z_2- z_3|^2 \int \BWd^2 z_4 \prod_{i,j =1\,,i>j}^4|z_i-z_j|^{\alpha' k_i\cdot k_j}\,.
\nonumber
\end{align}
We have put in an explicit $g_s$ dependence and we absorbed all other constants (except various factors of $2\pi$) into $\tilde C_{S^2}$.  The exponent contains a contribution $2$ from the Euler number $(\chi_{S^2}=-2$) and $1$ for each vertex operator inserted. Because the original integral was invariant under diffeomorphisms of the worldsheet and in particular under M\"obius transformations, also the gauge fixed path integral will be. Therefore we can freely pick values for the fixed positions $z_1,z_2,z_3$. We make the convenient choice $z_1 = 0, z_2 = 1, z_3 = \infty$. Also, remember that tachyons obey the mass-shell condition $k^2 = -4/\alpha'$.

\subsubsection{3-particle scattering}
Here we find a remarkably simple result. The amplitude \eqref{BWeq:ntachyonampl} for three tachyon scattering becomes:
\begin{equation}
A_{S^2}(k_1,k_2,k_3) = \BWi g_s^3 \tilde C_{S^2}(2\pi)^d \delta^d(\sum_i k_i)|z_1-z_2|^{2+\alpha' k_1\cdot k_2}|z_1-z_3|^{2+\alpha' k_1\cdot k_3}|z_2-z_3|^{2+\alpha' k_2\cdot k_3} \,,\label{BWeq:ntachyonampl3}
\end{equation}
This result seems at odds with M\"obius invariance: the amplitude cannot depend on the choice we make for $z_1,z_2,z_3$. To see that this dependence drops out, remember that these tachyon states are on-shell and obey $k_i^2 = 4/\alpha'$. Then we can rewrite the momenta in the exponents as:
\begin{align}
  2 k_1 \cdot k_2 &= (k_1 + k_2)^2 - k_1^2 -k_2^2\nonumber\\
  &= k_3^2 - k_1^2 -k_2^2 = -\frac4{\alpha'}\,,
\end{align}
where we used momentum conservation $(k_1 + k_2 + k_3 = 0)$ going to the second line, and by cyclicity the same result holds for $2k_2\cdot k_3$ and $2k_1\cdot k_3$. This leaves us with the result
\begin{equation}
  A_{S^2}(k_1,k_2,k_3) = \BWi g_s^3 \tilde C_{S^2}(2\pi)^d \delta^d(\sum_i k_i)
\end{equation}
for the 3-tachyon scattering amplitude.

\subsubsection{4-particle scattering: Virasoro-Shapiro}
We base our explanation on \cite{Polchinski:1998rq} and \cite{Tong:2009np}.
Start from \eqref{BWeq:ntachyonampl} and use momentum conservation to obtain the four-tachyon scattering as:
\begin{equation}
\BWi \tilde C_{S^2}(2\pi)^d \delta^d(\sum_i k_i) \int \BWd^2 z_4 |z_4|^{\alpha' k_1\cdot k_4}|1-z_4|^{\alpha' k_1\cdot k_4}\,,
\end{equation}
The constant in the amplitude is fixed by relating the four-point amplitude to the three-point amplitude through unitarity (i.e.\ by cutting a four-point amplitude up into a combination of two three-point interactions, see for instance Polchinski \cite{Polchinski:1998rq} for more details).
In field theory, scattering processes of four particles are often rewritten in terms of Mandelstam variables $s,t,u$ defined as:
\begin{equation}
  s = -(k_1+k_2)^2\,,\quad t= -(k_1+k_3)^2\,,\quad u = -(k_1+k_4)^2\,.
\end{equation}
with the property (as follows from momentum conservation):
\begin{equation}
 s +t + u = -16/\alpha'\,.
\end{equation}

To rewrite the 4-tachyon amplitude in terms of these variables, use the mass-shell condition to find the identity $2\alpha'k_i\cdot k_j = \alpha' (k_i +k_j)^2 - 8$. The integral in the amplitude then becomes:
\begin{equation}
  \int \BWd^2 z_4 |z_4|^{-\frac{\alpha' u}2 -4}|1-z_4|^{-\frac{\alpha' t}2 -4}\,.\label{BWeq:virasoro_integral}
\end{equation}
This integral has a lot of interesting properties. It is symmetric in $s,t,u$ and exhibits poles in these variables at exactly the closed string states of the spectrum. To make this obvious, we can recast the integral in a manifestly symmetric form using gamma functions. The gamma function can be defined as the following integral:
\begin{equation}
 \Gamma(a) = \int_0^\infty \BWd t\,t^{a-1} \BWe^{-t}\,.
\end{equation}
 It obeys the properties:
\begin{equation}
 \Gamma(x+1) = x \Gamma (x)\,,\qquad \Gamma(n+1) = n! \quad\forall n\in \mathbb{N}\,
\end{equation}
hence it can be seen as an extrapolation of the faculty operation for the integers ($n! = 1\cdot2\cdot3\cdots n)$ to a function on the complex numbers. However, from this expression is clear that the gamma function has poles at every negative integer (including zero). This is important for the structure of the four-point amplitude.

Using the gamma function, the integral \eqref{BWeq:virasoro_integral} can be written as:
\begin{equation}
 \int \BWd^2 z|z|^{2a-2}|1-z|^{2b-2}=\frac{\Gamma(a)\Gamma(b)\Gamma(c)}{\Gamma(1-a)\Gamma(1-b)\Gamma(1-c)}
\end{equation}
and we have our answer for the four-tachyon scattering amplitude:
\begin{equation}
  \BWi g_s^4 C_{S^2}(2\pi)^d \delta^d(\sum_i k_i) \frac{\Gamma(-1-\alpha's/4)\Gamma(-1-\alpha't/4)\Gamma(-1-\alpha'u/4)}{\Gamma{(2+\alpha's/4)}\Gamma(2+\alpha't/4)\Gamma(2+\alpha'u/4)}\,.\label{BWeq:VirasoroShapiro}
\end{equation}
This is the famous Virasoro-Shapiro amplitude.

\subsubsection{Discussion about the Virasoro-Shapiro amplitude}
There are a lot of interesting facts hidden in the amplitude \eqref{BWeq:VirasoroShapiro}. What we notice immediately is that it is symmetric in the $s,t$ and $u$ channels. We will come back to this later, but discuss first the properties focusing on one channel, say $s$. The amplitude has an infinite number of poles in $s$. Using the property of the gamma function:
\begin{equation}	
  \Gamma(z+1) = z\Gamma(z)\,,
\end{equation}
we see that the poles arise from the gamma functions in the numerator of \eqref{BWeq:VirasoroShapiro}. The first pole is at
\begin{equation}
  -1-\alpha's/4 = 0\,,\qquad {\rm or}\qquad s = -\frac4{\alpha}\,.
\end{equation}
This is the mass squared of the tachyon! We conclude that this pole corresponds to a resonance coming from an intermediate tachyon state. Continuing the analysis, there are a lot more poles at the values:
\begin{equation}
  -1-\alpha's/4 = n\,, n\in {\mathbb N}\qquad {\rm or}\qquad s = \frac4{\alpha}(n-1)\,.
\end{equation}
These are exactly the masses of all the states in the bosonic string spectrum, see table \ref{tab:LowestModes} in the introduction. With this in mind, we can reinterpret the string scattering amplitude for four tachyons as an infinite sum over these poles, expressing the intermediate creation of all these particles, see figure \ref{BWfig:FourPoint_BreakingUp}.
\begin{figure}[ht!]
\centering
\begin{picture}(150,0)
\put(0,20){{ $A_{S^2}(k_1,k_2,k_3,k_4)\quad=\quad$} {\huge $\sum_n$} }
\put(180,30){\Large$m_n$}
\end{picture}
\includegraphics[height=.07\textheight]{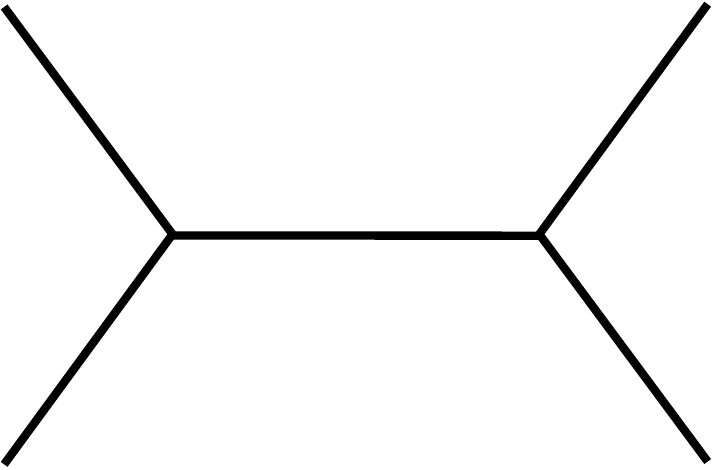}
 \caption{The four-point tachyon string amplitude can be seen as a sum over ordinary field theory amplitudes, one for each state in the string spectrum. The mass of each intermediate state is $(m_n)^2 = 4(n-1)/\alpha'$. (Figure inspired by \cite{Tong:2009np}.) \label{BWfig:FourPoint_BreakingUp}}
\end{figure}

When studying the residues of each pole in the amplitude, one sees that the highest power of momentum is
\begin{equation}
 A_{S^2}(k_1,k_2,k_3,k_4)\sim \frac{t^{2n}}{s - M_n^2}\,.\label{BWeq:4pt_sum_terms}
\end{equation}
The highest power of $t$ entering at each level, is related to the particle in the spectrum at level $n$ with the highest spin. Indeed, a particle $\Lambda$ with spin $J$ has indices $\Lambda_{\mu_1\ldots \mu_J}$. We need a contraction with $J$ derivatives in a cubic interaction, giving $J$ derivatives (or momenta in Fourier space popping op in the numerator of the amplitude) at each vertex. This tells us that the particle with highest spin at level $n$ has $J=2n$. Note that this is what we found earlier for level $n=1$ in the introduction, since there the highest spin particle is the graviton (spin 2). 

Now we return to the symmetry between the $s,t,u$ channels of the 4-tachyon amplitude. In the above discussion, we kept the $s$-channel fixed and wrote the amplitude as a sum over $t$. We could as well have done the opposite: keep $t$ fixed and sum over $s$ (and the same holds for the $u$-channel).
In field theory, we are used to calculating the contribution for all of these channels and summing them up. Somehow, we do not need to do this here, the result already exhibits an equivalence between these channels. In the early string literature, when people were trying to model the strong interactions, they even called their enterprise ``dual resonance model'' because of this $s,t,u$ duality.

Where does this duality come from? From a worldsheet point of view, it is understood easily. Since we do not make a distinction between topologically equivalent worldsheets, all three channels really are equivalent: we can reinterpret the infinite sum over $t$-amplitudes as one over $t$ or $u$-amplitudes. This is illustrated in figure \ref{BWfig:DUalResonance}. 

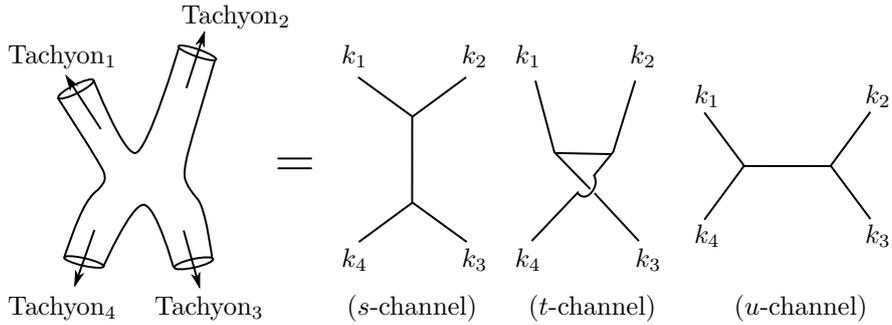
\begin{figure}[ht!]
\centering
\vspace{.025\textheight}
\begin{picture}(0,10)
\put(-15,-10){Tachyon$_4$}           
\put(40,-10){Tachyon$_3$}
\put(-15,85){Tachyon$_1$}
\put(50,100){Tachyon$_2$}
\put(85,42){\huge$=$}
\put(110,8){$k_4$}           
\put(155,8){$k_3$}
\put(110,85){$k_1$}
\put(155,85){$k_2$}
\put(112,-10){($s$-channel)}
\put(175,8){$k_4$}           
\put(220,8){$k_3$}
\put(175,85){$k_1$}
\put(218,85){$k_2$}
\put(180,-10){($t$-channel)}
\put(242,17){$k_4$}           
\put(306,17){$k_3$}
\put(242,70){$k_1$}
\put(306,70){$k_2$}
\put(257,-10){($u$-channel)}
\end{picture}
\includegraphics[height=.15\textheight]{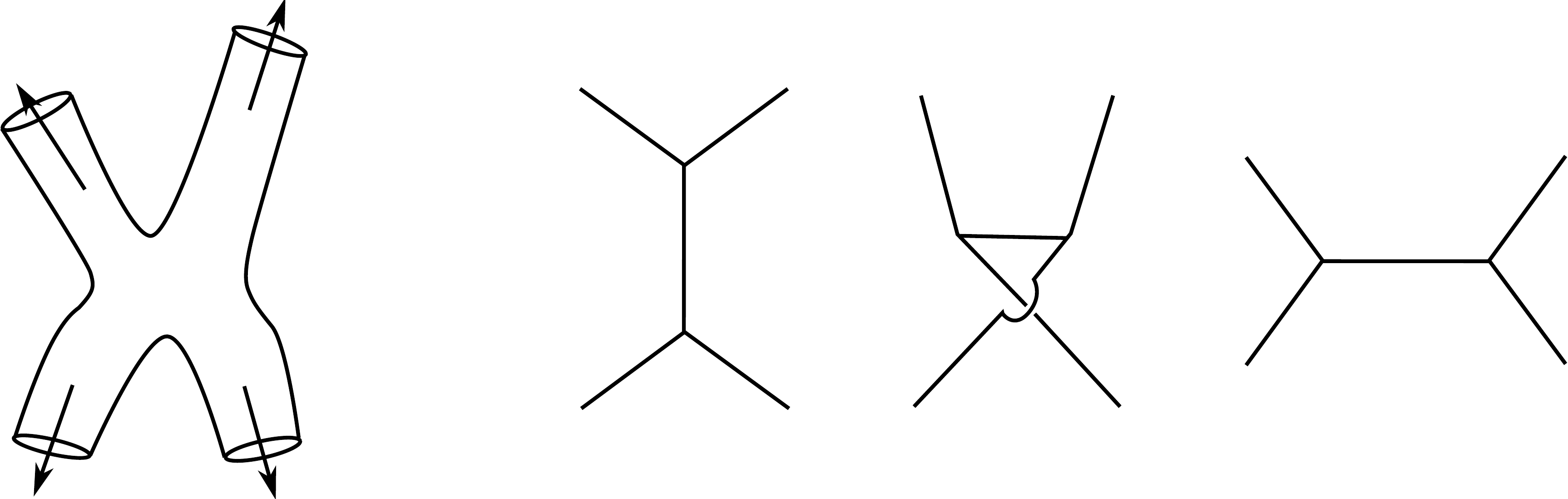}
\vspace{.01\textheight} 
\caption{The 4-tachyon amplitude in string theory on the Riemann sphere (worldsheet with legs to infinity, hence the arrows) corresponds to either choice of $s$, $t$ or $u$ channel, since we do not differentiate between worldsheets which are topologicaly equivalent in the path integral.\label{BWfig:DUalResonance}}
\end{figure}

As a third and final remark, the four-point tachyon amplitude has an improved high energy asymptotic behaviour, as compared to ordinary quantum field theory amplitudes. Let us examine the asymptotic behaviour of the amplitude at $s,t\to \infty$ with fixed ratio $s/t$ (the so-called hard scattering limit). This means we are looking at high energy scattering with the angle between incoming and outgoing states kept fixed. Using the asymptotic behaviour of the gamma function $\Gamma(x) \sim x \BWe^{-x} \sqrt{x}$ for $x \to \infty$ (Stirling's approximation), we see that the dominating behaviour is $\Gamma(x)\sim \BWe^{x\ln x}$ at large $x$. Using this in the 4-point amplitude, we find the following 
\begin{equation}
 A_{S_2}(k_1,k_2,k_3,k_4)|_{s,t\to\infty,\\ s/t {\rm \,fixed}} \sim g_s^2 \delta^4\left({\sum_i k_i}\right) \BWe^{-\alpha' (s\ln s +t\ln t +u\ln u)}
\end{equation}
We see the amplitude has an incredibly fast fall-off, faster than an exponential. This is much better than we could hope for in any field theory for point particles, as the best we can obtain there would be a power-law behaviour.  Even worse, when we take a look at an individual term in the expansion of the four-point amplitude \eqref{BWeq:4pt_sum_terms} in this high energy limit, we see that
\begin{equation} \frac{t^{2n}}{s-M_n^2} \to \infty\,,
\end{equation}
each individual term blows up! For instance, for the $n=1$ term, which corresponds to the graviton contribution, we encounter the well-known bad high energy behaviour that plagues many attempts to quantize gravity. We know that the dimensionless coupling in gravity is $G_N E^2$, with $G_N$ Newton's constant and $E$ a measure of the energy of the process, showing how amplitudes can scale with energy. This behaviour is recovered in the hard scattering limit, as then we have at $n=1$ (graviton in internal propagator): $t^2/s \sim t \sim E^2$. 

In conclusion, the four-point amplitude, as a sum of all these run away terms, is clearly better than its parts: it has a very good asymptotic behaviour. All the infinities seem to cancel each other off. This can be seen as an encouragement, a hint for the finiteness of string theory.

\section{Partition Function\label{BWs:PartitionFunction}}
The generating function $Z[J]$ is a very useful and important object in quantum field theory and string theory, as it enables us to calculate Green's function $G(\sigma, \sigma ')$. When we set the current $J$ to zero, we obtain the quantity $Z[0]$ known as the {\bf partition function},
\begin{eqnarray*}
Z = \sum^{\infty}_{{\rm g} = 0}  \int \frac{ \BWD g \BWD X }{{\cal N}}
e^{-S_P}.
\end{eqnarray*}
At first sight it is not clear why this quantity is so interesting
to study, since it merely comes down to calculating a scattering
amplitude without any vertex operator. Therefore, we shall first
argue why the partition function is interesting, after which we will
calculate the first two terms in the genus-expansion.

\subsection{The Meaning of the string Partition Function}
The name ``partition function" comes from the analogy with statistical mechanics. In statistical mechanics the partition function sums over all different energy levels. The partition function also encodes the distribution of  probabilities over the energy levels. In field theory or string theory the energy levels are represented by the masses of the particles or the  masses in the string spectrum. The string Partition Function can be interpreted as the sum of vacuum
diagrams in the genus-expansion,
\begin{eqnarray}
Z = \sum^{\infty}_{{\rm g} = 0} Z_{\rm g} ,
\end{eqnarray}
where $Z_{\rm g} \equiv \langle  {\bf 1} \rangle_{\rm g}$.
Technically we should only integrate over all inequivalent metric
and embeddings. Physically this corresponds to summing over all possible string excitations (thus
particle states) running in loops. The string partition
function is thus a measure for the vacuum energy density. When one is
interested in non-gravitational physics, the vacuum energy density
is (most of the times) irrelevant, as one uses the partition function as a normalization factor that can be divided out. But since string theory is interpreted as a gravitational theory, we can not ignore the string partition function. The vacuum energy density arising from string theory will couple to gravity and yields a contribution in the Einstein equations similar to a cosmological constant. It is therefore reasonable to interpret the string partition function as a measure for the spacetime cosmological constant.

\subsection{The Sphere}
To obtain $Z_{{\rm g} = 0}$ we should be able to perform the integral
$\BWD g$ over the world-sheet metric $g$. As noted in the previous
chapter \ref{BWc:Fixing}, this integral is not straightforward
due to the possible presence of moduli, CKV's and the overlap
between diffeomorphisms and Weyl-rescalings. In the case of the Riemann
sphere, there are no moduli and 6 CKV's, as pointed out in chapter \ref{BWc:Riemann}. Thus
choosing an appropriate gauge slice reduces $\BWD g$ to,
\begin{eqnarray*}
\BWD g = J \BWD \delta \phi\, \BWD \delta \sigma ,
\end{eqnarray*}
where $J$ represents the Jacobian, $\delta \sigma$ the diffeomorphism
degrees of freedom and $\delta \phi$ the Weyl-rescaling degrees of
freedom. When evaluating the jacobian correctly, we obtain the
following structure,
\begin{eqnarray*}
J \sim \det{}' (P_1^T P_1)^{\frac{1}{2}} \frac{d z_1 dz_2
dz_3}{Vol(CKV)},
\end{eqnarray*}
which together with the constant contribution of the $X$-integral,
section \ref{BWs:AmplitudesGenFunct}, vanishes due to the infinite
volume of the conformal Killing group.
In conclusion, there is no $S^2$-contribution to the partition
function, which corresponds to the interpretation that classical
value of the cosmological constant is zero.

\subsection{The Torus}
\subsubsection{The smart way}
To calculate the contribution $Z_{{\rm g} =1}$ of the torus, there
exist several equivalent methods which all lead to the same result,
when performed correctly. We shall limit ourselves to only two
methods. The first method which we shall discuss here starts from a
simplified picture. It hinges on some two-dimensional conformal field theory terminology. If these are unfamiliar to you, you can consult Raphael Benichous lecture notes on CFT of the Modave School 2009, or read up in a textbook on string theory, e.g.\ \cite{Polchinski:1998rq}. We can model a closed string running in a
loop as a field theory on a circle with length $2\pi Re(\tau)$,
evolving over a Euclidean periodic time $2\pi Im(\tau)$. Translating
this into operator language, we translate the states along $Re(z)$
over a distance $2\pi Re(\tau)$ and perform a time translation over
time $2\pi Im(\tau)$,
\begin{eqnarray*}
Z_{{\rm g} = 1}(\tau) = Tr\left[ e^{2\pi i Re(\tau)P} e^{-2\pi Im
(\tau ) H} \right].
\end{eqnarray*}
Keeping in mind that we can express the momentum operator $P$ as $P
= L_0 + \bar{L}_0$ and the energy operator $H$ as $H = L_0 + {\bar
L}_0 - \frac{c + \tilde c}{24} $, we can rewrite this expression as,
\begin{eqnarray*}
Z_{{\rm g} = 1}(\tau) = Tr\left[ q^{L_0 - c/24} \bar q^{\bar L_0 -
\tilde c / 24} \right],
\end{eqnarray*}
with $q\equiv e^{2\pi i \tau}$. The trace is taken over the complete
(physical) Hilbert space of the closed bosonic string.
Since the left- and right-moving sector of the closed string
decouple, it is sufficient to focus on only one sector in the
Hilbert space. Let us focus on the left-moving sector, then the
right-moving sector will yield similar expressions with barred
quantities. When acting with $L_0$ on a physical state, we have to
permute the oscillators of the physical state with the (normal
ordered) oscillators of $L_0$. Depending on the level number $n$ of
the oscillator, we obtain a contribution of the form,
\begin{eqnarray*}
q^{n N_{i, n}},
\end{eqnarray*}
where $N_{i, n} $ indicates the number of oscillators present in the
physical state with level $n$ and spatial direction $i$. Since we
must sum over all physical states, $N_{i, n} $ can vary from zero to
infinity, which implies an additional sum over the $N_{i, n} $. We
must also pay attention to the zero-mode oscillators. They will
appear as a momentum-integral. The full expression thus reads,
\begin{eqnarray*}
q^{-c/24} \bar q^{- \tilde c/24} V_{26} \int \frac{d^{26}
k}{(2\pi)^{26}} e^{-\alpha' \pi Im (\tau)k^2 }
 \prod_{i=1}^{26} \prod_{n=1}^{\infty} \sum_{N_{i, n}=0}^{\infty}
 q^{n N_{i, n}} \sum_{\bar N_{i, n}=0}^{\infty}\bar q^{n \bar N_{i,
n}},
\end{eqnarray*}
where the spacetime volume $V_{26}$ is a normalization factor,
coming from the transition of a discrete sum over momenta $k$ to a
continuous integral.
The central charges $c$, $\tilde c$ are equal to each other and
count the number of spacetime dimensions, so we have $c = \tilde c =
26$. We can recognize a geometric sum in this expression,
\begin{eqnarray}
\sum_{N_{i, n}=0}^{\infty}q^{n N_{i, n}} = \frac{1}{1-q^n},
\end{eqnarray}
and also the Dedekind eta-funcion $\eta(\tau)$,
\begin{eqnarray}
\eta (\tau) \equiv q^{1/24} \prod_{n = 1}^{\infty} (1-q^n).\label{BWeq:Dedekind}
\end{eqnarray}
Performing the integral over the momenta,
\begin{eqnarray*}
\int \frac{d^{26} k}{(2\pi)^{26}} e^{-\alpha' \pi Im (\tau)k^2 } =
\left(\frac{1}{(4\pi^2 \alpha' Im (\tau))^{1/2}}\right)^{26},
\end{eqnarray*}
we obtain the following expression for $Z_{{\rm g}= 1}(\tau)$,
\begin{eqnarray}
Z_{{\rm g} = 1}(\tau) = V_{26} \frac{1}{(4\pi^2 Im(\tau))^{13}}
|\eta(\tau) |^{-52}. \label{BWeq:almostPart}
\end{eqnarray}
At this point, we are only interested in the $\tau$-dependence of
$Z_{{\rm g} = 1}(\tau)$. Therefore we shall not pay any attention to
the (constant) volume $V_{26}$ of the spacetime. Although our
expression eq.~(\ref{BWeq:almostPart}) seems to be interpretable, we
first need to straighten out some crucial points about this
expression. First, as we have mentioned earlier, we should sum over
all physical states. However, the physical Hilbert space is smaller
than the Hilbert space we have summed over, as only the tranverse
directions of the modes $n>0$ contribute\footnote{This can be explained as follows.
Looking at the Polyakov action we start from 26 scalars $X^\mu$ and a two-dimensional metric $g_{mn}$. In total we have 29 bosonic degrees of freedom. The local symmetries (reparametriations + Weyl) allow us to gauge away the three degrees of freedom of the metric, such that the world-sheet metric is flat. The equations of motion for the metric then reduce to two constraints, that relate the light-cone coordinates $X^\pm$ to the transverse coordinates $X^i$. The two constraints eliminate two more degrees of freedom. Hence, 24 is the total number of physical degrees of freedom. The BRST-formulation offers another way to see this. The gauge degrees of freedom of the world-sheet metric can be translated in ghosts, summarized as the $bc$-system. The total action then reads ${\cal S} = {\cal S}_{X} + {\cal S}_{bc}$. To calculate the vacuum partition function we should also take the contribution from the $bc$-CFT into account, which is exactly $|\eta(\tau)|^4$.
} to the sum. We can correct this quite easily by multiplying our expression by $ |\eta(\tau)|^4$.
Secondly, to obtain the entire genus 1 contribution we must ensure
to sum only once over every modulus $\tau$. Therefore we must divide
by the volume of the conformal Killing group ($4\pi^2 Im (\tau))$
and integrate over the fundamental domain ${\mathfrak F}$ of the moduli
space.
Finally, we obtain the following expression for the partition
function at genus 1,
\begin{eqnarray}
Z_{{\rm g} = 1} \sim \int_{\mathfrak F} \frac{d^2 \tau}{(Im(\tau))^2}
\frac{1}{(4\pi^2 Im(\tau))^{12}} |\eta(\tau) |^{-48}.\label{BWeq:1loopeasy}
\end{eqnarray}

\subsubsection{The hard way}
The hard way uses the gauge fixing procedure we have elaborated in
the previous chapter \ref{BWc:Fixing}. This means that we have to
choose an appropriate gauge slice such that metric measure $\BWD g$
can be rewritten in terms of the modulus $\tau$, infinitesimal
diffeormophism $\delta \sigma_a$ (non-CKV) and infinitesimal
Weyl-rescaling $ \delta  \phi $,
\begin{eqnarray}
\BWD g = J(\tau)  d^2 \tau \,  \BWD \delta \phi\, \BWD \delta \sigma,
\end{eqnarray}
where $J(\tau)$ expresses the Jacobian relating the old degrees of freedom
with the new. The Jacobian $J(\tau)$ can be fully calculated as in \cite{D'Hoker:1988ta} and the calculation basically comes down to determining $\det ' (\nabla) $ on the torus. We refer to subsection \ref{BWref:torus} for this calculation and limit ourselves to stating the final result for the one-loop string partition function,
\begin{eqnarray}
Z_{{\rm g} = 1} = \int_{\mathfrak F} \frac{d^2 \tau}{8\pi^2 (Im(\tau))^2}
\frac{1}{(4\pi^2 Im(\tau))^{12}} |\eta(\tau) |^{-48}, \label{BWeq:1loophard}
\end{eqnarray}
where we integrate over a fundamental domain ${\mathfrak F}$ of the
moduli space ${\mathfrak M}_{T^2}$.

\subsubsection{Modular invariance}

Let us take a moment to focus on a particular property of the one-loop partition function eqs.~(\ref{BWeq:1loopeasy}) and
(\ref{BWeq:1loophard}). We already noticed in the chapter \ref{BWc:Riemann} that the modulus $\tau$ is determined up to the modular group $PSL(2,\BWIZ)$. This discrete group is generated by a succession of the elements T and S,
\begin{eqnarray*}
T: \tau &\mapsto& \tau + 1,\\
S: \tau &\mapsto& -1/\tau.
\end{eqnarray*}
We shall now investigate how the partition function behaves under the transformations T and S. Let us start with the volume-element of the moduli space $ \frac{d^2 \tau}{(Im(\tau))^2} $, as both T and S leave this volume-element invariant:
\begin{eqnarray*}
\begin{array}{lcl}
& T &\\
\frac{d^2 \tau}{(Im(\tau))^2} & \longrightarrow & \frac{d^2 \tau}{(Im(\tau +1))^2} = \frac{d^2 \tau}{(Im(\tau))^2}\, ,\\
& S &\\
&\longrightarrow & \frac{d \tau}{\tau^2} \frac{d \bar \tau}{\bar{\tau}^2} \frac{1}{(Im(-1/\tau))^2} = \frac{d^2 \tau}{(Im(\tau))^2}\, .
\end{array}
\end{eqnarray*}
Next we investigate how the Dedekind-$\eta$-function transforms. In case of the T-transformation, it suffices to take the definition of the Dedekind-$\eta$ function eq.~(\ref{BWeq:Dedekind}),
\begin{eqnarray*}
\begin{array}{lcl}
&T& \\
\eta (\tau)& \longrightarrow &e^{2\pi i/ 24} q^{1/24} \prod_{n=1}^{\infty} \big(  1- q^n e^{2\pi i n } \big)  = e^{\pi i / 12}\, \eta(\tau).
\end{array}
\end{eqnarray*}
When we take the modulus of the $\eta$-function, this part of the partition function remains invariant under a T-transformation. The S-transformation is somewhat harder to figure out. Here we will use the definition of the Dedekind $\eta$-function in terms of $\theta$-functions, as given in appendix 10.A in \cite{DiFrancesco:1997nk},
\begin{eqnarray}
\big(\eta (\tau) \big)^3 \equiv \frac{1}{2} \theta_2 (\tau) \theta_3(\tau) \theta_4(\tau).\label{BWeq:ThetaFunctions}
\end{eqnarray}
One can show using the Poisson resummation formula that the $\theta$-functions transform as follows under S,
\begin{eqnarray*}
\theta_2 (-1/\tau ) &=& \sqrt{-i \tau}\, \theta_4 (\tau), \\
\theta_3 (-1/\tau ) &=& \sqrt{-i \tau} \,\theta_3 (\tau), \\
\theta_4 (-1/\tau ) &=& \sqrt{-i \tau}\, \theta_2 (\tau).
\end{eqnarray*}
This implies that the Dedekind $\eta$-function will behave as,
\begin{eqnarray*}
\begin{array}{lcl}
&S& \\
\eta (\tau) &\longrightarrow & \sqrt{-i \tau} \, \eta(\tau).
\end{array}
\end{eqnarray*}
As $\eta$ is raised to the power -48 after taking the modulus, it is clear that this part of the partition function will not remain invariant under an S-transformation. Luckily there still is an additional part to consider, namely $\frac{1}{(Im(\tau))^{12}}$. This part is invariant under a T-transformation, but transforms non-trivially under S,
\begin{eqnarray*}
\begin{array}{lcl}
&T& \\
\frac{1}{(Im(\tau))^{12}} & \longrightarrow & \frac{1}{(Im(\tau +1))^{12}} = \frac{1}{(Im(\tau))^{12}},\\
&S& \\
&\longrightarrow & \frac{1}{(Im(-1/\tau))^{12}} = \frac{|\tau|^{24}}{(Im(\tau))^{12}}.
\end{array}
\end{eqnarray*}
Hence, we can show that the one loop partition function is invariant under the modular group. This invariance is called modular invariance. It also implies that the upper half plane is too large to use as integration region. We need to restrict the integration to a fundamental domain ${\mathfrak F}$, see figure \ref{BWref:Funddom}, which is usually taken to be the canonical one,
 \begin{eqnarray*}
{\mathfrak F} \equiv \{ \tau \in {\cal U}  \big| \,-\frac{1}{2} \leq Re(\tau) \leq \frac{1}{2}, \, |\tau| \geq 1 \}\,.
 \end{eqnarray*}
Modular invariance is an extremely important concept in string theory, as it identifies two equivalent tori. Imposing modular invariance then implies that only one of the two tori is taken into account. In the next chapter we will see that modular invariance is also a useful tool to construct consistent string theories, e.g. for the Heterotic string theory.









\chapter{Above and Beyond our Horizons\label{BWc:Conclusions}}
{\it In chapter \ref{BWc:Riemann} we have discussed the close connections
between Riemann surfaces and conformal structures. In chapter
\ref{BWc:Fixing} we set up the general ideas to calculate string
scattering amplitudes, where we encountered the problem of gauge
fixing the Polyakov path integral. Choosing an appropriate gauge
slice allowed us to calculate string scattering amplitudes and the
partition function on a sphere (genus 0) and a torus (genus 1) in
chapter \ref{BWc:Amplitudes}. This last chapter is devoted to subjects
based on the main themes of the previous chapters, but whose technical 
analysis is beyond the aim of these lectures. We start this last chapter with
a review of the difficulties one encounters when going to Riemann
surfaces with higher genus in a first section  and when one considers scattering
superstrings instead of bosonic strings in a second section. In a third section we
 briefly discuss how the low energy limit of string theory relates to
an effective field theory. In the last section we say a few words 
about divergences in string theory and how they should be interpreted.}

\section{Higher genus}
Before discussing the higher genus amplitudes, let us summarise some useful 
facts and concepts from the previous chapters.
The Riemann Sphere and the torus are very special Riemann surfaces,
in the sense that they are compact Riemann surfaces with a
continuous conformal Killing group. For higher genus the compact
Riemann surfaces have a discrete conformal Killing group. Hence, the
gauge fixing ambiguities related to the conformal Killing vectors
are absent for Riemann surfaces with genus ${\rm g} \geq 2$. By all
means this does not imply that the gauge fixing procedure simplifies. On the
contrary, for every extra handle we introduce, we increase the
number of moduli (describing infinitesimal deformations of different conformal structures on the Riemann surface). Recall that
the dimensionality of the moduli space is related to the
genus g, i.e. $dim_\BWIC\, {\mathfrak M}_{\rm g}= 3 {\rm g} - 3$.
To treat higher order string amplitudes we should have a theory of
moduli spaces for Riemann surfaces at our disposal. The moduli
spaces of Riemann surfaces are subject to mathematical research,
dating back to the great Bernhard Riemann himself. It is not our
intention to discuss the treatment of higher order genus moduli
spaces in detail here. Instead we  give a compact overview of some
attempts to tackle these higher order scattering amplitudes.

Let us first choose appropriate (orthogonal) variables describing
all possible variations of the metric,
\begin{eqnarray}
\delta g_{m n} = \{{\rm Weyl }, \delta\phi\}\oplus \{{\rm Diff_0},
\delta \sigma\} \oplus \{{\rm moduli}, \delta m_\alpha \}.
\end{eqnarray}
In terms of these variables we can write the metric
measure\footnote{An appropriate normalization in the path-integral makes sure that the
measure preserves Weyl-invariance, provided that the number of
spacetime dimensions is 26. The integration of the reparametrizations $\delta \sigma$ and the Weyl-transformations $\delta \phi$ will cancel with this normalization factor for a Weyl-invariant theory. A thorough discussion of the normalization and Weyl-anomaly cancellation can be found
in \cite{D'Hoker:1988ta} and references therein.} (symbolically) as,
\begin{eqnarray}
D g_{mn} = J(m^\BWmodPara) D \sigma D \phi \prod_{\alpha=1}^{6g-6} dm^\BWmodPara \, ,
\end{eqnarray}
where the jacobian $J$ accounts for the coordinate transformation we
have performed. We notice that the metric measure basically reduces to an
integral over the moduli space\footnote{To be more precise, the metric measure reduces to an integral over the Teichm\"{u}ller space, when there are no perturbative (gravitational) anomalies due to the breakdown of diffeomorphism invariance (the so called Diff$_0$ ($M$) anomalies). The large reparametrizations  are represented by the mapping class group $ {\rm Diff}^+(M) / {\rm Diff}_0(M)$ and can cause global (gravitational) anomalies. When also these global anomalies are absent, the metric measure reduces to the integral over the moduli space. In case of the bosonic string reparametrization invariance remains manifest throughout the calculations and it is clear that we should not be worried about anomalies. In case of the superstring one should pay special attention to anomaly cancellation. More information about anomaly cancellation can be found in the review paper of D'Hoker and Phong  \cite{D'Hoker:1988ta} and references therein.}. Therefore we must find suitable coordinates to parameterize
the moduli space. We have seen in chapter \ref{BWc:Riemann} that the
moduli space of a Riemann surface with genus ${\rm g} \geq 2$ is given by
an orbifold space (Teichm\"{u}ller space), a complex manifold moded
out by a non-freely acting discrete group $\Gamma_{\rm g}$ (the
mapping class group). Hence, finding suitable coordinates to cover
this orbifold space is a highly non-trivial task, but indispensable
to perform the integral over the moduli space. We shall not try to
give a review of the many techniques that have been invented to
parameterize the moduli space. Instead we refer to the review paper
of D'Hoker and Phong \cite{D'Hoker:1988ta}, and references therein.

Once we have an appropriate parametrization of the moduli space, we
need to evaluate the jacobian $J$ in these coordinates. This implies
we should calculate various determinants with respect to these
coordinates on the moduli space. Next, we must also
perform the integral over the string embedding functions $X^\mu$.
One notices that the integral reduces to the following form,
\begin{eqnarray}
\int \BWD X\, \BWe^{-S[X, \hat g]} \sim \det \,'(\nabla)^{-d/2}\,
\BWe^{\int\BWd^2 \sigma \BWd^2\sigma' J(\sigma) G'(\sigma,\sigma')
J(\sigma')} ,
\end{eqnarray}
which can be computed once we know the determinant of the laplacian
and the Green function $G(\sigma_1,\sigma_2)$ on the Riemann surface
of genus g, analogous to what we discussed in the case of the
Riemann sphere (${\rm g} =0$) and the torus (${\rm g} =1$). The Riemann surface of
genus ${\rm g}\geq 2$ is conformally equivalent to the upper half plane
modded out by a freely acting (discrete) group (as we know from the
Uniformization theorem). In order to calculate the Green function,
we must find appropriate coordinates covering the Riemann surface
and in which we can express the metric on the Riemann surface. Of
course, there is a close connection between the expression of the
metric on the Riemann surface and the moduli related to this metric.
Choosing suitable coordinates to parameterize the Riemann surface
can be quite helpful to find suitable coordinates to describe the
moduli space. The explicit procedure to calculate the determinant of
the laplacian and the Green function is given in chapter V of
\cite{D'Hoker:1988ta}.

In the literature, a lot of attention was paid to multiloop vacuum
amplitudes\footnote{The generalization to scattering amplitudes is
performed by introducing the appropriate vertex-operators.}. In
\cite{Belavin:1986cy} the general structure of $p$-loop amplitudes for
the closed oriented bosonic string in 26 spacetime dimensions is
investigated. Those authors are able to parameterize the moduli space using
coordinates which are related to the so-called
Beltrami-differentials. The Beltrami-differentials $\mu$ are
complex differentials, which describe the infinitesimal deviation
from a flat metric,
\begin{eqnarray}
ds^2 = \rho | dz + \mu d\bar z|^2.
\end{eqnarray}
As the number of different conformal classes $[g]$ is equal to 3g $-3$,
we can choose a basis $\mu_i$ of 3g $-3$ Beltrami-differentials ($i
\in \{1, \ldots, {\rm 3g-3} \}$). An arbitrary conformal class can
then be presented by a general Beltrami differential of the form,
\begin{eqnarray}
\mu = \sum_{i=1}^{\rm 3g - 3} y_\alpha \mu^\alpha.
\end{eqnarray}
The coordinates $y_\alpha$ parameterize the moduli space. With this
construction the authors \cite{Belavin:1986cy} are able to show that the vacuum partition
function can be written as the modulus of a holomorphic function
which only depends on the coordinates $y_i$. Furthermore, the
authors claim that any multiloop amplitude in any conformal
invariant string theory can be written in terms of purely algebraic
objects (such as holomorphic functions) on the moduli spaces of the
Riemann surfaces. In \cite{Belavin:1986tv, Moore:1986rh,
Kato:1986wj}, the two-loop vacuum amplitude was calculated explicitly
for the closed bosonic string in 26 spacetime dimensions. For genus-2 
Riemann surfaces the moduli can be represented by a special type
of matrices (period matrices), which allows to be more explicit
in the case of two-loop amplitudes. The general statement resulting
from these three papers teaches us that the two-loop vacuum
amplitude can be given in terms of a product of even
$\theta$-functions\footnote{{These are the same $\theta$-functions as in eq.\ \eqref{BWeq:ThetaFunctions}. For more on $\theta$-functions, see e.g.\ appendix 10.A of \cite{DiFrancesco:1997nk}.}}. For
more information about higher genus amplitudes we refer to the above
given references and the references therein. Other approaches to
calculate string amplitudes, such as the formulation with Faddeev-Popov ghosts or methods based on the CFT structure, can be found in the literature, e.g. \cite{Polchinski:1998rq, D'Hoker:1988ta, Friedan:1985ge} and references therein.

\section{The Superstring}
In these lectures we have focused on (closed) bosonic string theory,
as it offers the simplest model to explain the problems related to
gauge fixing and the calculations of scattering amplitudes. But bosonic string theory itself raises some issues
which are not favorable from a phenomenological point of view. The
fact that the spectrum contains a tachyon indicates that the bosonic
theory is not quantized around a proper, stable vacuum. Furthermore, fermions do
not appear naturally in the spectrum of the bosonic string. So, there is no
stringy description of scattering fermions (which make up the observed matter), when we restrict to
bosonic string theory. A last phenomenological concern is the number
of spacetime dimensions, i.e. 25 spatial dimensions and one time
dimension. These three issues can actually be handled at the same
time by introducing more symmetry on the string worldsheet, namely
supersymmetry. Supersymmetry is a symmetry which relates bosonic
degrees of freedom with fermionic degrees of freedom.  The introduction 
of supersymmetry implies the existence of fermions. When
quantizing the bosonic and fermionic degrees of freedom, one can
then construct states in the spectrum that correspond to spacetime
fermions. Moreover, a consistent spectrum of the supersymmetric
string (superstring) does not contain any tachyons, as these are
truncated out by the GSO projection \cite{Gliozzi:1976qd}. The presence of fermionic
degrees of freedom on the worldsheet also alters the quantization 
of the string in such a way that the critical dimension recudes from 26 to 10 spacetime 
dimensions. 

As is well-known, five consistent superstring theories in ten-dimensional Minkowski spacetime have 
been constructed: Type I, Type IIA, Type IIB, Type HO and Type HE. 
Every string theory has its own specific spectrum, depending on the construction. 
Type I string theory contains both closed and open (unoriented) strings. 
The other superstring theories only contain (oriented) closed strings. In this 
discussion about string amplitudes we shall continue to focus on closed string 
theories. Type IIA and Type IIB do not differ that much, such that we can talk about 
them together as Type II superstring\footnote{Let us remind the reader that we are considering 
strings propagating in a Minkowski spacetime with the correct critical dimension ($D=10$). 
In that case Type II superstring theory only contains closed strings. About fifteen years ago it was realized \cite{Polchinski:1995mt} that
Type II superstrings can also consist of open strings attached to dynamical hyperplanes, called D-branes. As Type II Superstrings also contain Ramond-Ramond (RR)-fluxes in their spectrum, they should couple to RR-charged objects. The D-branes are the desired objects to which the Type II Superstring couples. This means that we can also consider Type II Superstrings in a background with D-branes. Close to the D-branes the superstrings will couple to the D-branes and their (massless) degrees of freedom will describe the field theory living on the D-branes. However, far away from the D-branes the spacetime is Minkowski and only closed strings will propagate in those regions.}. 
Type HO and Type HE have the same underlying 
construction scheme as well and we can denote them together as the Heterotic string. 
For detailed analysis of the different types of superstring theory we refer to the 
literature. We  briefly review some of the basic properties of the Type II 
superstring and the Type Heterotic string\footnote{More technical information about Type II superstring and Heterotic superstring theory will be postponed to the following section.} and discuss the relation between the 
superstring theories and low effective field theories in the next section.
    
Before we start to discuss the features of superstring amplitudes, we must 
note that there exist other formalisms to describe the superstring theories. 
In the Ramond-Neveu-Schwarz-formalism (RNS) \cite{Ramond:1971gb, Neveu:1971rx, Neveu:1971iv} one introduces supersymmetry on the worldsheet. Target-space supersymmetry is not manifest in this formalism and one should explicitly construct conserved supercharges in the target-space to show spacetime supersymmetry \cite{Friedan:1985ge}.  However, one might prefer spacetime supersymmetry to be manifest. The Green-Schwarz-superstring \cite{Green:1983wt} offers a formalism in which the target space supersymmetry is manifest. On the worldsheet we have the usual local conformal symmetries supplemented with a local fermionic symmetry, so-called $\kappa$-symmetry. Due to this local $\kappa$-symmetry quantizing the Green-Schwarz-
superstring becomes problematic by itself. One can use the local symmetries to make appropriate gauge choices to eliminate the gauge degrees of freedom and perform a consistent quantization. This can only be done in light-cone gauge, by which we have to give up manifest Lorentz-invariance.  Another way to get rid off the $\kappa$-symmetry is to couple additional ghost-fields to the Green-Schwarz-string which break the $\kappa$-symmetry. In ten-dimensional Minkowski spacetime this can certainly be done. The ghost-fields have to satisfy additional constraints, which make them pure spinors, hence the name pure spinor formalism \cite{Berkovits:2000fe}. As it is our intention to briefly discuss superstring amplitudes, it is easiest to generalize the techniques we have set out for the bosonic string for the purpose of superstrings. The RNS-formalism is most suitable to perform this generalization. Nonetheless, string amplitudes can also be calculated in the other two formalisms.

Without further ado let us focus now on the RNS-formalism to describe the superstring.
The supersymmetry on the worldsheet can be made manifest using superspace\footnote{A complete dissertation about the superspace formulation of the RNS superstring can be found in \cite{D'Hoker:1988ta}.}.
This implies we should introduce anti-commuting coordinates ($\theta, \bar \theta$) 
on the string worldsheet. The superstring worldsheet is then seen as a super-Riemann surface. 
The degrees of freedom living on the super-Riemann surface are the scalar functions $X$ (string embedding functions), spin 1/2 fields $\psi$ (as superpartners of the embedding functions), the worldsheet metric $g_{ab}$ and a spin 3/2 field representing the gravitino $\xi$ (superpartner of the metric)\footnote{There might also be auxiliary fields in order for the supersymmetry algebra to close off-shell. If one imposes the equations of motion, the auxiliary fields will not appear in the component action. We will therefore pay no attention to them.}. The fields $X$ and $\psi$ can be taken together in a superfield $\Phi$, while the metric and the gravitino are combined into a superzweibein ${E_M}^A$ and a super-connection $\Omega_M$, with which one can define a covariant super-derivative $\BWID$. The partition function\footnote{To fully specify the super-geometry of a super-Riemann surface one should also impose additional constraints on the super-torsion and super-curvature. As a consequence of these constraints, the super-connection can be given in terms of the superzweibein. The constraints itself also need to be introduced in the partition function through a delta-function. The torsion constraints complicate the integration over the super-geometries. We shall not go in more detail here and assume that the constraints have been solved such that we are working only with the independent degrees of freedom. } at g loops can then be written (formally) as,
\begin{eqnarray}
Z_{\rm g}^{\rm super} = \int \BWD {E_M}^A\,  \BWD \Omega_M\,  \BWD \Phi\,  \BWe^{- {\cal S}[\Phi, {E_M}^A]},
\end{eqnarray}
with the action given by,
\begin{eqnarray}
{\cal S}[\Phi, {E_M}^A] = \frac{1}{4 \pi \alpha'} \int d^2 z d\theta d\bar \theta\,  E \, \BWID^a \Phi^\mu \BWID_a \Phi_\mu + \frac{\lambda}{4 \pi} \int d^2 z \sqrt{g} R.  
\end{eqnarray}
We introduced the notation $E$ as the superdeterminant of ${E_M}^A$. The full partition function is then given by summing over 
all possible topologies of the super-Riemann surface,
\begin{eqnarray}
Z^{\rm super} = \sum_{{\rm g} = 0}^{\infty} Z_{\rm g}^{\rm super}.
\end{eqnarray}
Let us first comment about the integration over the superfield $\Phi$, by which we integrate over the bosonic fields $X$ as well as over the fermionic fields $\psi$. Here a first ambiguity arises due to possible spin structures one can define on a Riemann surface. The parallel transport of a spinor along a closed curve renders the same spinor up to a phase shift. Spinors can therefore only be defined on manifolds for which there is no topological obstruction to make a consistent phase choice. The possible topological obstruction is represented by the second Stiefel-Whitney class. Luckily, for oriented surfaces the second Stiefel-Whitney class is zero and we can make two consistent choices for a phase shift. The spin structure with phase shift $\BWe^{i \pi}$ are called odd, while the spin structure with phase shift $\BWe^{i 0}$ is labelled even. Essentially the spin structure represents the periodicity relation of a spinor when transporting the spinor along a closed curve. Let us take the torus as an example, where the torus is defined by the relations $z\sim z +1$ and $z \sim z + \tau$. Then for each periodic direction the spinor can be periodic (odd $\BWe^{i \pi}$) or anti-periodic (even $\BWe^{i 0}$),
\begin{eqnarray}
\begin{array}{ll}
\psi(z+1) = - \psi (z)\, \BWe^{i \pi}, & \psi(z + 1) = - \psi (z)\, \BWe^{i 0} ,\\
\psi(z+\tau) = - \psi (z)\, \BWe^{i \pi}, & \psi(z + \tau) = - \psi (z)\, \BWe^{i 0}.
\end{array}
\end{eqnarray} 
Hence, we have four different spin structures in case of the torus: odd-odd, even-odd, odd-even and even-even.
When we consider  a Riemann surface with genus g, then there exist 2g non-trivial closed curves (cycles). Thus on a Riemann surface with genus g there exist $2^{\rm 2 g}$ spin structures. When integrating over the fermionic fields, one must sum over all spin structures as well. The treatment of even spin structures is different from the one of odd spin structures, since odd spin structures also give rise to zero mode contributions of the Dirac-operator. Both for even and odd structures, the functional integrals give rise to (different) functional determinants, that can be computed if we know the eigenfunctions and eigenvalues of the laplacian and Dirac-operator on the Riemann surface.

Let us now focus on the integration over the metric and the gravitino $\xi$. Similar to the metric in bosonic string theory, the (independent components of the) superzweibein (are) is determined up to local symmetries: reparametrizations, local supersymmetry, Weyl-transformations and super-Weyl transformations. This means that we should make sure not to overcount, when integrating over the superzweibein. Similar to the bosonic case, we should make an appropriate choice for a gauge slice to count every superzweibein only once. We are thus led to the construction of superconformal classes for the superzweibein. When integrating over the superzweibein, we should only count the representative of the superconformal class. The number of  superconformal classes depends on the super-genus of the super-Riemann surface. One can smoothly deform a representative superzweibein to end up in another inequivalent superconformal class. These deformations form the supermoduli space of the super-Riemann surface. The superzweibein (or better the representative of the superconformal class) itself will have an additional invariance under transformations generated by super conformal Killing vectors, which differ according to the genus of the super-Riemann surface.  The variation of the superzweibein can be split into three different pieces\footnote{To be correct, we should also take into account a local $U(1)$ symmetry. However in order not to overload our dissertation with technical details, we will omit this additional local symmetry.},
\begin{eqnarray}
\delta {E_M}^A = \{ {\rm super-diffeo} \} \oplus \{ {\rm super-Weyl}\}  \oplus \{ {\rm super-moduli}\}.
\end{eqnarray}      
While other techniques (such as the BRST-formulation \cite{Verlinde:1987sd}) were much more fruitful in the eighties to discuss the superstring amplitudes, we will stick with the strategy to find an appropriate gauge slice, analogous to our treatment of the bosonic closed string. Hence the calculation of a superstring scattering amplitude reduces to finding an appropriate gauge slice and appropriate coordinates for the supermoduli space. However, pinning down an appropriate gauge slice and appropriate coordinates for the supermoduli space  are more complicated for the superstring, especially when the superstring lives on a Riemann surface with genus 2 or higher. Many proposals have been put forward to tackle the problems of finding an appropriate gauge slice and appropriate coordinates for supermoduli space. And for a long time it remained unclear whether it was possible to find such a gauge slice and whether the resulting amplitudes reveal properties like holomorphicity, modular invariance, etc. It is definitely beyond the reach of these lectures to overview the different proposals to overcome the ambiguities related to this method of gauge-fixing. In recent times D'Hoker and Phong (see e.g. \cite{D'Hoker:2001nj, D'Hoker:2002gw}) proposed a technique to consistently construct an appropriate gauge for Riemann surfaces with genus equal to or higher than 2. The technique works both for vacuum amplitudes and for $n$-point amplitudes.      

Let us for one moment forget the difficulties related to the gauge-fixing and not dwell over the different approaches to calculate the
superstring amplitudes. Then we are able to give a brief overview of some of the superstring amplitudes that have been calculated in the literature. The zeroth and first order amplitudes for the Type II Superstring are worked out explicitly in e.g. \cite{D'Hoker:1988ta} using the covariant formulation of the Polyakov string (and the BRST-calculation). These results are similar to the original results of Green and Schwarz \cite{Green:1981yb}, where they used the operator formalism. The zero, one and two point functions on the (super-)Riemann sphere vanish due to superconformal invariance. Analogous to the conformal Killing group in bosonic string theory, the super-conformal Killing group has an infinite volume. Hence, the amplitude vanishes unless we can fix at least three vertex operators, in which case the freedom of the super-conformal Killing vectors is completely fixed. 
Remarkably one can argue and show that the first non-trivial one loop amplitude appears for six external legs. Or equivalently, zero, one, two, three, four and five point amplitudes on the torus vanish. For the Heterotic Superstring the zeroth and first order amplitudes are calculated by the original creators in \cite{Gross:1985rr}, in to the operator formalism. It was shown that up to one-loop the partition function vanishes, indicating that the cosmological constant for the heterotic theories vanishes, similar to the Type II Superstring.

\section{String Theory as a Low Energy Effective Field Theory}
In the discussion of four-point tachyon scattering, given in subsection \ref{BWss:>3_Vertex_Ops}, we noticed that the Virasoro-Shapiro amplitude exhibits poles for certain values of $s$ (for a fixed $t$). A closer look at these poles tells us that the poles exactly correspond to the mass states of the closed string. At such a pole we can imagine the amplitude to correspond to a field theory diagram with two cubic vertices and a propagator of the form,
\begin{eqnarray}
\frac{t^{2n}}{s-M_n^2}.
\end{eqnarray}
The full string amplitude is obtained by summing these tree-level amplitudes over all $n$. When we want to reconstruct this amplitude from a quantum field theory, we should have a closer look at the cubic vertices. Usually\footnote{A complementary strategy consists of considering the non-linear $\sigma$-model corresponding to a string moving in a general background of coherent superpositions of massless string modes. By requiring conformal invariance for the two-dimensional $\sigma$-model one can obtain the equations of motion for the background fields. In a final step one constructs the action whose variation will yield the same equations of motion.} the corresponding low energy field theory of the massless string modes is constructed by reproducing the three-point string amplitudes from cubic vertices in field theory. Let us start by first considering the tachyonic mode as a scalar field $T$. We can have for instance a vertex with three tachyons, which corresponds to an interaction term of the form $\kappa_1 T^3$. At the massless level, the string spectrum contains a graviton (symmetric, traceless tensor) and a dilaton. The vertex with two tachyonic legs and one graviton can be represented by $s^{\mu \nu} \partial_\mu T \partial_\nu T$ in position space. The vertex with two tachyons and one dilaton ($\phi$) has a similar form as the one with three tachyons, namely $\kappa_2\, \phi\, T^2$. Taking all these terms together with the kinetic part of the tachyon, the graviton, the two-form and the dilaton, the action takes the form,
\begin{eqnarray}
{\cal S}& =& {\cal S^{\rm kin}_{\rm grav}} + {\cal S^{\rm kin}_{\rm two-form}} +{\cal S^{\rm kin}_{\rm dilaton}} + \int d^{26}x  \Big(\frac{1}{2}\partial^\mu T \partial_\mu T  + \frac{4}{\alpha'} T^2 + \kappa_1 T^3 \nonumber\\ && \qquad+ \kappa_2\, s^{\mu \nu} \partial_\mu T \partial_\nu T + \kappa_3\, \phi T^2 + \ldots   \Big)
\end{eqnarray}
To obtain a field theory starting from the string scattering amplitudes, we considered a bottom-up approach by looking at possible cubic vertices to reconstruct the scattering amplitudes. We have considered here standard couplings in field theory, but we should just as well consider less obvious couplings like a momentum-dependent 3-tachyon vertex $ T\, \partial^\mu T\, \partial_\mu T $ or couplings to the anti-symmetric tensor. Furthermore, as the string spectrum also contains massive states, we should also take these states into account and write down couplings to these states. All possible omitted terms are represented in the action by ``$\ldots$''. The different coupling constants $\kappa_1$, $\kappa_2$, $\kappa_3$, ... can be related to the string coupling constant $g_s$ and the $\alpha '$ parameter (and possibly other factor) by matching the field theory amplitudes with the string theory amplitudes.  

We leave the tachyon for now and shift our attention to the massless states in the string spectrum. Thus far we have only considered amplitudes without external legs (vacuum amplitudes) or with tachyons as external legs. Evaluating scattering amplitudes with massless string modes  (such as gravitons, anti-symmetric tensor fields, dilatons, etc) as external legs is an equally crucial step to make contact with field theory. Reconstructing the low energy field theory from a string theory corresponds to taking $\alpha' \rightarrow 0$. In this limit the massive string states become very heavy and we are left with the massless string states. The low energy field theory derived from a string theory thus corresponds to the field theory of the massless coherent string modes. 

For the 26-dimensional closed bosonic string theory the massless string modes correspond to a graviton-like field, an anti-symmetric tensor field and a dilaton. The graviton, anti-symmetric tensor field (two-form) and dilaton can be represented in terms of vertex operators of the form $\partial X \bar \partial X e^{i k .X}$, where the traceless symmetric part of this expression corresponds to the graviton, the anti-symmetric part to the two-form and the trace part to the dilaton. The amplitudes we obtain using the massless excitations of the string also resemble amplitudes characteristic for a quantum field theory. When we try to recover these amplitudes from a field theory action, we will also need the kinetic terms for the graviton, two-form and dilaton (${\cal S^{\rm kin}_{\rm grav}}$, ${\cal S^{\rm kin}_{\rm two-form}} $, ${\cal S^{\rm kin}_{\rm dilaton}}$). The field theory action from which we can obtain the scattering amplitudes for the massless string states, corresponds to the action of a gravity theory coupled to a two-form and a dilaton,
\begin{eqnarray}
{\cal S} \sim \int d^{26} x \sqrt{- G}\, \BWe^{-2 \phi} \left(R - \frac{1}{12} H_{\mu \nu \lambda} H^{\mu \nu \lambda} + 4 \partial^\mu \phi\, \partial_\mu \phi   \right).
\end{eqnarray}
In this action $R$ represents the Ricci-scalar derived from the metric $G_{\mu \nu}$, and $G$ represents $\det(G_{\mu \nu})$. We introduced the three-form $H_{\mu \nu \lambda}$, which can be deduced from the two-form $B_{\mu \nu}$,
\begin{eqnarray}
H_{\mu \nu \lambda} = \partial_\mu B_{\nu \lambda} + \partial_\nu B_{\lambda \mu} + \partial_\lambda B_{\mu \nu}.
\end{eqnarray}
This action is written down in the string frame. It means that one should perform a field redefinition of the metric to obtain the Einstein metric and the usual form of Einstein-Hilbert action. 
The string amplitudes for e.g. the gravitons can indeed be obtained from this action by expanding the metric as $G_{\mu \nu} = \eta_{\mu \nu} + 2 \kappa e_{\mu \nu} \BWe^{i k.x}$, where $e_{\mu \nu}$ represents the polarization of the graviton.
\\*

We can also have a quick look at superstring theory and its relation to quantum field theory. Type II superstrings are constructed by introducing fermionic degrees of freedom ($\psi^\mu$) in both decoupled left- and right-moving sectors. One can still choose different periodicity conditions in left- and right-moving sector for the fermions, yielding two distinct theories, namely the non-chiral Type IIA (fermions in left-and right-moving sector have a different chirality) and Type IIB (fermions in left-and right-moving sector have the same chirality). The bosonic and fermionic degrees of freedom are related to each other through supersymmetry, which should be made local when the worldsheet metric is no longer flat. Type II superstrings can therefore be described by an ${\cal N}$ = 1 two-dimensional supergravity, consisting of a matter multiplet ($X^\mu$ and $\psi^\mu$) and a gravity multiplet (zweibein $e_m^a$ and a spin 3/2 gravitino field $\chi_m$). There is however enough local symmetry to gauge away the gravity multiplet and to obtain a simplified action in terms of the matter multiplet. The equations of motion for the zweibein and the gravitino lead to additional constraints, which eliminate even more degrees of freedom. An appropriate treatment of these constraints upon quantization yields the Hilbert space for the closed superstring, which consists of products of independently left- and right-moving string excitations in a ten-dimensional Minkowski spacetime. A consistent quantization also requires the GSO projection \cite{Gliozzi:1976qd} which eliminates certain states (e.g. tachyon) from the Hilbert space and truncates the Hilbert space to a consistent physical spectrum. 

Looking at the massless level of Type II superstrings, we can distinguish two different bosonic sectors: the so-called NS-NS and RR sectors. Both for IIA and IIB the NS-NS sector consists of a graviton-like field, an anti-symmetric tensorfield (two-form) and a dilaton. For both Type II superstrings the RR sector consists of $p$-forms, but the type of p-forms appearing differs. For Type IIA we encounter a one-form and a three-form, for Type IIB a zero-form, a two-form (different from the one in the NS-NS sector) and a four-form. The form of the bosonic action can be split into three parts,
\begin{eqnarray}
{\cal S}_{\rm bosonic} = {\cal S}_{\rm NS} + {\cal S}_{\rm R} + {\cal S}_{\rm CS},
\end{eqnarray}
where ${\cal S}_{\rm NS}$ describes the coupling between the graviton, the NS-NS two-form and the dilaton,
\begin{eqnarray}
{\cal S}_{\rm NS}= \frac{1}{2 \kappa^2} \int d^{10}x \sqrt{-G} \BWe^{-2 \phi} \left( R -\frac{1}{12} H_{\mu \nu \lambda} H^{\mu \nu \lambda}  + 4 \partial^\mu \phi \partial_\mu \phi  \right).
\end{eqnarray}
$ {\cal S}_{\rm R}$ describes the coupling of the RR-forms to the graviton and ${\cal S}_{\rm CS}$ represents the Chern-Simon action. These last two terms depend on the type of superstring theory. $\kappa$ represents the ten-dimensional gravitational coupling constant. We shall not go into further detail. More information about type II Supergravity can be found in \cite{Polchinski:1998rr} and references therein.

The Heterotic (super)string \cite{Gross:1985fr} is constructed in a different manner. The heterotic string can be seen as a hybrid string theory consisting of one chiral half of  the Type II string and one half of closed bosonic string theory compactified on a 16-dimensional torus. The right-moving sector of the Heterotic string consists of bosonic degrees of freedom $X_R^\mu (\tau - \sigma)$ and fermionic degrees of freedom $\psi^\mu (\tau - \sigma)$. The left-moving sector then consists of bosonic degrees of freedom which are divided into two sets: target space coordinates $X_L^\mu (\tau + \sigma)$ and internal coordinates $X_L^I (\tau + \sigma)$. $X^\mu_R$ and $X_L^\mu$ combine to the target space coordinates $X^\mu$, where $\mu \in \{ 0, \dots, 9\}$, such that it describes the position of the string in the ten-dimensional Minkowski spacetime. The right-moving bosonic and fermionic degrees of freedom are related to each other by supersymmetry. It is possible to consider such a construction as the left-moving and right-moving modes decouple for a closed string theory. The only constraints, which relates the right- and left-moving degrees of freedom, are found in the mass-relation of the spectrum and the origin-relation of the closed string. The first relation describes the mass of the states in terms of the excitation modes of the string. The latter relation expresses that it is impossible to choose a specific origin on a closed string, as all points on the closed string are indistinguishable. It is this relation which prohibits the presence of a tachyon and therefore also the Heterotic string is tachyon-free.  

The additional left-moving 16 dimensions $X_L^I$, $I \in \{1, \ldots, 16 \}$, are considered to parameterize an internal compact space. This implies that the momenta are quantized and form a tower of Kaluza-Klein momenta ($m_I/R_I$). The compact space itself is seen as a sixteen-dimensional torus, represented as $\BWIR^{16}/\Gamma$, around which strings can also be winded. The winding of the strings results in a quantized winding number ($n_I R_I$). The quantized number $m_I$ and $n_I$ allow to construct more massless states with respect to the uncompactified version, when all the radii $R_I$ take the special value $R_I = R =\sqrt{\alpha '}$. In that case we can use some of the oscillators in $X^I_L$,
\begin{eqnarray}
X^I_L (\tau + \sigma) = x^I + p^I (\tau + \sigma) + {\rm oscillators},
\end{eqnarray}
as well to construct massless vector bosons in the adjoint representation of a Lie-group $G$. The quantized momenta $p^I$ take value on a lattice $\Gamma$, which coincides with the root lattice of the Lie-group. The lattice itself is not arbitrary, but specified by constraints. The first constraint arises as an allowed momentum should coincide with an allowed winding number, which defines the dual lattice $\tilde \Gamma$. A second constraint arises due to the origin-relation we discussed above, and reads that the allowed momenta should be even integers. In summary, the lattice is a sixteen-dimensional even self-dual lattice. In sixteen dimensions there exist only two lattices, which are even and self-dual, namely $\Gamma_8 \times \Gamma_8$ and $\Gamma_{16}$. The lattice $\Gamma_8$ corresponds to the root lattice of $E_8$ and thus we obtain a first Heterotic string theory for which the massless vector bosons are in the adjoint representation of $E_8\times E_8$. This Heterotic string theory is called Type HE. The other lattice, $\Gamma_{16}$ represents the root lattice of $Spin(32)/\BWIZ_2$. Hence, we find a second Heterotic string theory, called Type HO, which is linked to the group $Spin(32)/\BWIZ_2$. A complementary argument \cite{Gross:1985rr} for these two groups can be found by studying the vacuum partition function for the Heterotic string. Imposing modular invariance of the vacuum partition function constrains the 16-dimensional lattice to be even and self-dual. We thus find again the groups $E_8\times E_8$ and $Spin(32)/\BWIZ_2$.

From the spacetime perspective the Heterotic string is a ten-dimensional Lorentz-invariant, $N=1$ supersymmetric theory. The spectrum of the heterotic string does not contain tachyons. Thus the first states in the spectrum are the massless string modes, which consist of a ten-dimensional $N=1$ supergravity multiplet and ten-dimensional $N=1$ super-Yang-Mills gauge bosons with fermionic superpartners. The supergravity action to start with is given by the Chapline-Manton action, where ten-dimensional $N=1$ supergravity is coupled to $N=1$ Yang-Mills and an anti-symmetric tensor field strength $H_{\mu \nu \lambda}$.  However to reproduce \cite{Gross:1985rr} all trilinear couplings we should also include gauge and Lorentz Chern-Simons terms in the field strength $H_{\mu \nu \lambda}$. The full bosonic action \cite{Gross:1985rr} is therefore given by,
\begin{eqnarray}
{\cal S}_{\rm bosonic}&=& \frac{1}{2 \kappa^2} \int d^{10}x \sqrt{-G} \BWe^{-2 \phi} \left(R - \frac{1}{12} H_{\mu \nu \lambda} H^{\mu \nu \lambda}  + 4 \partial^\mu \phi \partial_\mu \phi \right.\\
&& \left. -\frac{1}{4}Tr(F_{\mu \nu} F^{\mu \nu})  \right),
\end{eqnarray}
with,
\begin{eqnarray}
H_{\mu \nu \rho} &=&  \Big[ \partial_\mu B_{\nu \rho}  + \frac{1}{4} tr \left( A_\mu F_{\nu \rho} - \frac{1}{3} A_\mu A_\nu A_\rho  \right) \nonumber\\
&& - \frac{1}{4} tr\left( \omega_\mu R_{\nu \rho} - \frac{1}{3} \omega_\mu \omega_\nu \omega_\rho \right) + {\rm cyclic\, permutations}\Big].
\end{eqnarray}
Anomaly cancellation requires the additional terms in the $H$-field and constrains the gauge field $A_\mu$ and spin connection $\omega_\mu$ to be representations of the group $SO(32)$ or $E_8 \times E_8$. $F_{\mu \nu}$ is the field strength derived from the gauge field $A_\mu$, while $R_{\mu \nu}$ is the field strength derived from the spin connection $\omega_\mu$.



\section{Divergences in String Theory}
We have indicated in the previous section how scattering strings resemble and differ from their point particle equivalent in field theory. Most field theories should be seen as effective field theories, describing interactions between low energy modes, and ought to be altered at higher energies\footnote{It is unclear to us whether field theory remains a sensible formalism up to the Planck scale. As the importance of gravity increases at high energies we imagine that the usual spacetime description in terms of a metric becomes invalid. Adding fields and interactions would not help at that point and a new formalism to describe (quantum) gravity is required.} (by introducing additional interactions and possibly additional fields). In string theory however, we are led to believe that we integrate over  all possible energy modes of the string. Adding additional interactions to the string action breaks local symmetries on the worldsheet and adding additional fields changes the string theory drastically. Due to the huge amount of symmetries there is however a good chance that string scattering amplitudes yield finite and sensible results without having to resort to extra fields or interactions. 
Therefore, we would like to take a brief moment to discuss possible divergences arising in the string amplitudes. Following the literature we  spend most of our time to the vacuum partition function, as these form a good indicator when infinities might appear or vanish. There are two possible ways to discuss the finiteness of a string theory. A first method consists of calculating the string amplitudes order by order in the genus-expansion and checking where possible divergences might appear. The second method tries to encapsulate the general form of the scattering amplitude at any order and to discuss the vanishing of possible divergences using arguments like modular invariance and holomorphic properties of the amplitudes.

Let us first see how these methods can teach us something more about the bosonic string theory in 26 spacetime dimensions. We have seen that at zeroth order ( the Riemann sphere) the vacuum partition function vanishes. At first order (the torus) the vacuum 
partition function reduces to an integral over the complex upper  half plane. But due to modular invariance this integral can be even more reduced to the integral over a fundamental domain, thus evading the dangerous point $Im (\tau) \rightarrow 0$. However we are not yet out of the danger zones. Another perilous point is $Im (\tau ) \rightarrow \infty$, which corresponds to an infinitely long torus. In that case the Dedekind $\eta$-function, which was introduced in chapter \ref{BWc:Amplitudes}, behaves as,
\begin{eqnarray}
| \eta (\tau)  |^{-48} \sim \BWe^{4 \pi Im (\tau)} + \ldots\, .
\end{eqnarray}
This divergence arises due to the presence of the tachyon\footnote{We can compare this exponential power to its field theory equivalent  $e^{- m^2 l/2}$, where $l$ is the modulus of the loop in which the particle with mass $m$ is running. It is clear that the negative mass of the tachyon renders the divergent term.}. We notice that the asymptotic  behavior is controlled by the lightest string states and that there are no regions in the moduli space of the torus which lead to UV divergences (i.e. the very massive string states do not cause divergent terms). We expect that the tachyon divergence\footnote{As always a tachyon indicates the expansion around a wrong vacuum background. Minkowski spacetime is not a solution to the bosonic string equations of motion.} disappears in superstring theories and that this divergence is therefore an artifact of bosonic string theory. Even without the divergence due to the tachyon, it is already remarkable that the cosmological constant does not vanish at first order.  The question then rises whether there exist other types of amplitudes which vanish at tree-level, but have a non-zero value at first order. Our attention is immediately drawn to the dilaton tadpole diagram given by,
\begin{eqnarray}
\langle 0 | \phi  | 0 \rangle \sim \sum_{{\rm g}=0}^{\infty} \int  \frac{\BWD X \BWD g}{\cal N}\, \BWe^ {-S_X -\lambda \chi}  \int d^2 \sigma \sqrt{g}\,   (\partial X )^2.
\end{eqnarray}
At zeroth order (sphere) this amplitude vanishes as it can be seen as a one-point function on the sphere. The conformal Killing group is infinite and one inserted vertex operator does not fix this finite group completely. At first order (torus) the dilaton tadpole diagram does not vanish. This process then corresponds to an on-mass-shell scalar particle being absorbed by the vacuum. The fact that $\langle 0 | \phi | 0 \rangle \neq 0$ indicates a vacuum instability\footnote{It was pointed out \cite{Fischler:1986ci, Fischler:1986tb} that the non-vanishing dilaton tadpoles implies an expansion around the wrong background.  Instead one should use a de Sitter-vacuum as exact string background, which agrees with the fact that the cosmological constant does not vanish in bosonic string theory. One can then show that when the value of the cosmological constant is equal to the non-vanishing value of the dilaton tadpole, de Sitter spacetime is a solution of the string equations of motion up to first order in $\alpha '$.} for the bosonic string theory.

 The second order for the vacuum amplitude and the dilaton tadpole was investigated explicitly in e.g.  \cite{Belavin:1986tv, Moore:1986rh, Kato:1986wj}. It was realized that two types of divergences can arise: a divergence due to the presence of tachyons and another divergence due to the presence of a dilaton. These divergences are mathematically related to the integration over regions of the moduli space where the Riemann surfaces degenerate. The degeneration of Riemann surfaces is geometrically interpreted as non-trivial closed loops that are pinched down to zero size. Physically we obtain a non-vanishing cosmological constant and a non-vanishing dilaton tadpole. 
 
The moduli spaces for Riemann surfaces with genus ${\rm g} \geq 3$ are far less understood. Hence, higher loop calculations become less straightforward and we are forced to employ the second method. This was done in e.g. \cite{Gava:1986ei}. There it was realized that also for arbitrary genus Riemann surfaces the divergences arise due to the presence of the dilaton tadpole (one-point function of the dilaton) and due to the presence of the tachyon. It was pointed out that from a mathematical point of view the divergences arise when we consider particular boundary regions of the moduli space. These boundaries of the moduli space correspond to the shrinking of one of the moduli and thus to the shrinking of one of the non-trivial closed loops. From the second method we can conclude as well that the cosmological constant and the dilaton tadpole do not vanish.    
\\*

The tachyon divergence from bosonic string theory is not supposed to appear in superstring theory, as the tachyon is projected out. Therefore, we can pay all our attention to the vacuum amplitude and the dilaton tadpole amplitude and ask ouselves whether they vanish or not. If they both vanish (to all orders), we can conclude that Minkowski spacetime is an appropriate and exact background for superstring theory. Although calculating superstring amplitudes is a highly non-trivial task, it did not prevent the Verlinde twins to come up with a formula that allows to investigate the divergences in superstring theory. In \cite{Verlinde:1987sd} they use the BRST-formalism to write down the superstring multiloop amplitudes in terms of theta functions. Their analysis reveals that the partition function can be expressed as a total derivative on the moduli space. This immediately implies that possible divergences in superstring theory are due to contributions from the boundary of the moduli space. Similar to the bosonic string, the boundary of the moduli space corresponds to those Riemann surfaces with genus $\rm g$ that break up into two Riemann surfaces with lower genus, or Riemann surfaces with genus ${\rm g} -1$ to which an infinitely long handle is attached (dilaton tadpole). For Type II Superstring theory this boundary term was shown to vanish for genus two surfaces in \cite{Atick:1987rk}. The same authors showed in \cite{Atick:1987zt} that the boundary term vanishes for heterotic theories order by order using an inductive reasoning. For superstrings one can relate the g-loop vacuum amplitude to the one-loop vacuum amplitude. Since the one-loop vacuum amplitude vanishes so does the g-loop vacuum amplitude.

As superstring amplitudes are harder to calculate, the second method to deal with divergences is often more illuminating as it evades difficulties with calculations  and emphasizes the general structure and properties of the amplitudes. In \cite{Martinec:1986wa} general considerations about the supersymmetric properties of the superstring theories lead to the conclusion that the vacuum amplitude vanishes and that the dilaton tadpole vanishes as well. 
Furthermore it is pointed out that there are no UV divergences, if the amplitudes are modular invariant. 
Another approach is taken in \cite{Gliozzi:1987aw}. If one assumes that only the transverse degrees of freedom are dynamically relevant (i.e. one starts from a light-cone quantization perspective) and the multiloop amplitudes are modular invariant, one can construct relations between the theta-functions in which the multiloop-amplitudes are expressed. These relations rely on representation properties of the modular group and yield a vanishing vacuum amplitude, if the string spectrum contains spacetime fermions. A more analytic approach to discuss the general form of the scattering amplitudes (both bosonic as supersymmetric) is used in \cite{Mandelstam:1991tw}. Also finiteness of the superstring is discussed in terms of the vanishing of the dilaton tadpole amplitude as a combined consequence of supersymmetry and dynamics. 
\\*

This brief overview tells us that bosonic string theory in a 26-dimensional Minkowski spacetime does have divergences due to the presence of a tachyon and seems to be inconsistent due to a non-vanishing cosmological constant and a non-vanishing dilaton tadpole. We are led to believe that Minkowski spacetime is not a solution to the equations of motion of the closed bosonic string and that it can not serve as a consistent background for the closed bosonic string. However, we do notice that modular invariance implies that UV divergences due to massive string states do not appear. In superstring theory there is no tachyon present, therefore this divergence is not present. The cosmological constant and the dilaton tadpole in superstring theories vanish by virtue of the supersymmetric properties of the two-dimensional super-conformal field theory. 
\chapter*{Acknowledgments}

We would like to thank Raphael Benichou, Ben Craps, Alexei Koshelev, Amitabh Virmani and Walter Troost for stimulating discussions and the Modave organizers for the opportunity to give these lectures and for use of the Baratin. W.S.\ is and B.V.\ was, for the largest part of this work, supported by an aspirant fellowship from FWO-Vlaanderen. The work of W.S. was supported in part by the Belgian Federal Science Policy Office through the Interuniversity Attraction Pole P6/11, and in part by the FWO-Vlaanderen through project G.0428.06. The work of B.V. was supported by the Federal Office for Scientic, Technical and Cultural Affairs through the Interuniversity Attraction Poles Programme Belgian Science Policy P6/11-P and by the project G.0235.05 of the FWO-Vlaanderen and by DSM CEA-Saclay, by the ANR grant 08-JCJC-0001-0, and by the ERC Starting Independent Researcher Grant 240210 - String-QCD-BH.

\appendix
\renewcommand{\theequation}{\Alph{chapter}.\arabic{equation}}
\renewcommand{\thechapter}{\Alph{chapter}}
\pagestyle{plain}
\chapter{The path integral measure\label{BWapp:Gaussians}}

\emph{We have not properly defined the integration measures appearing in the path integral. The problem lies in the functional integration, which comes down to an infinite limit of a number of ordinary integrations. In fact, the measure could be traced back from the careful construction of the path integral. But there is an easier way. We determine the normalization of the measure in an indirect manner, and fix it completely by demanding functional integration obeys the same integration rules as ordinary integration.}\\[12pt]

Therefore, consider the ordinary Gaussian integral
\begin{equation}
 \int \BWd x \,\BWe^{-\frac12  x^2 } = \sqrt{{\pi}}\,.\label{BWeq:xintegral}
\end{equation}
We can indirectly define the measure by fixing the value of this integral. When integrating using the ordinary Lebesgue measure, we would obtain as a result $\sqrt{\pi}$. However, for convenience we can pick the value of the integral to be 1, thus fixing the normalization of the measure in an indirect way:
\begin{equation}
\int \BWd x \,\BWe^{-\frac12  x^2} \equiv 1\,.
\end{equation}
In short, a norm on $\mathbb{R}$, in this case $||x||^2 = x^2$, provides us with a natural measure.

For a functional integral, we can take an analogous path. The measures we are considering are  $\BWD g$ and $\BWD X$. We also need to consider measures on $\BWgauge$. In the following, we write infinitesimal generators of these groups in terms of vector fields (Diff: $\sigma^a \to \sigma^a + v^a$) and a scalar function (Weyl: $g_{ab} \to g_{ab}+\phi(\sigma) g_{ab}$). And finally, from the discussion of the gauge slice, we know that we will end up with an integral over the moduli. The latter are real functions, so for these we can use the Lebesgue measure. The natural generalization of the above Gaussian integral to symmetric two-tensor $h_{ab}$, spacetime vectors $X^\mu$, worldsheet vector fields $v^a$ and a scalar $\phi$ is given by:
\begin{align}
1 &= \int \BWD h \,\BWe^{-\tfrac12||h||^2}\nonumber\\
1 &= \int \BWD X \,\BWe^{-\tfrac12||X||^2}\nonumber\\
1 &= \int \BWD v \,\BWe^{-\tfrac12||v||^2}\nonumber\\
1 &= \int \BWD \phi \,\BWe^{-\tfrac12||\phi||^2} \label{eq:Path_Gaussian}
\end{align}
As before, the integration measure is determined indirectly from some notion of length on the integration space. The (quite natural) norms we consider are written in terms of integrals over the worldsheet as:
\begin{align}
  ||h||^2&=\int \BWd^2\sigma \sqrt{g} g^{ac}g^{bd} h_{ab}h_{cd}\,,\nonumber\\
  ||X||^2&=\int \BWd^2\sigma \sqrt{g}X^\mu X_\mu\,,\nonumber\\
  ||v||^2&=\int \BWd^2\sigma \sqrt{g}g_{ab}v^a v^b\,,\nonumber\\
  ||\phi||^2&=\int \BWd^2\sigma \sqrt{g}\phi^2\,.\label{eq:Path_Scalar product}
\end{align}
These notions of length also define a natural inner product, which we denote $\langle\cdot,\cdot\rangle$. 
For path integration, we then take on the following three conditions:
\begin{itemize}
\item The rules of path integration (linearity, change of variables\ldots) are the same as for ordinary differentiation
\item The measures obey the condition \eqref{eq:Path_Gaussian} for Gaussian integration
\item The normalizations of the measures are fixed by eq. \eqref{eq:Path_Scalar product}
\end{itemize}

These three conditions define and fix the measure completely and make it possible to define path integration for Gaussian integrands \cite{D'Hoker:1985pj,D'Hoker:1988ta,Polchinski:1985zf}. We then have the result for matrices in the exponential:
\begin{align}
(\det M)^{-1/2} &= \int \BWD X \BWe^{-\tfrac12\langle X, M X\rangle}
\end{align}
and likewise for the other fields. As an example, we have the following result for Gaussian integrals of $X^\mu$:
\begin{align}
  \int \BWD X &\exp \left(\int \BWd^2 \sigma \,A X^\mu(\sigma)X_\mu(\sigma) + \int \BWd^2 \sigma J_\mu(\sigma) X^\mu(\sigma)\right)\nonumber\\
&=\prod_\mu \int \BWD X^\mu \exp \left(\int \BWd^2 \sigma \,A (X^\mu(\sigma)+J^\mu)(X_\mu(\sigma)+J) -\int \BWd^2 \sigma J_\mu(\sigma) X^\mu(\sigma)\right)\nonumber\\
&=\left(\frac1A\right)^{D/2} \exp\left(-A^{-1} \int \BWd^2 \sigma J^\mu J_\mu\right)
\end{align}
where $D$ is the dimensionality of spacetime (the number of $X^\mu$ components), we completed the squares in the second line and performed a Gaussian integration in the last line.

Finally, note that the scalar products and norms \eqref{eq:Path_Scalar product} are not all Weyl and Diff invariant. In particular,  the norm of Weyl transformations is not invariant under Weyl transformations, as $\sqrt{g}$ not Weyl invariant. As a consequence, the path integral measure $\BWD\phi$ implicitly defined through the above norm $||\phi||^2$ cannot be Weyl invariant either and the Polyakov path integral has a Weyl anomaly -- the Weyl symmetry is not pertained at the quantum level, see section \ref{ss:RemarksMeasure}. As discussed there, only when $D=26$, the anomaly disappears and the symmetries of the classical theory carry over to the quantum string theory.

\chapter{Some details on the Faddeev-Popov procedure}\label{BWapp:FaddeevPopov}

\emph{
The main idea in using the Faddeev-Popov trick, is to represent the so-called Faddeev-Popov determinant (the analogue of the Jacobian we discussed before) as a path integral over ghost fields. We first give some insight on the appearance of ghost fields with the wrong spin-statistics, and then show how the Faddeev-Popov can be applied in detail. 
}\\[12pt]

Consider the complex integral
\begin{equation}
 \int \BWd z \BWd \bar z \BWe^{-\frac12 a z \bar z } = \frac{\pi}{a} \sim \frac 1 a\,.\label{BWeq:zintegral}
\end{equation}
as a toy example. The result ($\frac 1 a$ up to a numerical prefactor) can be inverted with a simple trick. Write down the same integral, but now over Grassmannian variables $\theta, \bar \theta$. Using the rules of Grassmannian integration (i.e.\ it acts as a differentiation), we get:
\begin{equation}
 \int \BWd \theta \BWd \bar \theta \BWe^{-\frac12 a \theta \bar \theta } = a/2 \sim a\,,
\end{equation}
we get exactly the inverse of the bosonic integral \eqref{BWeq:zintegral}, up to a prefactor we can always absorb in the measure. This result generalizes to matrices (you can see this by going to a basis of eigenvectors):
\begin{align}
  \int \prod_{i=1}^N \BWd z_i \BWd \bar z_i \BWe^{-\frac12 z_i A_{ij} \bar z_j } &\sim (\det A)^{-1} \label{BWeq:app-FP_Gaussian_integral} \\
\int \prod_{i=1}^N \BWd \theta_i \BWd \bar \theta_i \BWe^{-\frac12  \theta_i A_{ij}\bar \theta_j } &= \det A
\end{align}
and we can play the same game for operators in functional differentiation, by expanding the fields in the functional integral in terms of eigenmodes of the operators in the exponentials. We conclude that in principle, we should be able to represent the Jacobian determinant we put forward in chapter \ref{BWc:Fixing} by an integral over Grassmann numbers. Because the integration variables are associated to the generators of bosonic symmetries (i.e.\ diffeomorphisms on the world sheet and conformal Killing vectors), they obey the wrong spin-statistics: we will have anticommuting vectors and tensors. The ghost fields for the bosonic string are the so-called $b$ and $c$ ghosts. The field $b$ is a symmetric, traceless tensor, related to the gauge fixing that gets rid of the path integral over the metric $g_{ab}$, while the $c$ ghosts are related to the gauge fixing that fixes a certain number of vertex operator positions.

The above trick is applied to the Faddeev-Popov procedure by means of a delta-functional. Consider the following property of the Dirac delta function over the real numbers:
\begin{equation}
 \delta(x-x_0)=|f'(x_0)|\delta(f(x)) \,,
\end{equation}
where $x_0$  is a zero of the function $f$ (i.e.\ $f(x_0)=0$) and a prime denotes the derivative. When there are multiple zeroes, one should consider a sum over all zeroes of the function $x_0$. We see that the proportionality factor between these two functions is given by the first order derivative.  Taking the integral on both sides give $1 = |f'(x_0)| \int  \BWd x\delta(f(x)) $. The result immediately applies to vector  functions of $\mathbb{R}^n \to \mathbb{R}$. Taking the integral on both sides would then yield:
\begin{equation}
 1 = \int  \left(\prod_i \BWd x^i \right) J \delta(f(\vec x))\,,
\end{equation}
where $J$ is the Jacobian of the transformation $f(x)$:
\begin{equation}
 J = \left|\frac{\partial f^i(\vec x)}{\partial x^j}\right|\,.
\end{equation}
By analogy, one extends this result to functional integration. Say we have a field configuration $\phi$ (possibly multi-dimensional) and some gauge transformations which we denote by $\Omega$. Then the above results generalizes to:
\begin{equation}
 1 = \Delta_{FP}(\phi_0)\int \BWd t\BWD \Omega \delta(F(\phi^\Omega(t)))\,,\label{BWeq:FP}
\end{equation}
where the integral is taken over the gauge group. The generalization of the Jacobian determinant $J$, denoted $\Delta_{FP}$, is the Faddeev-Popov determinant one would like to pin down. The function $F$ is a gauge-fixing constraint. The gauge-fixing value for the fields $\phi^\Omega(t)$. It depends on the gauge choice $\Omega$ and possibly on moduli $t$ that can vary along the gauge slice. The Faddeev-Popov trick is now simply to insert the function \eqref{BWeq:FP} into the path integral. Since it equals one, we are not changing the result. But inserting this function allows one to easily decouple the integrals over the gauge group and the moduli space, while the integrand is automatically restricted to lie on the gauge slice through the delta functional. So far so good: we have an easier way of restricting the integral to the gauge slice and to decouple the integral over the gauge group. But how to calculate the ``Jacobian'' of this transformation, the Faddeev-Popov determinant $\Delta_{FP}$? This is done by representing the delta functional as a path integral. One can rewrite \eqref{BWeq:FP} as 
\begin{equation}
  \Delta_{FP}^{-1}(\phi_0) = \int \BWd t\BWD \Omega \delta(F(\phi^{\Omega}(t))\,.
\end{equation}
Using standard path integral methods, one can rewrite the delta functional on the right-hand side as a (Gaussian) path integral over bosonic fields (commuting numbers). To invert the determinant, one simply writes down the same path integral, but now over anticommuting fields. In general, the fields in the path integral correspond to the gauge symmetries one considers in the gauge-fixing procedure.

In conclusion, we see that the Faddeev-Popov procedure allows one to replace the Jacobian determinant we had before in terms of a path integral over ghost fields. The restriction of the regularized path integral to a chosen gauge slice is obtained by means of the function $F(\phi^{\Omega}(t))$. The procedure simplifies the calculation and allows for a convenient characterization of the superstring in particular. Also in the bosonic string, the Faddeev-Popov procedure can shed light on the structure of the theory.  We do not consider this method further, as we are only interested in aspects of the bosonic string that do not need the introduction of ghost fields.

\cleardoublepage
\pagestyle{plain}
\def\href#1#2{#2}
\bibliographystyle{BibliographyStyle}
\addcontentsline{toc}{chapter}{\sffamily\bfseries Bibliography}

\bibliography{StringAmplitudes}

\end{document}